\documentclass[11pt,reqno,UTF8,twoside]{amsart}
\usepackage{mathrsfs}
\usepackage{amsfonts,amssymb,amsmath,amsthm}
\usepackage{cite}
\usepackage[colorlinks=true,citecolor=cyan]{hyperref}
\usepackage{titletoc}
\usepackage{geometry}
\usepackage{bm}
\usepackage{indentfirst}
\usepackage{graphicx}
\usepackage{float} 
\usepackage{booktabs}
\usepackage{longtable}
\usepackage{subfigure}
\usepackage{stmaryrd}
\usepackage{enumerate}
\usepackage{tikz}
\usepackage{ulem}
\usepackage{color,soul}
\usepackage{setspace}
\usepackage{stmaryrd}
\usepackage[pagewise]{lineno} 
\allowdisplaybreaks[2]
\usepackage{appendix}
\usepackage{cases}
\usepackage{todonotes}
\usepackage{enumitem}
\hypersetup{linkcolor=magenta}

\makeatletter
\@namedef{subjclassname@2020}{\textup{2020} Mathematics Subject 
	Classification}
\makeatother
\usepackage{amssymb}
\usepackage[dvips]{epsfig}
\usepackage{graphicx}
\usepackage{color}
\usepackage{ulem}
\usepackage{tikz}
\usepackage{amsmath}
\usepackage{amssymb}
\usepackage{amsfonts}
\usepackage{amsthm}
\usepackage{tikz}
\usepackage{bm}
\usetikzlibrary{calc}
\usepackage{changes}

\newcommand{\R}{\mathbb R}

\newcommand{\Z}{\mathbb Z}

\newcommand{\C}{\mathbb C}

\usepackage{float}

\newcommand{\N}{\mathbb{N}}

\newcommand{\T}{\mathbb{T}}

\newtheorem{thm}{Theorem}[section]

\newtheorem{lem}[thm]{Lemma}

\newtheorem{cor}[thm]{\bf Corollary}

\newtheorem{rem}[thm]{\bf Remark}

\theoremstyle{definition}
\newtheorem{defn}[thm]{Definition}

\theoremstyle{statement}

\numberwithin{equation}{section}
\usepackage{bm}

\usepackage{caption}

\usepackage{tikz}
\usetikzlibrary{patterns, decorations.pathreplacing, arrows.meta}
\begin{document}
	
\title[]{Anderson localized states for the quasi-periodic nonlinear Schr\"odinger  equation on $\mathbb Z^d$}

\author[]{Yunfeng Shi}
\address[Y.S.]{School of Mathematics,
Sichuan University,
Chengdu 610064,
China}
\email{yunfengshi@scu.edu.cn}
\author[]{W.-M. Wang}
\address[W.W.] {CNRS and D\'epartment De Math\'ematique ,
Cergy Paris Universit\'e,
Cergy-Pontoise Cedex 95302,
France}
\email{wei-min.wang@math.cnrs.fr}

\date{\today}

\keywords{Anderson localization, discrete nonlinear Schr\"odinger equation, quasi-periodic potential, Diophantine estimates, semi-algebraic sets, Lyapunov-Schmidt decomposition, Green's function estimates, Newton iteration}

\begin{abstract} We establish large sets of Anderson localized states for the quasi-periodic nonlinear Schr\"odinger equation on $\mathbb Z^d$, thus extending  Anderson localization from the linear  (cf. Bourgain [Geom. Funct. Anal., 17(3):682--706,  2007])  to a nonlinear  setting,  and  the random (cf. Bourgain-Wang [J. Eur. Math. Soc., 10(1):1--45, 2008]) to a deterministic  setting. Among the main ingredients are a new Diophantine estimate of quasi-periodic functions in arbitrarily  dimensional phase space,  and the application of Bourgain's geometric lemma  in  [Geom. Funct. Anal.,  17(3):682--706,  2007]. 

\end{abstract}

\maketitle
	\section{Introduction and Main Result}
	Let
	\begin{align*}
\bm\alpha&=(\alpha_1,\cdots,\alpha_d)\in [0,1]^d,\\
\bm \theta&=(\theta_1,\cdots, \theta_d)\in [0,1]^d,
\end{align*}
%$$\bm n\bm \alpha=(n_1\alpha_1,\cdots, n_d\alpha_d)\in \R^d,$$
and consider the nonlinear Schr\"odinger  equation with a quasi-periodic potential on $\mathbb Z^d$: 
	\begin{align}\label{NLS}
iu_{t}+(\varepsilon \Delta+V(\bm\theta+\bm n \bm \alpha)\delta_{\bm n, \bm n'})u+\delta |u|^{2p}u=0, %\, \text{on } \mathbb Z^d,
\end{align}
where $\varepsilon$ and $\delta$ are parameters in $[0, 1]$, $\bm n\in\mathbb Z^d$,
$p\in\N$ and $\Delta (\bm n, \bm n')=\delta_{|\bm n-\bm n'|_1, 1}$  denotes the {adjacency  Laplacian}  with $|\bm n|_1:=\sum\limits_{\ell=1}^d|n_\ell|$
and $$\bm n\bm \alpha=(n_1\alpha_1,\cdots, n_d\alpha_d)\in \R^d;$$
%we assume \begin{align*}
%\bm\alpha&=(\alpha_1,\cdots,\alpha_d)\in [0,1]^d,\\
%\bm \theta&=(\theta_1,\cdots, \theta_d)\in [0,1]^d,
%\end{align*}
%and define
%$$\bm n\bm \alpha=(n_1\alpha_1,\cdots, n_d\alpha_d)\in \R^d;$$
the potential $V$ is a trigonometric polynomial: 
	\begin{align*}
	V(\bm \theta)=\sum_{ \bm\ell \in \Gamma_K}v_{\bm\ell}\cos2\pi( \bm\ell  \cdot\bm\theta),\   \bm\ell\cdot \bm \theta=\sum_{s=1}^d\ell_s\theta_s, 
	\end{align*}
	with  $\Gamma_K\subset [-K, K]^d\setminus\{\bm 0\},\ K\geq 1$  
	%has the property that  
	%\begin{equation}\label{(1.2)}
	%\begin{aligned}
	%&\text {if } \bm\ell\in\Gamma_K, {\rm then},  \ell_k\neq 0, \forall \ k=1, 2, ..., d, \\
	%&\text {and if }\ \bm\ell, \bm\ell'\in\Gamma_K, \ {\rm then}\ \bm\ell+\bm\ell'\neq \bm 0,
	%\end{aligned}
	%\end{equation}
	%Then  $$\#\Gamma_K=\frac{(2K+1)^d-1}{2}.$$
and 
	$$\bm v=(v_{\bm \ell})_{\bm \ell\in\Gamma_K}\in \R^{\#\Gamma_K}\setminus\{\bm 0\},$$
	where $\#(\cdot)$ denotes the cardinality of a set. We further assume that the set $\Gamma_K$  is {\it maximal} so that: 
	\begin{equation}\label{(1.2)}
	\begin{aligned}
	&\text {if } \bm\ell\in\Gamma_K, {\rm then},  \ell_s\neq 0\ {\rm for}\  \forall \ 1\leq s\leq d; \\
	&\text {and if }\ \bm\ell, \bm\ell'\in\Gamma_K, \ {\rm then}\ \bm\ell+\bm\ell'\neq \bm 0.
	\end{aligned}
	\end{equation}

	We assume that $V$ is non-degenerate:  For any fixed $1\leq \ell\leq  d$ and $(\theta_s)_{s\neq \ell}$,   $f(\theta_\ell):=V((\theta_s)_{s\neq \ell}, \theta_\ell)$
	is not a constant function in $\theta_\ell.$ 
	%Let $\ell\in\{1, 2, \cdots, d\}$ and fix $\theta_i$, $\forall \ i\neq \ell$. Then
	%$V$ is not a constant function in $\theta_\ell$ for all fixed $\theta_i$, $i\neq \ell$, for all $\ell$. 
	Note that the first condition in 
	\eqref{(1.2)} then says that, in fact, each term in the polynomial is non-degenerate.
	%\begin{rem}  
	%We may assume, in addition that 
	%\begin{align}\label{vdef}
	%\min_{\bm\theta \in [0, 1]^d} V(\bm \theta)\geq 1.
	%\end{align}
	%Indeed, if we replace $V$ in \eqref{NLS} with $V+m$ for some constant $m>\|V\|_\infty+1$ and have constructed  its corresponding solution  $u(t, \bm n)$,   then $e^{-imt}u(t, \bm n)$ must be a solution to the original equation \eqref{NLS}. Since our goal is to establish the existence of time quasi-periodic solutions to \eqref{NLS}, such a phase shift is largely immaterial.  The bound \eqref{vdef} is convenient for imposing certain lower bounds in the analysis later.
		
	%\end{rem}

When $\delta=0$, it is known  from the breakthrough paper \cite{Bou07} (cf. also the recent refinement \cite{JLS20}) that the linear operator 
\begin{equation}\label{SO}
\mathcal H=\varepsilon \Delta+V(\bm n\bm \alpha+\bm \theta)\delta_{\bm n, \bm n'}\  {\rm on}\  \mathbb Z^d,
\end{equation}
exhibits Anderson localization, namely pure point spectrum with exponentially decaying eigenfunctions,
 for small $\varepsilon$, on a large set in $(\bm\alpha, \bm \theta)$. 
 %The conditions on $\alpha$ and
 %$\theta_0$ can, moreover, be made explicit, giving arithmetic Anderson localization, see e.g., \cite{GY20, CSZ22, CSZ23}, see also \cite{Jit99, JL18}.
It follows that  when $0<\varepsilon \ll 1$, the linear Schr\"odinger equation 
 \begin{equation}\label{ls}
 iu_{t}+(\varepsilon \Delta+V(\bm n\bm \alpha+\bm \theta)\delta_{\bm n, \bm n'})u=0 
 \end{equation}
 has only Anderson localized states, i.e., roughly speaking, wave packets localized about the origin
 remain localized for all time.  
 
The present paper addresses the persistence question when $\delta\neq 0$, namely the existence of Anderson localized type solutions for the nonlinear Schr\"odinger equation \eqref{NLS}
when $0<\varepsilon, \delta\ll 1$, under appropriate conditions on $\bm\alpha$ and $\bm\theta$. Since both $\varepsilon$ and $\delta$ are small, we start from the (decoupled) equation  
\begin{equation}\label{decoup}
iu_{t}+V(\bm n \bm\alpha+\bm\theta)\delta_{\bm n, \bm n'}u=0.
\end{equation}
It has solutions of the form 
\begin{equation}\label{qp0}
u^{(0)}(t, \bm n)=\sum_{\bm n\in \Z^d} a_{\bm n} e^{i\mu_{\bm n} t},
\end{equation}
where 
$$\mu_{\bm n}={V(\bm n\bm\alpha+\bm \theta)}$$
and $a_{\bm n}$ decays rapidly. The solutions \eqref{qp0}, in general, have infinite number of frequencies and are {\it almost-periodic} in time. 
%  as $|\bm n|\to\infty.$
% $\delta_{\bm n}(\bm x)=1$, if $\bm x=\bm n$ and $0$ otherwise.

We study the persistence of the above type of solutions with finite (but arbitrary) number of frequencies in time, i.e.,  the {\it quasi-periodic} in time solutions.
Denoting the frequency by $\bm \omega$, $\bm \omega= (\omega_1, \omega_2,\cdots, \omega_b)$, as an Ansatz, we seek solutions to \eqref{NLS} in the form of a convergent series:
\begin{align}\label{Ant}
	u(t,\bm n)=\sum_{(\bm k, \bm n)\in\Z^b\times\Z^d}\hat u (\bm k,\bm n) e^{i\bm k\cdot\bm \omega t}.
	\end{align} 
	%or equivalently 
%\begin{align*}
	%u(t,\bm n)=\sum_{\bm k\in\Z^b}q (\bm k,\bm n) e^{i\bm k\cdot\bm \omega t},
	%\end{align*}
%with appropriate conditions on 	$q (\bm k,\bm n)$ as $|\bm k|+|\bm n|\to\infty$. 
Note that solutions 
of the above form are consistent with the nonlinearity in \eqref{NLS}.%, as the cosine series 
Denote by ${\rm meas}(\cdot)$ the Lebesgue measure of a set.  Let $|\cdot|$ denote  the supremum norm.  We have
	%\begin{thm}\label{mthm}

%Fix (any) $\{\bm n_{l}\}_{l=1}^b\subset \Z^d$ and let $\bm a=(a_l)_{l=1}^b\in [1,2]^b.$
%Consider a solution of \eqref{NLS} when $\varepsilon=\delta=0$,
%\begin{align}\label{u0}
%u^{(0)}(t,\bm n)=\sum_{l=1}^ba_l e^{i\omega_l^{(0)}t}\delta_{\bm n_{l}, \bm n},
%\end{align}
%where 
%$$\bm \omega^{(0)}=(\omega_l^{(0)})_{l=1}^b=\left({V(\bm n_{l}\bm\alpha+\bm \theta)}\right)_{l=1}^b\in[-\|V\|_\infty,  \|V\|_\infty]^b.$$
%Then for $0<\varepsilon\lesssim\delta\ll1$,   there is $\mathcal{W}\subset (0,1]^{2d}$ with $
	%{\rm meas}((0,1]^{2d}\setminus\mathcal{W})=o(1)$, such that for $(\bm \alpha, \bm \theta)\in \mathcal{W}$, there is $\mathcal{R}=\mathcal{R}_{\bm\alpha, \bm \theta}\subset [1,2]^b$ with ${\rm meas}([1,2]^b\setminus\mathcal{R})= o(1)$  so that  the following holds true. If $\bm a\in \mathcal{R}$, there is $\bm \omega=\bm \omega(\bm a)$ satisfying $|\bm \omega-\bm \omega^{(0)}|\lesssim \delta$  so that
	%\begin{align*}
	%u(t,\bm n)=\sum_{\bm k\in\Z^b}\hat u (\bm k,\bm n)  e^{i\bm k\cdot\bm\omega t}%e^{ik\cdot\omega t}%=u^{(0)}+O(\sqrt{\varepsilon+\delta})
	%\end{align*}
			%is a solution of \eqref{NLS}. Moreover, we have 
			%\begin{align*}
			%q (k,n)&=q(-k,n)\in\R,\\
			% \hat u (\bm e_l, \bm n_{l})&=a_l, \ 1\leq l\leq b,\\
			%\sum _{(\bm k, \bm n)\notin \mathcal{S}_+} |\hat u (\bm k,\bm n)|e^{\rho(|\bm k|+|\bm n|)}&<\sqrt{\varepsilon+\delta},\  \rho>0,
			%\end{align*}
			%where $\mathcal{S}_+=\{(\bm e_l,\bm n_{l})_{l=1}^b\}$ with $\bm e_l$ the standard basis vectors (i.e., $e_l(j)=\delta_{l,j}, j=1,\cdots,b$) for $\Z^b.$
			%\end{thm}

\begin{thm}\label{mthm}
%Consider the NLS in \eqref {NLS}. 
Fix (any) distinct  $\bm n_1,\cdots,\bm n_b\in \Z^d$ and let 
%$\bm a=(a_l)_{l=1}^b\in [1,2]^b$, and
%Consider a solution of \eqref{NLS} when $\varepsilon=\delta=0$,
\begin{align}\label{u0}
u^{(0)}(t,\bm n)=\sum_{\ell=1}^ba_\ell e^{i\omega_\ell^{(0)}t}\delta_{\bm n_{\ell}, \bm n}
\end{align}
 be a solution to \eqref{decoup},
where $\bm a=(a_\ell)_{\ell =1}^b\in [1,2]^b$ and
$$\bm \omega^{(0)}=(\omega_\ell^{(0)})_{\ell=1}^b=\left({V(\bm n_{\ell}\bm\alpha+\bm \theta)}\right)_{\ell=1}^b. $$%\in[-\|V\|_\infty,  \|V\|_\infty]^b \ {\rm with}\ \|V\|_{\infty}:=\max_{\bm \theta\in[0,1]^d}|V(\bm\theta)|.$$
Then for $0<\varepsilon\leq\delta\leq\log^{-1}\frac1\varepsilon\leq \delta_0(b,d,K,V,\max\limits_{1\leq \ell\leq b}|\bm n_\ell|)\ll1$,  there is a set $\mathcal{W}\subset [0,1]^{2d}$ in $(\bm\alpha,\bm\theta)$, satisfying
	${\rm meas}([0,1]^{2d}\setminus\mathcal{W})\leq \log^{ -c} \frac{1}{\varepsilon+\delta}$ $(c=c(b,d)>0)$  such that,  for any $(\bm \alpha, \bm \theta)\in \mathcal{W}$, 
 %there is $\mathcal{W}\subset (0,1]^{2d}$ with $
%{\rm meas}((0,1]^{2d}\setminus\mathcal{W})=o(1)$, such that for $(\bm \alpha, \bm \theta)\in \mathcal{W}$, 
	there exist a diffeomorphism $\bm\omega=\bm\omega(\bm a)$ on $[0,1]^b$  and a set in $\bm a$,   $\mathcal{R}\subset [1,2]^b$ of ${\rm meas}([1,2]^b\setminus\mathcal{R})\leq  (\varepsilon+\delta)^{c}$  so that  
	the following holds true:
	If $\bm a\in \mathcal{R}$ and $\bm\omega=\bm\omega(\bm a)$, then  $|\bm \omega-\bm \omega^{(0)}|\lesssim \delta$  and 
	\begin{align*}
	u(t,\bm n)=\sum_{\bm k\in\Z^b}\hat u (\bm k,\bm n)  e^{i\bm k\cdot\bm\omega t}%e^{ik\cdot\omega t}%=u^{(0)}+O(\sqrt{\varepsilon+\delta})
	\end{align*}
			is a solution to \eqref{NLS}. Moreover, one has
			\begin{align*}
			%q (k,n)&=q(-k,n)\in\R,\\
			 \hat u (\bm e_\ell, \bm n_{\ell})&=a_\ell, \ 1\leq \ell\leq b,\\
			\sum _{(\bm k, \bm n)\notin \mathcal{S}_+} |\hat u (\bm k,\bm n)|e^{|\bm k|+|\bm n|}&<\sqrt{\varepsilon+\delta},
			\end{align*}
			where $\mathcal{S}_+=\{(\bm e_\ell,\bm n_{\ell})_{\ell=1}^b\}$ with $\bm e_\ell \ (1\leq\ell \leq b)$ the standard basis vectors for $\Z^b$  (i.e., $\bm e_\ell=(\delta_{\ell,j})_{j=1}^b$). 
			\end{thm}
			
			\begin{rem} Theorem~\ref{mthm} readily generalizes to potentials $V$, which are given by finite Fourier series with both cosine and sine terms.
			\end{rem}

\begin{rem} 
%Denote by $\mathcal{A}$ the set of  $\bm \alpha \in [0,1]^d$ on which Bourgain's geometric lemma (Lemma \ref{Boulem}) is applicable, and $\mathcal{M}$ the set of $(\bm \alpha, \bm \theta)$ so that the Diophantine estimates (in Theorem \ref{clusthm}) hold true. 
%As mentioned earlier, Bourgain proved \cite{Bou07} that on a large set in $(\bm \alpha, \bm \theta)$, 
%Then $\mathcal W=(\mathcal A\times [0,1]^d)\cap\mathcal M$. 
%has large measure, namely,  ${\rm meas}([0,1]^{2d}\setminus\mathcal{W})\lesssim \log^{ -c} \frac{1}{\varepsilon}+(\varepsilon+\delta)^c$ for some $ c=c(b,d, K)>0$. 
%Theorem~\ref{mthm} states that on $\mathcal W$, there are  nonlinear Anderson localized states. 
Note that we do not require the linear equation \eqref {ls} to satisfy Anderson localization.  However, linear Anderson localization does hold
on a large subset of $\mathcal W$ (cf. \cite{Bou07, JLS20}).

%: linear Anderson localization needs further excisions and is on a subset of $\mathcal A\times [0,1]^d$.
%of $\bm\alpha$ from $\mathcal A$ to  eliminate   energies. 
\end{rem}

\subsection {Ideas of the proof} The multi-dimensional phase space, namely $\bm \theta$ being a vector
instead of a scalar as in e.g., \cite{SW23} poses a major challenge. Diophantine properties
of the eigenvalues of the linear Schr\"odinger operator \eqref{SO} plays an essential role in the proof. 
To $O(\delta)$, this could be replaced by Diophantine properties of the potential $V(\bm n\bm\alpha+\bm \theta)$
at different lattice sites $\bm n\in\mathbb Z^d$. We need to show, that the values of the function at these chosen sites are,
in some sense,  {\it linearly independent}.  This is, however, a priori, not obvious, since $V$ is a {\it given function}
on $\mathbb R^d$. To prove linear independence, we use a generalized Wronskian approach  and bound the determinant away from zero by carefully removing some $(\bm \alpha, \bm \theta)$.
This method is applicable to {\it any} trigonometric polynomials, generalizing the approach in \cite{SW23} for the cosine function on $\mathbb R$. It is independent of the main body of the proof and could be of independent interest. Such arguments may be closely related to  Diophantine approximations on manifold developed by Kleinbock and Margulis \cite{KM98}.

The other main ingredient of the proof is Bourgain's approach to Anderson localization of linear quasi-periodic 
Schr\"odinger operators in $d$-dimensions \cite{Bou07}, such as the $\mathcal H$ in \eqref{SO}.  This seems to be, so far,
the only method available to deal with the $d$-dimensional phase space, i.e., $\bm \theta \in [0, 1]^d$, $d\geq3$. 
Among other innovations, it uses semi-algebraic geometry arguments (e.g., Yomdin-Gromov algebraic lemma) to control the resonances. The difficulty of $d$-dimensional phase
space manifests as well in our paper, in that for the linear analysis, semi-algebraic geometry already seems indispensable to
control the resonances; while in previous works, see \cite{BW08, Wan16, LW22, SW23}, this role is filled by Diophantine properties.
To control the resonances in time, we also need a weak second Melnikov type estimates. This could be readily obtained, 
%using the decaying nonlinearity,
similar to \cite{Bou05, BW08, LW22, SW23}.  In this paper,  however,  we have to  overcome   much more serious resonances  in space (cf. Figure \ref{graph}),  which  seems  novel.  The proof is then accomplished by doing 
multi-scale analysis   with some  new  ideas.

{The general problem discussed here is the persistency of quasi-periodic solutions of linear or integrable equations after Hamiltonian perturbation. This
subject is closely related to the well-known ``KAM theory''  of invariant tori in
smooth dynamical systems.  Results along this line were first obtained  by Kuksin \cite{Kuk87} and Wayne  \cite{Way90},  stimulating   considerable progress, cf. e.g.,  \cite{KP96,CY00, BG01,  Bou05b, EK10, LY10, LY11, GT11, BBM13, PP15, EGK16, BBHM18, BKM18, CLSY18, Yua21, BHM23} to
name but a few.   The  existence  of quasi-periodic solutions can also be obtained  using the direct Craig-Wayne-Bourgain method \cite{CW93, Bou94, Bou98, Bou05} (cf. \cite{BW08, BB13, Wan16, HSSY, Wan21,Wan21b, LW22, SW23, KLW23} for more recent results), which turns out to be  more robust  when dealing with multi-dimensional Hamiltonian PDEs (cf.  Chapters 18--20 of \cite{Bou05}).  }
% \cite{BW08, Wan16, LW22, GSW23, SW23, KLW23}.
%The use of Bourgain's method \cite{Bou07} in the linear analysis also entails

Finally, we  mention the work \cite{Yua02} in which a  KAM scheme  was first  developed to prove the  existence of quasi-periodic solutions for some nonlinear discrete  equations without external parameters (i.e., regarding  the initial states  as parameters).  Later in \cite{GYZ14}, the authors  also applied the KAM  method  to obtain the existence of quasi-periodic solutions for the nonlinear quasi-periodic Schr\"odinger equation on $\Z$.

\subsection {The nonlinear random Schr\"odinger equation}
When the potential $V$ with $V(\bm n)=V(\bm n \bm \alpha+\bm \theta)$ is replaced
by a random potential, for example,  with $V(\bm n)$, a family of
independently identically distributed random variables (with the uniform distribution),  it is known
from \cite{BW08} 
that the nonlinear random Schr\"odinger equation 
\begin{align*}
iu_{t}+\varepsilon\Delta u +Vu+\delta |u|^{2p}u=0\ {\rm on} \ \mathbb Z^d 
\end{align*}
has large sets of Anderson localized states for small $\varepsilon$ and $\delta$. 
When $d=1$, it is shown further in the important work \cite{LW22} that large sets of Anderson localized states persist for all $\varepsilon\neq 0$ 
(similar results could be proven for the 
nonlinear wave equation).  
Recall that  the first proof of  the Anderson localization for the linear random Schr\"odinger equation 
\begin{align*}
iu_{t}+\varepsilon\Delta u+Vu=0\   {\rm on }\  \mathbb Z^d 
\end{align*}
for small $\varepsilon$ was based on multi-scale analysis type Green's function estimates  developed in \cite{FS83}, see also \cite{AM93}. %In $d=1$ this was proven 
%for all $\varepsilon\neq 0$ {\cite{GMP77, KS80}}. 

It is generally believed that the linear random Schr\"odinger equation and the linear 
quasi-periodic Schr\"odinger equation should share common localization features in the perturbative regime, despite
the quasi-periodic problems being more delicate and the results more
difficult to obtain. The use of semi-algebraic geometry and Cartan's Lemma \cite{BGS02, Bou07, JLS20}, for example,
certainly renders  the quasi-periodic issue more involved than that of the random, and this seems to continue to the nonlinear case, 
as mentioned earlier. Theorem~\ref{mthm} generalizes this circle of ideas to the nonlinear setting 
by providing a concrete example where the quasi-periodic and the random have indeed similar behavior.

\subsection{Structure of the paper} In Section \ref{ldtsec}, we focus on the  linear estimates, in the form of large deviation theorem. The paper concludes
by constructing the quasi-periodic in time solutions in Section \ref{nonsect}. The proof of the Diophantine type estimates and some important lemmas on resolvent identities    are presented in the Appendix. 

%\section{The general scheme}

\section{Linear analysis}\label{ldtsec}
	In this section,  we investigate  properties of linearized operators on $\Z^{b+d}$. Of particular importance is the {\it large deviation theorem (LDT)}  for the corresponding Green's functions.  
	\subsection{The operator on the lattice}
	For a vector $\bm x\in\R^r$, we denote by $|\bm x|$ (resp. $|\bm x|_2$) the supremum norm (resp. the Euclidean norm). For an operator (or a matrix), we denote by $\|\cdot\|$  its operator norm. 
	
	We define the lattice:
	$$\Z^{b+d}_{{\rm pm}}=\Z^{b+d}\times\{+, -\}.$$
Let  $\mathcal{S}_+=\{(\bm e_\ell, \bm n_\ell)_{\ell=1}^b\}\subset\Z^{b+d}$ and we   identify $ \mathcal{S}_+$ with $\{(\bm e_\ell, \bm n_\ell)_{\ell=1}^b\}\times\{+\}\subset \Z^{b+d}_{{\rm pm}}$.
Similarly, let $\mathcal{S}_-=\{(-\bm e_\ell, \bm n_\ell)_{\ell=1}^b\}\subset \Z^{b+d}$, which is then identified with $\{(-\bm e_\ell, \bm n_\ell)_{\ell=1}^b\}\times\{-\}\subset \Z^{b+d}_{{\rm pm}}$. We let $\mathcal S=\mathcal S_+\cup\mathcal S_-$ and  $$\Z^{b+d}_{{\rm pm}, *}=\Z_{\rm pm}^{b+d}\setminus \mathcal{S}.$$

	%Recall that $ \mathcal{S}_+=\{(\bm e_l, \bm n_l)_{l=1}^b\},\,  \mathcal{S}_-=\{(-\bm e_l, \bm n_l)_{l=1}^b\}$  { and } $\mathcal{S}=\mathcal{S}_+\cup \mathcal{S}_-.$ Define 
	%$$\Z^{b+d}_{{\rm pm}}=\Z^{b+d}\times\{+, -\},$$
	%and identify $ \mathcal{S}_+$ with $\{(\bm e_l, \bm n_l)_{l=1}^b\}\times\{+\}$ and $\mathcal{S}_-$ with $\{(-\bm e_l, \bm n_l)_{l=1}^b\}\times\{-\}$, which are all subsets of $\Z^{b+d}_{{\rm pm}}.$
	%Thus we  can write $\Z^{b+d}_{{\rm pm},*}=\Z^{b+d}_{{\rm pm}}\setminus\mathcal{S}$. 
	%We also denote $\Z_*^{b+d}=\Z^{b+d}\setminus \mathcal{S}.$
	
	In the following,  we study on $\Z^{b+d}_{{\rm pm},*}$
	%$\Z_*^{b+d}=\Z^{b+d}\setminus \mathcal{S}$
	%the $(2\times 2)$-block 
	the operator 
	\begin{align}\label{hsigm}
	{H}(\sigma)={D}(\sigma)+\varepsilon (\Delta\oplus\Delta)+\delta S,\ \sigma\in\R, 
	\end{align}
	where 
	\begin{align*}
%\mathcal{L}&=\mathcal{D}+\varepsilon \mathcal{T}\\
{D}(\sigma)&={\rm diag}_{(\bm k,\bm n)\in\Z^{b+d}}\left(\begin{array}{cc}
	{-\sigma-\bm k\cdot\bm\omega+\mu_{\bm n}} & {0} \\
	{0}& {\sigma+\bm k\cdot\bm\omega+\mu_{\bm n}} \\
	\end{array} \right),\ \mu_{\bm n}=V(\bm n\bm \alpha+\bm\theta)\\
%\mu_n&=\sqrt{\cos(n\cdot\alpha+\theta_0)+m}, \\
%{\rm diag}_{(k,n)\in \Z_*^{b+d}}\left(
%	\begin{array}{cc}
%(-(\sigma+k\cdot\omega)+\mu_n)\mu_n& 0\\
%0&(\sigma+k\cdot\omega+\mu_n)\mu_n
%\end{array}\right),\\
%{T}_\phi((k,n);(k',n'))&=\phi(k-k',n)\delta_{n,n'},
\end{align*}
	and the operator $S$ refers typically to the linearized operator of the nonlinear  perturbation. So we may let $S$ satisfy the following properties:
	\begin{itemize}
	\item  First $$S:\ \ell^2(\Z^{b+d}_{{\rm pm}})\to \ell^2(\Z^{b+d}_{{\rm pm}})$$
	is a bounded self-adjoint operator. 
	\item For all $\bm k,\bm k', \bm k''\in\Z^b$, $\bm n,\bm n'\in\Z^d$ and $\xi,\xi'\in\{+, -\}$,  we have the T\"oplitz property in the  $\bm k$-variable: 
	\begin{align*}
	S((\bm k'+\bm k, \bm n, \xi); (\bm k''+\bm k, \bm n', \xi'))=S((\bm k', \bm n, \xi ); (\bm k'', \bm n', \xi')). 
	\end{align*}
%Write $\bm j=(\bm k, \bm n)$. %We assume the two diagonal elements of $S(\bm j;\bm j')$ are identical for all $\bm j,\bm j'.$
	\item %Denote by $\|\cdot\|$ the operator norm. 
	We have for some $C_2>0$ and $\gamma\in(\frac12, 10), $
	\begin{align*}
	|S((\bm k, \bm n, \xi); (\bm k', \bm n', \xi'))|\leq  C_2(1+|\bm k-\bm k'|)^{C_2} e^{-\gamma|\bm k-\bm k'|-\gamma|\bm n|}\delta_{\bm n,\bm n'}. 
	\end{align*}
	\end{itemize} 
	%with  $\phi:\ \Z^{b+d}_*\to\R$ satisfying $\phi(k,n)=\phi(-k,n)$ and for some $C>0, \gamma>0$
	%\begin{align*}
	%|\phi(k,n)|\leq Ce^{-\gamma(|k|+|n|)}\ {\rm for}\ \forall\  (k,n)\in\Z_*^{b+d} .
	%\end{align*}
	The frequency $\bm \omega$ satisfies 
	\begin{equation}\label{Omega}
	\bm\omega\in\Omega=\bm\omega^{(0)}+[-C\delta, C\delta]^b,
	\end{equation}
	 where $C>0$ and $\bm\omega^{(0)}=\bm \omega^{(0)}(\bm \alpha, \bm \theta)$ is defined in Theorem \ref{mthm}.

\subsection{The (generalized) elementary regions on the lattice}
For some technical reason, we will introduce the definitions of   {\it elementary  regions}  and {\it generalized elementary regions}  on $\Z^{b+d}$ originated from \cite{BGS02, Bou07}. We refer to  \cite{Liu22}  (cf. Section 2 in  \cite{Liu22}) for some  important clarifications and   improvements on those concepts.  %In particular,  Liu used the {\it width}  of a set, which is important for covering generalized elementary regions   with elementary regions.  Since we will project elementary regions  of  $\Z^{b+d}$ onto $\Z^d$ and we hope the projected sets can be well-covered by elementary regions on $\Z^d$, Liu's argument  in \cite{Liu22} becomes necessary for us. 

Denote by $\Lambda_N(\bm x)$ ($\bm x\in\Z^{r},\ r=b+d,d$), $$\Lambda_N(\bm x)=\{\bm y\in\Z^{r}:\ |\bm y-\bm x|\leq N\},$$ the cube of radius $N$  centered at  $\bm x$.	 In particular, we write  $\Lambda_N=\Lambda_N(\bm 0)$.  For $\Lambda\subset\Z^{r}$, we define its  diameter  to be  ${\rm diam}\  \Lambda=\sup\limits_{\bm x, \bm y\in\Lambda}|\bm x-\bm y|$. 
For $\bm x\in\Z^{r}$ and $\bm w\in\R_+^{r}$, let $R_{\bm w}(\bm x)=\{\bm y\in\Z^{r}:\ |y_\ell-x_\ell|\leq w_l , \, 1\leq\ell\leq r\}$ denote the rectangle.

A  \textit{generalized elementary region} is defined to be a set $\Lambda$ of the form 
$$\Lambda=R_{\bm w}(\bm x)\setminus(R_{\bm w}(\bm x)+\bm z),$$
where  $\bm z\in\Z^{r}$ is arbitrary.  The size of  a generalized elementary region  is simply its diameter. The set of all {\it generalized elementary regions}  of size at most $M$  will be denoted by $\mathcal{E}_M$. %Elements of $\mathcal{ER}(M)$ are also referred to as $M$-regions.  
%Define 
%$$\mathcal{E}_{M, \bm 0}=\left\{\Lambda\in\mathcal{E}_M:\   \Lambda=R_{\bm w}(\bm 0)\setminus(R_{\bm w}(\bm 0)+\bm z), \text{ diam }\Lambda=M,\ \bm z\in\Z^r\right\}.$$

An {\it elementary region}  $Q_N$ of size $N$ and centered at $\bm 0$ is one of the following regions 
$$Q_N=\Lambda_N\ {\rm or}\ Q_N=\Lambda_N\setminus\{\bm x\in\Z^r:\ x_\ell \square_\ell 0,\ 1\leq\ell\leq r\},$$
 where  $\square_\ell \in \{<, >, \emptyset\}$ and at least two $\square_\ell$'s are not $\emptyset.$ Denote by $\mathcal{ER}_{\bm 0}(N)$ the set of all elementary regions of size $N$ and centered at $\bm 0$. Let $$\mathcal {ER}(N)=\{\bm x+Q_N:\ \bm x\in\Z^r,\ Q_N\in \mathcal{ER}_{\bm 0}(N)\}.$$

Let   $\mathcal E_{M}^{L}$  denote  the  set of all $M$-size generalized elementary regions of  width at least $L\leq M$, namely,  $\Lambda\in \mathcal {E}_{M}^{L}$ iff $\Lambda\in \mathcal {E}_{M}$ and,   for any  $\bm x\in\Lambda, 0<L'<L$, there is some $\Lambda'\in\mathcal{ER}(L')$ so that $\bm x\in \Lambda'\subset\Lambda$, ${\rm dist}(\bm x, \Lambda\setminus\Lambda')\geq L'/2$,  where  ${\rm dist} (\cdot, \cdot)$ is induced by the supremum norm.  So we have  $\mathcal {ER}(N) \subset \mathcal E_{2N}^{N}$. 

With a slight abuse of notation, we also use $\mathcal {ER}_{\bm 0}(N), \mathcal {ER}(N), \mathcal E_N, \mathcal E_N^L , \Lambda_N(\bm x),\ \Lambda_N$   to denote $\mathcal {ER}_{\bm 0}(N)\times\{+,-\}, \mathcal {ER}(N)\times\{+,-\}, \mathcal E_N\times\{+,-\}, \mathcal E_N^L \times\{+,-\}, \Lambda_N(\bm x)\times\{+,-\}, \Lambda_N\times\{+,-\}$. Similarly,  for any $\Lambda\subset \Z^r$, denote by $R_\Lambda$ the restriction to
$\Lambda\times\{+,-\}$. 

\subsection{Green's functions and LDE}
We now define Green's functions  on subsets of $\Z_{{\rm pm},*}^{b+d}$. Let $\Lambda\subset\Z_{{\rm pm},*}^{b+d}$ be a non-empty set, $R_\Lambda$ 
 the restriction operator to $\Lambda$, and define the $(\# \Lambda) \times (\#\Lambda)$ matrix: 
 \begin{equation}\label{rest}
 (R_\Lambda {H}(\sigma)R_\Lambda)((\bm x, \xi);(\bm x', \xi'))={H}(\sigma)((\bm x, \xi);(\bm x', \xi')), \ (\bm x, \xi), (\bm x', \xi')\in\Lambda,
 \end{equation}
 where $H(\sigma)$ is defined by \eqref{hsigm}. 
Define  the Green's function  to be (if it exists) 
 $$G_\Lambda(\sigma)=(R_\Lambda {H}(\sigma)R_\Lambda)^{-1}.$$
 %where $R_\Lambda$ denotes the restriction operator to $\Lambda$: $$(R_\Lambda {H}(\sigma)R_\Lambda)(i, j)={H}(\sigma)(i, j),$$ if $i, j\in\Lambda$, and $0$ otherwise.
%Denote by $\Pi_{b, d}\Lambda$ the projection of $\Lambda\subset\Z_{{\rm pm},*}^{b+d}$ onto $\Z^{b+d}$.
	
Before addressing the LDT for Green's functions, let us first give the definition of large deviation estimates (LDE): 
\begin{defn}[{\bf LDE}]\label{LDEdefn}
Let $\rho>0$, $0<\gamma'<\gamma$ and $M\in\N$.  We say ${H}(\sigma)$ defined by \eqref{hsigm} satisfies the $(\rho, \gamma',M)$-LDE  if 
there exists a set $\Sigma_M\subset \R$ with $${\rm meas}(\Sigma_M)\leq e^{-M^{\rho}},$$ so that for $\sigma\notin \Sigma_M$ the following estimates hold:  If $\Lambda\in(\bm 0,\bm n)+ \mathcal{ER}_{\bm 0}(M)$ satisfies  $|\bm n|\leq 10M$, then
\begin{align*}
\|G_{\Lambda}(\sigma)\|&\leq e^{M^{\frac34}},\\
%|G_{\Lambda}(\sigma)(\bm x,  \bm x')|&:=
|G_{\Lambda}(\sigma)((\bm x,  \xi); (\bm x', \xi'))|
&\leq e^{-\gamma'|\bm x-\bm x'|}\ {\rm for}\ |\bm x-\bm x'|\geq M^{\frac 89}\ {\rm and }\  \xi,\xi'\in\{\pm\}, 
\end{align*}

where $\bm x=(\bm k',\bm n'),  \bm x'=(\bm k'',\bm n'')$. % and $\|G_{\Lambda}(\sigma)((\bm x,  \xi); (\bm x', \xi'))\|$  represents the  . %  and $ \xi,  {\xi'} \in\{+, -\}$.

% and $|G_{\Lambda}(\sigma)(j;j')|$ denotes the operator norm of the $2\times2$ real matrix $G_{\Lambda}(\sigma)(j;j')$. 
\end{defn}
%\begin{rem}
%For simplicity, in the following, we will also denote $\bm x=(\bm k, \bm n, \xi)$ for $\xi\in\{\pm\}$ with a slight abuse of notation.  So with this convenience,  by $|\bm x|$ we mean $|(\bm k, \bm n)|$. 
%\end{rem}

We will show in this section that the $(\rho, \frac{\gamma}{2},N)$-LDE hold  for  some  small $$\rho=\rho(b,d)>0$$  and all $N\gg1$,
under certain non-resonant conditions on $\bm\alpha, \bm\theta$ and $\bm\omega$.  This leads to the LDT, Theorem~\ref{ldtthm} below. Since we view $\bm\omega$ as an $O(\delta)$ perturbation of $\bm\omega^{(0)}$, the properties of $\bm\omega^{(0)}$ become essential in the proof of LDT.  So, in the following,  we first establish non-resonant properties (called Diophantine estimates) of $\bm\omega^{(0)}$. Then the proof of LDT will follow from a multi-scale analysis scheme in the spirit of \cite{BGS02, Bou05, Bou07}.

	%The following three subsections are devoted to prove Theorem \ref{ldtthm}.

	\subsection{Diophantine estimates}\label{NRomega0}
		Recall that 
	\begin{align*}
	\bm \omega^{(0)}&=\bm\omega^{(0)}(\bm\alpha,\bm\theta)=(V(\bm n_\ell\bm \alpha+\bm \theta))_{\ell=1}^b,\\
	\mu_{\bm n}&=\mu_{\bm n}(\bm\alpha,\bm\theta)=V(\bm n\bm \alpha+\bm \theta),\ \bm n\in\Z^d.
	\end{align*}
The main purpose of this section is to obtain  lower bounds on 
	\begin{align*}
	&|\bm k\cdot\bm\omega^{(0)}|\ {\rm for}\ \bm k\in\Z^b\setminus\{\bm 0\},\\
	&|\pm\bm k\cdot\bm\omega^{(0)}+\mu_{\bm n}| \ {\rm for}\ (\bm k,\bm n)\in\Z^{b+d}\backslash {\mathcal S} ,\\
	&|\bm k\cdot\bm \omega^{(0)}+\mu_{\bm n}-\mu_{\bm n'}| \ {\rm for} \ \bm n\neq \bm n'\in\Z^d,
	\end{align*}
	under certain restrictions on $(\bm\alpha, \bm\theta).$  As mentioned in the Introduction, due to the $d$-dimensional phase space, such Diophantine-type estimates are 
	highly non-trivial. We use a generalized Wronskian approach and work out the details in the Appendix (cf. Theorem \ref{nlsthm} and Corollary \ref{dccor}). These estimates 
	%maybe of independent interest
	can be read independently, and are key to the existence of quasi-periodic solutions to the nonlinear equation.
Applying  Corollary \ref{dccor} % Theorem \ref{DC} using \eqref{L}
%and taking all  intersections (i.e., $\mathcal{M}=\bigcap_{1\leq i\leq 4}\mathcal{M}_i$) in the above definition  
%with the above  
%$$L_{\varepsilon,\delta}=(\varepsilon+\delta)^{-\frac{1}{10^3b^2d^2\tau_{b+2}^2}}, \iota=(\varepsilon+\delta)^{\frac{1}{10b}}$$
then  leads to 
\begin{thm}\label{clusthm}
There exist some $0<c_1=c_1(b,d,K,V)<\frac{1}{100b}, C_1=C_1(b,d, K, V)>1$ such that, if  $0<\varepsilon+\delta\leq \delta_0(b,d,K,V,\max\limits_{1\leq\ell\leq b}|\bm n_\ell|)\ll1$, then there is some $\mathcal M \subset [0, 1]^{2d}$
satisfying  
\begin{align}\label{mmeas}
{\rm meas}([0,1]^{2d}\setminus\mathcal M)\leq \log^{-1}\frac{1}{\varepsilon+\delta}, 
\end{align}
so that the following properties hold true for $(\bm \alpha, \bm \theta)\in\mathcal M$  and $L_{\varepsilon,\delta}:=100(\varepsilon+\delta)^{-c_1}$. 
  \begin{itemize}
%\item[(1)]  For all $n\neq n'\in \Z^d$ with $|(n,n')|\leq L$, we have 
%\begin{align}\label{spb}
%|\mu_n-\mu_{n'}|\geq \frac{2}{\pi^2} L^{-6d}.
%\end{align}
%\item[(2)] For all $k\in\Z^{b}\setminus \{0\}$ with $|k|\leq 2L$, we have
%\begin{align}\label{dcthm}
% |k\cdot\omega^{(0)}|>\eta.
%\end{align}
\item [(1)] For all $\log \frac{1}{\varepsilon+\delta}\leq L\leq L_{\varepsilon,\delta}$ and all $(\bm k,\bm n)\in \Lambda_{L}\setminus\mathcal S,$  we have 
\begin{align}\label{mk1thm}
 \min_{\xi=\pm1}|\xi\bm k\cdot\bm \omega^{(0)}+\mu_{\bm n}|>L^{-C_1}.
\end{align}
\item[(2)]  We have for any $\xi=\pm1,$
\begin{align}\label{sublthm}
\sup_{\sigma\in\R}\#\left\{(\bm k,\bm n)\in \Lambda_{L_{\varepsilon,\delta}}:\   |\xi(\sigma+\bm k\cdot\bm \omega^{(0)})+\mu_{\bm n}|<\frac{(\varepsilon+\delta)^{\frac{1}{8b}}}{4}\right\}\leq b.
%\sup_{\sigma\in\R}\#\left\{(k,n)\in \Lambda_L:\ |\sigma+k\cdot\omega^{(0)}+\mu_n|<\frac\eta2\right\}\leq b
\end{align}
%where $\#(\cdot)$ denotes the cardinality of a set, as before.
\end{itemize}
\end{thm}
\begin{rem}
The estimate in (1) plays an essential role in the initial iteration steps for  the nonlinear analysis.  
The conclusion (2) of this theorem will be only used in the proof of intermediate scales LDE. 
\end{rem}

\begin{proof}
The proof is similar to  that in \cite{SW23}. First, the measure estimate \eqref{mmeas} is a consequence of Corollary \ref{dccor}.  Next,  the inequality  \eqref{mk1thm} in (1)  follows directly from (iii) of  Corollary \ref{dccor}. 

So,  it suffices to establish \eqref{sublthm} of (2).  We prove it by the contradiction. Without loss of generality,  we consider the case $\xi=+1$.  Assume  that there are $b+1$ distinct   $\{(\bm k_\ell,  \bm n_{\ell}')\}_{1\leq \ell\leq b+1}$ satisfying 
\begin{align}\label{sings}
|\sigma+\bm k_\ell\cdot\bm\omega^{(0)}+\mu_{\bm n_{\ell}'}|<\frac{(\varepsilon+\delta)^{\frac{1}{8b}}}{4}.
\end{align}
We claim that all  $\bm k_{\ell}$ ($1\leq \ell\leq b+1$) are distinct. Actually, if there are $\bm k_{\ell_1}=\bm k_{\ell_2}$ for some $1\leq \ell_1\neq \ell_2\leq b+1$, then it must be that $\bm n_{\ell_1}'\neq \bm n_{\ell_2}'$. As a result, we obtain using \eqref{sings} that
\begin{align*}
|\mu_{\bm n_{\ell_1}'}-\mu_{\bm n_{\ell_2}'}|&\leq |\sigma+\bm k_{\ell_1}\cdot\bm\omega^{(0)}+\mu_{\bm n_{\ell_1}'}|\\
&\ \ \ +|\sigma+\bm k_{\ell_2}\cdot\bm\omega^{(0)}+\mu_{\bm n_{\ell_2}'}|\\
&\leq \frac{(\varepsilon+\delta)^{\frac{1}{8b}}}{2},
\end{align*}
which contradicts  (i) of Corollary \ref{dccor}. Next, we claim that  $\{\bm n_{\ell}'\}_{1\leq \ell\leq  b+1}\subset \{\bm n_{\ell}\}_{1\leq \ell\leq b}.$ In fact, if $\bm n_{\ell_1}'\notin  \{\bm n_{\ell}\}_{1\leq \ell\leq b}$ for some $1\leq \ell_1\leq b+1$, we can assume that there is $\ell_2\neq \ell_1$ so that $\bm n_{\ell_2}'\neq \bm n_{\ell_1}'$. Otherwise, we may have 
\begin{align*}
|(\bm k_{\ell_1}-\bm k_{\ell_2})\cdot\bm\omega^{(0)}|\leq\frac{(\varepsilon+\delta)^{\frac{1}{8b}}}{2},
\end{align*}
which contradicts  (ii) of Corollary \ref{dccor}, since $\bm k_{\ell_1}\neq \bm k_{\ell_2}$ as claimed above. 
So  we obtain for some  $\ell_2\neq \ell_1$ and $1\leq \ell_2\leq b+1$  with $\bm n_{\ell_2}'\neq \bm n_{\ell_1}'$ that
\begin{align*}
|(\bm k_{\ell_1}-\bm k_{\ell_2})\cdot\bm\omega^{(0)}+\mu_{\bm n_{\ell_1}'}-\mu_{\bm n_{\ell_2}'}|\leq\frac{(\varepsilon+\delta)^{\frac{1}{8b}}}{2},
\end{align*}
which contradicts (iv) of Corollary \ref{dccor}.  So  it suffices to assume $\{\bm n_{\ell}'\}_{1\leq \ell\leq b+1}\subset \{\bm n_{\ell}\}_{1\leq \ell\leq b}.$ However, in this case,  we have  by  the pigeonhole principle  that there are $1\leq \ell_1\neq \ell_2\leq b+1$ so that $\bm n_{\ell_1}'=\bm n_{\ell_2}'$. Then we get 
\begin{align*}
|(\bm k_{\ell_1}-\bm k_{\ell_2})\cdot \bm \omega^{(0)}|\leq \frac{(\varepsilon+\delta)^{\frac{1}{8b}}}{2},
\end{align*}
which contradicts  (ii) of Corollary \ref{dccor}  as shown above.  This proves \eqref{sublthm}. 
\end{proof}

\subsection{Bourgain's geometric lemma} To proceed with the analysis, we will also need Anderson localization-type properties for the linear quasi-periodic Schr\"odinger operator 
$\mathcal H$ in \eqref{SO}, namely, 
$$\mathcal H=\mathcal H(\bm \alpha, \bm\theta)=\varepsilon\Delta+V(\bm\theta+\bm n\bm\alpha)\delta_{\bm n, \bm n'}.$$
 This was proven by Bourgain in the remarkable work \cite{Bou07}, where he established Anderson localization for general  quasi-periodic Schr\"odinger operators on $\Z^d$ with non-degenerate analytic potentials on $\T^d$ (identified with $[0,1]^d$).
%established a powerful method based on semi-algebraic geometry and Cartan's Lemma to prove
%to eliminate resonances in the Anderson localization of general  quasi-periodic Schr\"odinger operators on $\Z^d$ with analytic potentials on $\T^d$. 
This result was later significantly  extended  by Jitomirskaya-Liu-Shi \cite{JLS20} (cf. also \cite{Shi22, Liu22}) to  potentials on $\T^b$ for $b\geq d.$ 

In the present paper, due to the $d$-dimensional phase space, we need to apply Bourgain's  arguments. 
Using short-range property of $S$ in the $\bm n$-variable, Bourgain's geometric lemma is made available to handle LDE for Green's functions at large scales. {\it This is a main new aspect not present in \cite{SW23}}.  The general scheme can be contrast with that in Chapter 19 of  \cite{Bou05}, where  Bourgain made a great breakthrough and first proved the existence of quasi-periodic solutions  for the nonlinear Schr\"odinger equation  on arbitrarily dimensional torus (a KAM theorem was  later established  by Eliasson-Kuksin  in an important paper  \cite{EK10}).  Among others, Bourgain \cite{Bou05} employed heavily the {\it separation property} of eigenvalues  of the Laplace operator on the higher dimensional  torus, while  the {\it separation property}  is missing  in the present setting (indeed, the spectrum of $\mathcal H$ may contain  an interval).

{We start with  introducing  the definition of semi-algebraic sets. 
\begin{defn}
A set $X\subset \mathbb{R}^r$ is called {\it semi-algebraic} if it is a finite union of sets defined by a finite number of polynomial equalities and inequalities. More precisely, let $\{P_1,\cdots,P_k\}\subset\mathbb{R}[x_1,\cdots,x_r]$ be a family of real polynomials whose degrees are bounded by $p$. A (closed) semi-algebraic set ${X}$ is given by an expression
\begin{equation}\label{smd1221}
	{X}=\bigcup\limits_{s}\bigcap\limits_{\ell\in\mathcal{K}_s}\left\{\bm x\in\mathbb{R}^r: \ P_{\ell}(\bm x)\varsigma_{s\ell}0\right\},
\end{equation}
where $\mathcal{K}_s\subset\{1,\cdots,k\}$ and $\varsigma_{s\ell}\in\{\geq,\leq,=\}$. Then we say that ${X}$ has degree at most $kp$. In fact, the degree of ${X}$ which is denoted by $\deg {X}$, means the  smallest $kp$ over all representations as in (\ref{smd1221}).
\end{defn}}

%We also remark that since our  operators $\mathcal{D}+\varepsilon \mathcal{T}$ are \textit{non-diagonal}  in the  $n$-direction for  $n\in\Z^d$, 
%to establish LDT,  it requires fine properties of such  operators projected on $\Z^d$. In particular, an additional restriction on $\alpha\in[0, 1]$ is necessary.  D
Denote by $\mathcal{ER}_{\Z^d}(M)$ (resp. $\mathcal{ER}_{\Z^d, \bm 0}(M)$)  the set of all $M$-size elementary regions  on $\Z^d$ (resp. centered at $\bm 0$).   We have 
\begin{lem}[cf. the \textsc{Claim} (page 694)  in \cite{Bou07} and Theorem 2.7 in \cite{JLS20}]\label{Boulem}
There are small constants $0<\kappa_1<\kappa_2<1$  depending only on $d$ so that the following holds true. 
For $0<\varepsilon\leq \varepsilon_0(V,d)\ll1 $ and  $N\geq N_1:=10^{-\frac{32}{\kappa_1^2}}\log^{\frac{4}{\kappa_1}}(\log \frac1\varepsilon)$, 
 there is a semi-algebraic set  
$\mathcal{A}_N\subset [0,1]^d$ satisfying $$\deg \mathcal A_N\leq N^{4d},\ {\rm meas}(\bigcap_{N\geq N_1}\mathcal{A}_N)=1-O(\log^{-c}\frac1\varepsilon),\ c=c(d)>0, $$   so that the following properties hold true for $\bm\alpha\in \mathcal{A}_N.$  
%For  $N\geq N_0\gg1$,  
\begin{itemize}
\item[(1)] For all $E\in\R$ and all  $\bm\theta\in[0,1]^d$,  there is  $L\in[N^{\kappa_1}, N^{\kappa_2}]$ so that, 
%there is $L\in[N^{\kappa_1}, N^{\kappa_2}]$ 
for all $Q\in\mathcal{ER}_{\Z^d}(L_1) $ satisfying $L_1\sim (\log N)^{\frac{4}{\kappa_1}}$ and 
$$ Q\subset [-L,L]^d\setminus[-L^{\frac{1}{10d}}, L^{\frac{1}{10d}}]^d,$$
one has 
\begin{align*}
\|\mathcal{H}^{-1}_{Q}(E;\bm\theta)\|&\leq e^{\sqrt{L_1}},\\
|\mathcal{H}^{-1}_{Q}(E;\bm\theta)(\bm n; \bm n')|&\leq e^{-\frac12|\log\varepsilon|\cdot|\bm n-\bm n'|}\ {\rm for}\ |\bm n-\bm n'|\geq  L_1^{\frac 89},
\end{align*}
where 
\begin{align*}
\mathcal{H}_{Q}(E;\bm\theta)&=R_{Q}\left(\varepsilon\Delta+V(\bm\theta+\bm n\bm\alpha)\delta_{\bm n, \bm n'}\right)R_{Q}-E\\
&:=\mathcal{L}_{Q}(\bm \theta)-E. 
\end{align*}
%and $0<\tilde \rho<\tilde \frac{8}{9}<1$ depend only on $d$. 
\item[(2)] There is some $\Theta_N=\Theta_N(\bm \alpha)\subset[0,1]^d$ satisfying 
$$ {\rm meas}([0, 1]^d\setminus\Theta_N)\leq e^{-N^{\frac{\kappa_1}{4}}}$$
so that, if $\bm \theta\in\Theta_N$, then one has for all $\Lambda\in \mathcal{ER}_{\Z^d, \bm 0}(N),$
\begin{align*}
\|\mathcal{H}^{-1}_{\Lambda}(\bm \alpha, \bm\theta)\|&\leq e^{\sqrt{N}},\\
|\mathcal{H}^{-1}_{\Lambda}(\bm \alpha, \bm\theta)(\bm n; \bm n')|&\leq e^{-\frac12|\log\varepsilon|\cdot|\bm n-\bm n'|}\ {\rm for}\ |\bm n-\bm n'|\geq {N^{\frac 89}}. 
\end{align*}
Denote 
\begin{align}\label{wpdefn}
\nonumber&\mathcal{A}=\bigcap_{N\geq N_1}\mathcal{A}_N,\ \tilde\Theta_N(\bm \alpha)=\bigcap_{|\bm n|\leq e^{N^{\frac{\kappa_1}{5}}}}\left\{\bm\theta:\ \bm\theta+\bm n\bm\alpha\in\Theta_N\right\}, \\
&\mathcal W'=\left\{(\bm \alpha, \bm\theta)\in[0,1]^{2d}:\ \bm\alpha\in\mathcal A,\ \bm\theta\in\bigcap_{N\geq N_1}\tilde\Theta_N(\bm\alpha)\right\}. 
\end{align}
Then 
\begin{align}\label{Wpri}
{\rm meas}(\mathcal W')\geq 1-\log^{-c}\frac1\varepsilon. 
\end{align}

\end{itemize}
\end{lem}
\begin{rem}
In this  lemma, the constants $\kappa_1, \kappa_2$ correspond to $c_3, c_4$  of  \cite{JLS20}, respectively. Originally, the off-diagonal exponential decay distance (cf. e.g.,  decay estimates in (2) of Lemma \ref{Boulem})  in \cite{Bou07, JLS20} is $|\bm n-\bm n'|\geq \frac {N}{10}$. However, it can be improved to the present sublinear scale $N^{\frac89}.$ This issue is important for the nonlinear analysis, and can be resolved with little effort, cf.  Remark \ref{rmkJLS}  and Remark \ref{rmkLiu} in the Appendix  for more details. 
\end{rem}

\begin{rem}\label{Boulemrm}
The proof of this lemma is based on analysis of semi-algebraic sets  (cf. \cite{Bou07, JLS20}), and applies to general non-degenerate  analytic potentials.  
\end{rem}

\begin{rem} 
We mention that the set $\mathcal W$ on which Theorem~\ref{mthm} holds satisfies
\begin{align}\label{ALset}
\mathcal W=\mathcal W'\cap\mathcal M,
\end{align}
where $\mathcal M$ is the set on which the Diophantine estimates in Theorem~\ref{clusthm} hold, and $\mathcal W'$ is defined by \eqref{wpdefn}.  So, we have 
$${\rm meas}(\mathcal W)\geq 1-\log^{-c}\frac{1}{\varepsilon+\delta},\ c=c(d)>0$$
as  shown in  Theorem \ref{mthm}. 
  For $(\bm \alpha, \bm\theta)\in \mathcal W$, the conclusion  in (2)  plays an essential role in the nonlinear analysis. 

\end{rem}

{\begin{proof}

 It suffices to establish  the measure bounds. The proof is a small modification of  that of Theorem 4.1 in \cite{JLS20} (cf. pages 475--476). Indeed, the only difference is that we will apply Theorem 3.7 in \cite{JLS20} starting from $N_1=10^{-\frac{32}{\kappa_1^2}}\log^{\frac{4}{\kappa_1}}(\log \frac1\varepsilon)$ rather than $\log\log \frac1\varepsilon$ as in \cite{JLS20}.  This change makes sense since the large deviation estimates in \cite{JLS20} hold  for the initial scales of $N_0(V,d)\leq N\leq N_2:=\frac{1}{10}\log^2\frac1 \varepsilon$ (via the Neumann series argument), and for $f(x)=e^{x^{\frac{\kappa_1}{4}}},$
$$[N_1, (f(N_1))^{\frac {8}{\kappa_1}}]\subset [N_0, N_2].$$

This  modification then  leads to the following  measure estimate: 
\begin{align*}
{\rm meas}(\bigcap_{N\geq N_1}\mathcal{A}_N)& \geq 1-\sum_{N\geq N_1} f(N)^{-\kappa_1}\\
&\geq 1-\sum_{N\geq N_1} e^{-\kappa_1N^{\frac{\kappa_1}{4}}}\\
&\geq 1- C(d)e^{-\frac{\kappa_1}{2}N_1^{\frac{\kappa_1}{4}}}\\
&\geq  1-C(d) \log^{-\frac{\kappa_1}{2}\times10^{-\frac{8}{\kappa_1}}} \frac1\varepsilon. 
\end{align*}

For the proof of \eqref{Wpri},  it  follows directly   from  applying   the Fubini's theorem and $${\rm meas}(\bigcap_{N\geq N_1}\tilde\Theta_N(\bm\alpha))\geq 1-C(d) \log^{-\frac{1}{5}\times10^{-\frac{8}{\kappa_1}}} \frac1\varepsilon.$$ 

\end{proof}
}

%\begin{lem}[\cite{Bou07, JLS20}]
%\label{Boulem}
%For $0<\varepsilon\ll1,$ there is $\mathcal{A}_N\subset [0,1]^d$ with $${\rm meas}(\cap_{N\geq N_0}\mathcal{A}_N)=1-O(\log^{-c}\frac1\varepsilon) \ {\rm as}\  \varepsilon\to 0, c>0, $$   so that the following holds true for $\bm\alpha\in \mathcal{A}_N:$  For  $N\geq N_0\gg1$,  there exist $0<\kappa_1<\kappa_2<1$ (depend only on $d$) so that,  for all $E\in\R$ and all $\bm\theta\in[0,1]^d$, there is $L\in[N^{\kappa_1}, N^{\kappa_2}]$, for  all $\Lambda_1\in\mathcal{ER}_{\Z^d}(L_1) $ with $$L_1\sim (\log N)^C, C=C(d)>1, \Lambda_1\subset [-L,L]^d\setminus[-L^{\frac{1}{10d}}, L^{\frac{1}{10d}}]^d,$$
%one has 
%\begin{align*}
%\|\mathcal{H}^{-1}_{\Lambda_1}(E;\bm\theta)\|&\leq e^{\sqrt{L_1}},\\
%|\mathcal{H}^{-1}_{\Lambda_1}(E;\bm\theta)(\bm n; \bm n')|&\leq e^{-\frac12|\log\varepsilon|\cdot|\bm n-\bm n'|}\ {\rm for}\ |\bm n-\bm n'|\geq \frac{L_1}{10},
%\end{align*}
%where 
%\begin{align*}
%\mathcal{H}_{\Lambda_1}(E;\bm\theta)&=R_{\Lambda_1}\left( \varepsilon\Delta+\mu_{\bm n}\delta_{\bm n, \bm n'}\right)R_{\Lambda_1}-E\\
%&:=\mathcal{L}_{\Lambda_1}(\bm \theta)-E. 
%\end{align*}
%and $0<\tilde \rho<\tilde \frac{8}{9}<1$ depend only on $d$. 
%\end{lem}

\subsection{Large deviation theorem} In this section, we establish that the LDE  (cf.  Definition~\ref{LDEdefn})   hold for {\it all} sufficiently large scales by a multi-scale analysis. From \eqref{Omega},
however, the variation of the frequency, $\bm \omega$ is only $O(\delta)$ and not $O(1)$, hence 
	%Throughout this section we will choose 
	%\begin{align}
% \eta=(\varepsilon+\delta)^{\frac{1}{8b}}, \ L=(\varepsilon+\delta)^{-\frac{1}{10^4db^4}}
	%\end{align}
	%%in Theorem~\ref{clusthm}. 
	%So, the set $\mathcal{M}\subset [2, 3]$ given in Theorem~\ref{clusthm}  will satisfy 
	%$$	{\rm meas}([2,3]\setminus\mathcal{M})\leq L^{50db^2}\eta^{\frac{1}{b+2}}\leq (\varepsilon+\delta)^{\frac{1}{30b^2}}.$$
	%We fix $m\in \mathcal{M}.$  The LDT can then be established assuming further restrictions of $\omega=\omega^{(0)}+O(\delta)\in\Omega.$
	as we will see shortly, the proof of LDT will be accomplished in three steps instead of the usual two: 
	\begin{itemize}
	\item The first step deals with the small scales, for which we establish LDE for all $\bm\omega\in\Omega$ and scales 
	$$N_0\leq N \leq 10^{-\frac{1}{\rho}}\log^{\frac{1}{\rho}} \frac{1}{\varepsilon+\delta},$$
	for some $\rho>0.$  In this step,  we only use the Neumann series argument, which leads to the  above  scales  (cf. Lemma \ref{inilem}). 
	\item %To continue the LDE for scales $$N\geq c  \log^{\frac{1}{\rho}} \frac{1}{\varepsilon+\delta},$$
	%%with some $c\in (0,1)$, one typically imposes certain Diophantine conditions on $\omega$, namely,   
	%\begin{align*}
	%|\bm k\cdot \bm \omega |\geq {N^{-C}}\ {\rm for}\ 0<|\bm k|\lesssim N,
	%%\end{align*}
	%for some $C>1.$ This can be done, as mentioned in sect.~1, and corresponds to approach (i). 
	%However, as mentioned there, we have decided to adopt approach (ii),
	%Using clustering properties of the spectrum of $D$
	%If we regard $\omega$ as a free parameter varying in $\Omega,$ the above condition leads to removing a set of $\omega$ with  measure  $\sim \log^{-C_1}\frac{1}{\varepsilon+\delta}$ for some $C_1>1$,  which is not permitted  since  ${\rm meas} (\Omega)\sim \delta^b\ll  \log^{-C_1}\frac{1}{\varepsilon+\delta}.$ However, keeping in mind that $\omega=\omega^{(0)}+O(\delta)$ and $\omega^{(0)}$ has fine properties 
	 Next, we establish  intermediate scales LDE, i.e., scales in the range
	$$10^{-\frac{1}{\rho}}\log^{\frac{1}{\rho}} \frac{1}{\varepsilon+\delta}\leq  N\leq (\varepsilon+\delta)^{-c_1}$$
	{for all $\bm\omega\in\Omega$}, where $c_1$ is defined in Theorem \ref{clusthm}.  Note that Theorem \ref{clusthm} is available in the above interval of scales. The proof of such  intermediate scales LDE  is based on  preparation type  theorem together with clustering  properties of the spectrum of the diagonal matrix $D(0)$ (cf. \eqref {sublthm}, Theorem~\ref{clusthm}). 
	\begin{rem}  The intermediate scales are needed here because for a {\it fixed} $(\bm \alpha,\bm \theta)$, hence {\it fixed}    $\bm \omega^{(0)}$, $\bm \omega$ only varies in an interval
	of size $O(\delta)$, see \eqref{Omega}, as mentioned above. Consequently, there is a gap in the range of scales before Bourgain-type analysis becomes applicable.
	Fortunately,  the Diophantine estimates for $(\bm \alpha,\bm \theta)\in \mathcal M$ can control the resonances directly as in  \eqref {sublthm},
	leading to LDE for the range of scales in the gap, which we term the {\it intermediate scales}.

	%{\color{red}Let us explain the necessity of the intermediate scales of LDE. As we will see later,  when we deal with the large scales (say, scale $N$) LDE, the following weak second Melnikov's conditions 
	% \begin{align}\label{wsecond}
	% |\bm k\cdot\bm \omega+{\lambda_l}\pm{\lambda_{l'}}|>2 e^{-\frac 14 N_1^{\kappa_1\frac{3}{4}}}\ {\rm for}\ \forall\ 0<|\bm k|\leq 2N
	% \end{align}
	%are imposed to ensure the sublinear bound holds true, which is crucial for the use of matrix-valued Cartan's lemma (cf. Lemma \ref{mcl}), where
	%$$\ \kappa_1\in(0,1),\tilde C>1,\  N\sim N_1^{\tilde C^2}.$$
	 %If without intermediate scales of LDE, we may start with  
	%\begin{align}\label{n1rest}
	%N\sim \log^{\frac{1}{\rho}} \frac{1}{\varepsilon+\delta}.
	%\end{align}
	%Combining \eqref{wsecond} and \eqref{n1rest} leads to the excision of a set $\Omega'\subset\Omega$  with measure 
	%$$ {\rm meas}(\Omega')\lesssim e^{-\frac14\log^{\frac{\kappa_1\frac{3}{4}}{\tilde C^2\rho}} \frac{1}{\varepsilon+\delta}}.$$
%From $ {\rm meas}(\Omega)\sim \delta^b$, we have
%\begin{align}\label{rho12}
%{\kappa_1\frac{3}{4}}\geq\tilde C^2\rho.
%\end{align}
%However, to apply Lemma \ref{mcl}, we have to verify the condition (iii) there,  and thus  impose the restriction 
%$$\tilde C\rho>\frac{3}{4},$$
	%which is impossible since \eqref{rho12}. }
	\end{rem}
	
	\item For the large scales (i.e., $N>(\varepsilon+\delta)^{-c_1}$) LDE, we will apply matrix-valued Cartan's lemma  and semi-algebraic sets theory  of \cite{BGS02, Bou07}. Of particular importance is Bourgain's geometric lemma (cf. Lemma~\ref{Boulem}).  This step requires both the Diophantine condition and the so called weak second Melnikov's condition on $\bm\omega.$ 
	
	To be coherent, since the variation in $\bm \omega$ is only $O(\delta)$, the excised measure in $\Omega$ needs to be less than
	$\delta^b$. So we make a particular choice of the exponent  and impose the following Diophantine condition: 
	%More precisely, we denote 
	\begin{align}
\nonumber&\ \ \ {\rm DC}_{\bm \omega}(N)\\
\label{dcomg}&=\left\{\bm\omega\in\Omega:\ |\bm k\cdot\bm\omega|\geq e^{-N^{\rho^4}} \  {\rm for}\ \forall\ 0<|\bm k|\leq 100N^2\right\}.
\end{align}
Then we  have for $\varepsilon\leq \delta$, 
\begin{align} \label{meas}             
{\rm meas}\left(\Omega\setminus\left(\bigcap_{N>(\varepsilon+\delta)^{-c_1} }{\rm DC}_{\bm \omega}(N)\right)\right) \ll\delta^b.
\end{align}
%Note that because the variation in $\bm \omega$ is only $O(\delta)$, to be coherent, the upper bound in \eqref{meas} needs to be less than $\delta^b$.
%This is why we used the exponent $10^4b^4d^2\tau_{b+2}^2$ in the Diophantine condition on $\bm\omega.$  
%(The exponent in \eqref{exponent} is chosen to provide the upper bound in \eqref{meas}.)

The weak second Melnikov's condition is imposed in order to restrict the time $\bm k$-direction  where resonances could possibly occur. This greatly reduces the number 
of resonances, and paves the way for the application of Lemma~\ref{mcl}.
%which  is imposed so that the sublinear bound  on bad boxes along the $\Z^d\ni \bm k$-direction holds true 
%This leads to a sub-exponential bound in the scale $N$. 
\begin{rem} The condition on the size of $\varepsilon$ relative to $\delta$ stems again from the modulation in $\bm \omega$ being $O(\delta)$; while 
the excised measure depends on both $\varepsilon$ and $\delta$, as we will see in the proofs. This also ensures, in addition, that the amplitude-frequency map: 
$\bm a\to \bm\omega(\bm a)$ is non-degenerate in the nonlinear analysis, so that we may indeed vary $\bm \omega$ in the linear analysis.

\end{rem}
\end{itemize}

	{Before stating  the  LDT, let us summarize various parameters and the relation between  $\varepsilon, \delta$. 
	\begin{itemize}
	\item Throughout this paper, we first fix  constants $$\kappa_1=\kappa_1(d), \kappa_2=\kappa_2(d)\in(0,1)$$  given in   Bourgain's geometric  lemma  (cf. Lemma \ref{Boulem}). 
	\item The most important parameter $\rho>0$ is then fixed  so that it depends only on $b,d$  and satisfies 
	\begin{align}\label{rhodefn}
	\rho=c_2 \kappa_1\leq \frac{\kappa_1}{10^4(b+d)^2},
	\end{align}
	where $c_2=c_2(b,d)$ is given by Lemma \ref{lsclem1}.  
	\item  The  parameters $$0<c_1=c_1(b,d, K, V)<\frac{1}{100b}, \ C_1=C_1(b,d,K,V)>1$$   are  fixed, which appear in  the Diophantine estimates  (cf. Theorem \ref{clusthm}).  Also,   the parameter   $c_1$ is the power-law exponent   of the upper bound of intermediate scales.  % (cf. Lemma \ref{}). 
	
	\item  Fix   $$ \kappa=\frac12\min\left\{\vartheta, \ \frac{1}{10}\right\}\leq \frac{1}{20},$$
	where $\vartheta>0$ is an absolute constant  appeared in the coupling lemma of \cite{Liu22} (cf. Lemma \ref{Liulem1}). Then $\kappa$ (cf. \eqref{kappa}) describes the changes  of off-diagonal exponential decay rates of Green's functions along the multi-scale    iterations.

	\item For convenience,  we assume the decay rate of $S$ satisfies  $\gamma\in (\frac 12, 10)$.  It turns out  that   the $(\rho, \gamma_{\infty}, N)$-LDE   hold true for {\it all}  $N\geq \log^{3}\frac{1}{\varepsilon+\delta}$ and for some $\gamma_\infty\in[\frac\gamma 2, \gamma].$ The constant $C_2>0$ is  the power-law index in $S.$

	\item  With these parameters having been  fixed, we aim to find  $$\delta_0=\delta_0(b, d, K, V, \max\limits_{1\leq \ell\leq b}|\bm n_\ell|, \kappa_1, \kappa_2, \rho, c_1, C_2)>0,$$ so that, if $0<\varepsilon+\delta\leq \delta_0$, then  LDT  can be established. Since we will use  the multi-scale analysis method to prove LDT,   various conditions on  $\delta_0$  are required, and it is unclear a priori  if such restrictions are meaningful.  Fortunately, we  will  divide the  LDT proof  into 3 steps, and at each step, we can find a  desired  $\delta_\ell$  ($\ell=1,2,3$). Finally, it suffices to take  
	$\delta_0=\min\{\delta_1, \delta_2, \delta_3\}.$   
	\item As we will see later, we further assume $\delta\leq \log^{-1}\frac1\varepsilon$ when proving large scales LDE. 
\end{itemize}
		}

	Let $\rho, \kappa_1, \kappa_2, c_1, \gamma, C_2$ be fixed as  above.  We have  %The central theorem of our linear analysis  is 
	{\begin{thm}[{\bf LDT}]\label{ldtthm}
	There  is some 
	$$\delta_0=\delta_0(b, d, K, V, \max\limits_{1\leq \ell\leq b}|\bm n_\ell|, \kappa_1, \kappa_2, \rho, c_1, C_2)>0$$  so that the following holds true for $0<\varepsilon+\delta\leq \delta_0.$
	\begin{itemize}
	\item[(1)] For $\log^{3}\frac{1}{\varepsilon+\delta}< N\leq 10^{-\frac1\rho}\log^{\frac{1}{\rho}}\frac{1}{\varepsilon+\delta},$ the $(\rho, \gamma_N, N)$-LDE hold true for all  $(\bm \alpha, \bm\theta)\in[0,1]^{2d}$ (and thus for all $\bm\omega\in\Omega$) with  $\gamma_N\equiv\gamma_0=\gamma-\log^{-2}\frac{1}{\varepsilon+\delta}\in[\frac\gamma2, \gamma].$
	
	\item[(2)]  For $10^{-\frac1\rho}\log^{\frac{1}{\rho}}\frac{1}{\varepsilon+\delta}\leq N\leq (\varepsilon+\delta)^{-c_1},$  the $(\rho, \gamma_N, N)$-LDE hold true for $(\bm \alpha, \bm\theta)\in \mathcal W$ (cf. \eqref{ALset}) and all $\bm\omega\in\Omega$ (depending on $(\bm \alpha, \bm \theta)$)  with  $\gamma_N=\gamma_0-N^{-\vartheta}\in[\frac\gamma2,\gamma].$
	\item[(3)] Fix $(\bm \alpha, \bm\theta)\in \mathcal W $ (cf. \eqref{ALset}) and let  $\varepsilon\leq \delta\leq \log^{-1}\frac1\varepsilon$.  Let  $N>(\varepsilon+\delta)^{-c_1}$, $N_1\sim N^{\rho^2}$  and $N_2\sim N^{2\rho}$. Then there are some $\widetilde\Omega_N\subset \Omega$ (independent of $S$) satisfying ${\rm meas}(\Omega\setminus\widetilde\Omega_N)\leq e^{-\frac15 N^{\frac{3\kappa_1\rho^2}{4}}}$ and some $\gamma_N\in[\frac\gamma2, \gamma]$ so that the $(\rho, \gamma_N, N)$-LDE hold true for $\bm \omega\in\Omega_N:=\widetilde\Omega_N\cap {\rm DC}_{\bm\omega}(N)\cap\Omega_{N_2}$  with ${\rm DC}_{\bm\omega}(N)$ given by \eqref{dcomg} and %More precisely, 
	\begin{align*}
\widetilde\Omega_N=\bigcap_{1\leq \ell,\ell'\leq N^{4(b+d)}, \xi=\pm1, \xi'=\pm1, 0<|\bm k|\leq 2N^2}\widetilde\Omega_{\bm k,\ell,\ell',\xi,\xi'},
\end{align*}
where
\begin{align*}
\widetilde\Omega_{\bm k,\ell,\ell',\xi,\xi'}:=\left\{\bm\omega\in\Omega:\ |\bm k\cdot\bm \omega+\xi {\lambda_\ell}-\xi'{\lambda_{\ell'}}|>10e^{-\frac 14 N^{\frac{3\kappa_1\rho^2}{4}}}\right\},
\end{align*}
	%with  %for $\gamma_N$,  we have 
	$\gamma_N\geq \gamma_{N_1}-N^{-\kappa}$  %,\ N_1\sim N^{\rho^2}, N_2\sim N_1^{\frac\rho2}.$
and $\lambda_\ell= \lambda_\ell(\bm \alpha, \bm \theta)$ (${1\leq \ell \leq N^{4(b+d)}})$ are  all eigenvalues of the operators  $\mathcal H_{Q}(\bm \alpha, \bm\theta)$ for  all $Q\subset[-100N^2, 100N^2]^d$ satisfying   
$$Q\in\bigcup_{\tilde N\in [\frac14N^{\kappa_1\rho^2}, N^{\kappa_2\rho^2}]}\mathcal{ER}_{\Z^d}(\tilde N).$$ 
\end{itemize}
%with  $[\cdot]$ denoting  the integer part of a real number. 
	%Fix   $(\bm \alpha, \bm \theta)\in\mathcal{M}$ (cf. \eqref {ALset} and Theorem \ref{clusthm}).   Let $\rho< \frac{\kappa_1\frac{3}{4}^2}{30d}$ ($\kappa_1$ is given in Lemma~\ref{Boulem}). Then for every $N\geq N_0(b,d, \max\limits_{1\leq \ell\leq b}|\bm n_\ell|)\gg1$, there exists a semi-algebraic set   
	%$\widetilde\Omega_N\subset\Omega$, independent of $S$, of ${\rm meas}(\Omega\setminus\widetilde\Omega_N)\leq e^{-N^\zeta}$ and $\deg \widetilde\Omega_N\leq N^C$ for some $\zeta=\zeta(b,d)\in (0,1)$ and $C=C(b,d)>1,$ so that if we define 
%$\Omega_N$ to be 
%	\begin{align*}
%	\Omega_N&=\left\{
	%\begin{array}{ll}
%\Omega& {\rm if} \ N_0\leq N\leq (\varepsilon+\delta)^{-\frac{1}{10^3b^2d^2\tau_{b+2}^2}},\\
%{\rm DC}_{\bm \omega}(N)\cap\widetilde\Omega_N& {\rm if}\  N>(\varepsilon+\delta)^{-\frac{1}{10^3b^2d^2\tau_{b+2}^2}}, 
% |(k-e_{l'}, 1)|_2&{\rm if}\  n\notin \{n^{(l)}\}_{1\leq l\leq b},\ n'=n^{(l')},\  1\leq l'\leq b.
%\end{array}\right.
	%\end{align*}
	%then for $\bm\omega\in\Omega_N$, ${H}(\sigma)$ satisfies the 
	%There exist $0<c_1<c_2<c_3<c_4<1$  and $C_1>1$ depending only on $b,d$ with the following properties. 
	%Let   $N\geq \log^{C_1}|\varepsilon+\delta|$. Assume 
	%\begin{align}
	%|k\cdot\omega|\geq e^{-N^{c_4}}\ {\rm for}\ 0<|k|\leq 10N.
	%\end{align}
	%Then for $0<\varepsilon+\delta\ll1,$  $\mathcal{H}(\sigma)$ satisfies  
	%$(\rho, \frac{3}{4}, \frac{8}{9}, \gamma_{N},N)$-LDE  with $\gamma_{N}=\gamma-N^{-\kappa}$ for some  $\kappa=\kappa(b,d)>0$.
	\end{thm}}
	
	\begin {rem} \label{LDTrem1}
	%\begin{itemize}
	%\item 
	Note that the sets $\Omega_N$ are {\it independent of} $S$. This will be important for later nonlinear applications. Moreover,  we have    $\Omega_{N}\subset\Omega_{N-1}\subset\cdots\subset\Omega_{N_0}=\Omega$ and $\Omega_N$ is a semi-algebraic set of ${\rm deg}\ \Omega_N\leq N^{10(b+d)}$ and ${\rm meas}(\Omega_{N^{2\rho}}\setminus\Omega_N)\leq e^{-\frac12 N^{\rho^4}}.$ 
	%\item It is clear that the set $\Omega_N$ is basically stable under perturbations of  order $e^{- N^{\frac{3\kappa_1\rho^2}{4}}}$. More precisely, one can replace $\Omega_N$ with a larger set
%$$\Omega_N':=\bigcup_{\ell:\ I_\ell\cap \Omega_N\neq \emptyset}I_\ell,$$
%where the union runs over a partition of \ $\Omega$ of cubes (or intervals) of side length $e^{- N^{\frac{3\kappa_1\rho^2}{4}}}$. %This point is
%not an essential one, but will be useful in w below.
	%\end{itemize}
	\end {rem}

{\begin{rem}\label{LDTrem2}
One may observe that this LDT is for a fixed operator $H$;
%(i.e., $S$ does not change in the linear MSA steps), 
while the linearized operator is varying in a Newton scheme. This issue, however, can be remedied by using that the Green's function estimates in LDE (at scale $N$) are stable under perturbations 
%(on $S$)  
of order  $e^{-N^2}$ using Lemma \ref{lwlem}. The sup-exponentially decaying error (cf. \eqref{tq1}) in the Newton scheme can ensure this.  More precisely, let 
\begin{align*}
&\gamma\geq \gamma_1\geq \gamma_2\geq \cdots\geq \frac\gamma2,\\
&\gamma\geq \gamma_1'\geq \gamma_2'\geq\cdots\geq \frac\gamma2,\\ 
&\gamma_{\ell}'\leq \min\{\gamma_{[\ell^{\rho^2}]+1}, \gamma_{[\ell^{\rho^2}]+1}'\}-\ell^{-\kappa}.
\end{align*}
We can replace  $H$ in LDT with  a sequence of operators  $H^{(\ell)}(\sigma)=D{(\sigma)}+\delta S_\ell$ ($\ell\geq1$) satisfying (we hide the dependence on $\xi,\xi'\in\{\pm\}$)
\begin{align*}
|S_\ell((\bm k ,\bm n); (\bm k', \bm n'))|
&\leq C_2(1+|\bm k-\bm k'|)^{C_2}e^{-\gamma_\ell|\bm k-\bm k'|-\gamma_{\ell}|\bm n|}\delta_{\bm n, \bm n'},\\
\|H^{(\ell+1)}-H^{(\ell)}\|&\leq e^{-\ell^{\frac{2}{\rho^2}}}. 
%&\gamma'_\ell\geq\gamma_{\ell^{\rho^2}}+\ell^{-\frac{1}{10}},
\end{align*}
  It is remarkable that in our LDT,  $\Omega_N$ does not depend on $S$.  So with this modification,  the conclusion of LDT becomes   for $\bm \omega\in\Omega_N,$  $H^{(\ell)}$ has $(\rho, \gamma_N', N)$-LDE  for all $\ell\geq N\geq \log^3{\frac{1}{\varepsilon+\delta}}$ ($\Sigma_N$ depending only on $S_\ell$, $1\leq \ell  \leq N$). 
%where for each $N\geq \log^3\frac{1}{\varepsilon+\delta}$,   $[\frac\gamma 2, \gamma]\ni\gamma_N'\geq \min\{\gamma_{N^{\rho^2}}', \gamma_{N}\}-N^{-\kappa}.$ 

Indeed, assume we have established the LDE at scale $N_1=N^{\rho^2}$ for all $H^{(\ell)}$ with $\ell\geq N_1$. Then it suffices to establish $(\rho, \gamma_N', N)$-LDE for $H^{(N)}$,  since for all $\ell>N$, $$\|H^{(\ell)}-H^{(N)}\|\leq \sum\limits_{\ell\geq N}e^{-\ell^{\frac{2}{\rho^2}}}<e^{-\frac{N^{\frac{2}{\rho^2}}}{2}}\ll e^{-100N},$$ 
and  $(\rho, \gamma_N', N)$-LDE for $H^{(\ell)}$ can be obtained via the perturbation argument (cf. Lemma \ref{lwlem}).      
Next,  we fix $S_\ell\equiv S_{N_1}$ for $\ell\in[N_1, N]$.  Similar to the proof of Lemma \ref{lsclem}, we can show  using the modified resolvent identities Lemmas \ref{Liulem1}, \ref{JLSlem1}, \ref{JLSlem2}  (e.g.,  taking account of the power-law factor in the estimate of    $S_\ell$) that $H^{(N_1)}$ satisfies $(\rho, \gamma_N'', N)$-LDE with 
\begin{align*}
\gamma_N''\geq\min\{\gamma_{N_1}, \gamma_{N_1}'\}-N^{-\kappa}\geq \gamma_{N}'.
%&\geq \gamma_{\sqrt{N}}-\log^{-1}N-N^{-\kappa},\\
%&\geq  \gamma_{\sqrt{N}}'-2\log^{-1}N,
%&\geq \gamma'_0-\sum_{\ell} N^{-}
\end{align*}
Finally, using again the perturbation argument  Lemma \ref{lwlem}  and $$\|H^{(N_1)}-H^{(N)}\|\leq e^{-\frac{N_1^{\frac{2}{\rho^2}}}{2}}\leq e^{-\frac{N^2}{5}}\ll e^{-100N},$$
we get that $H^{(N)}$ satisfies $(\rho, \gamma_N', N)$-LDE by the assumptions on  $\gamma_\ell,\gamma_\ell'$.

%which combined with  the iteration similar to \eqref{gamitr}  will show  $\gamma'\in[\frac\gamma2,\gamma].$ 

% The key point is that

%The minor modification of our proof can be 
%The key point is that  in the multi-scales analysis scheme,   Green's function estimates  at scale $N$ are  obtained mainly based on  that  at a very small scale (e.g., $N_1=N^{\rho^2}\ll N$),  and a perturbation of  $S$ of size $e^{-10N_1}$   preserves essentially  the relevant estimates. 
%For more details, see sect.~V, F. Step 3 and the proof of Proposition~3.2 in \cite{KLW23}.
\end{rem}}
\begin{rem}
In the large scales case, we assume  $\delta\in[\varepsilon, \log^{-1}\frac{1}{\varepsilon}]$. The condition $\delta\geq \varepsilon$ is necessary since \eqref{meas}
as mentioned above. The restriction $\delta\leq \log^{-1}\frac{1}{\varepsilon}$, however,  is due to the fact that we need to  use Bourgain's geometric lemma  in the  form   of Lemma \ref{Boulem}  to handle  the large $\bm |\bm n|$ regime estimates.     In fact,  if we assume  that  $\mathcal W$ (cf. \eqref{ALset})  satisfies    ${\rm meas}([0,1]^{2d}\setminus\mathcal W)=o(1)$ (rather than the quantitative version   ${\rm meas}([0,1]^{2d}\setminus\mathcal W)=O(\log^{-c}\frac{1}{\varepsilon+\delta})$ in the present work), we can remove the condition $\delta\leq \log^{-1}\frac{1}{\varepsilon}$ (cf. Remark \ref{remn1} for details). 
\end{rem}

This  LDT  will be proved in detail in the following three sections.
	
\subsubsection{The small scales LDE}

This section is devoted to the proof of LDE for small scales, i.e., 
\begin{align*}
\log^{3}\frac{1}{\varepsilon+\delta}\leq N\leq 10^{-\frac1\rho}\log^{\frac{1}{\rho}}\frac{1}{\varepsilon+\delta}. 
\end{align*}
 In this case,  we only use the standard Neumann series argument. As a result, we can establish LDE  for all $\bm\omega\in\Omega.$ We have 

\begin{lem}\label{inilem}
Let  
\begin{align*}
&\ \ \ \Sigma_N\\
&=\bigcup_{\Lambda\in (\bm 0,\bm n)+\mathcal{ER}_{\bm 0}(N),\ |\bm n|\leq 10N}\left\{\sigma\in\R:\ \min_{ (\bm k,\bm n, \xi)\in\Lambda}|\xi(\sigma+\bm k\cdot\bm \omega)+\mu_{\bm n}|\leq e^{-2N^{\rho}}\right\}.
\end{align*}
Let $0<\rho<\frac{1}{10^4}$. Then there is some $\delta_1=\delta_1(b,d,C_2, \rho)>0$ so that if $0<\varepsilon+\delta\leq \delta_1,$   for $$\log^{3}\frac{1}{\varepsilon+\delta}\leq N\leq  10^{-\frac1\rho} \log^{\frac{1}{\rho}}\frac{1}{\varepsilon+\delta},$$
we have 
\begin{align}\label{inieq1}
{\rm meas}(\Sigma_N)\leq e^{-N^{\rho}}.
\end{align}
Moreover, if $\sigma\notin \Sigma_N$, we have  for all  $\Lambda\in (\bm 0,\bm n)+\mathcal{ER}_{\bm 0}(N)$ with $|\bm n|\leq 10N$,
\begin{align}
\label{inieq2}\|G_{\Lambda}(\sigma)\|&\leq e^{N^{\frac34}},\\
\label{inieq3}|G_{\Lambda}(\sigma)((\bm x,\xi);(\bm x',\xi'))|&\leq e^{-\gamma_0|\bm x-\bm x'|}\ {\rm for}\ |\bm x-\bm x'|\geq N^{\frac89}, 
\end{align}
%and moreover, for $|\bm x-\bm x'|\geq N^{\frac89},$
%where $\gamma_0=\gamma-\frac14\log \frac{1}{\varepsilon+\delta}$ and $\gamma_0\in[\frac\gamma 2,\gamma]$ if $0<\varepsilon+\delta   \ll1.$
where 
\begin{align}\label{gamma0}
\gamma_0=\gamma-\log ^{-2}\frac{1}{\varepsilon+\delta}\in[\frac\gamma2,\gamma]. 
\end{align}
\end{lem}
\begin{rem}
In this lemma,  we make no restrictions  on $\bm\omega.$
\end{rem}

\begin{proof}
We first establish the measure bound \eqref{inieq1}. %We let $N_\star$ satisfy 
%$$e^{-2N_\star^{\rho}}=(\varepsilon+\delta)^{\frac{1}{8b}}.$$
%Then $N_\star=(\frac{1}{16b})^{\frac{1}{\rho}}\log^{\frac{1}{\rho}}|\varepsilon+\delta|.$ 
We have 
\begin{align*}
{\rm meas}(\Sigma_N)&\leq C(b,d)N^{C(b,d)}e^{-2N^{\rho}}
%&\leq C(b,d)N^{C(b,d)}e^{-2N_\star^{\rho}}\\
\leq e^{-N^{\rho}},
\end{align*}
where in the last inequality we use   $0<\varepsilon+\delta\leq \delta_1^{(1)}(b,d,\rho)\ll1$  to ensure  $N\geq \log^{3}\frac{1}{\varepsilon+\delta}\geq N_0(b,d,\rho)\gg1.$ 

Next, we let $\sigma\notin \Sigma_N$ and fix any $\Lambda\in(\bm 0, \bm n)+\mathcal{ER}_{\bm 0}(N)$ with $|\bm n|\leq 10N$.  Let $A=D_{\Lambda}(\sigma), B=\delta S_\Lambda$. We have 
$$ |A^{-1}(\bm x, \bm x')|\leq e^{N^\rho}\delta_{\bm x, \bm x'},\ |B(\bm x, \bm x')|\leq C_2\delta (1+|\bm x-\bm x'|)^{C_2}e^{-\gamma |\bm x-\bm x'|}.$$
Note that 
\begin{align*}
CN^{2(b+d)+C_2} \delta e^{N^\rho}&\leq C(b,d, C_2)N^{2(b+d)+C_2}(\varepsilon+\delta)^{\frac{9}{10}}\\
&\leq C(b,d, C_2)10^{-\frac{2(b+d)+C_2}{\rho}}(\varepsilon+\delta)^{\frac{9}{10}}\log^{-\frac{2(b+d)+C_2}{\rho}}\frac{1}{\varepsilon+\delta}\\
&\leq \frac{1}{2}
\end{align*}
assuming  $0<\varepsilon+\delta\leq \delta_1^{(2)}(b,d, C_2, \rho)\ll1.$
So we can apply  Lemma \ref{lwlem} with $A=D_{\Lambda}(\sigma), B=\delta S_\Lambda, \epsilon_2=C_2\delta, \epsilon_1=e^{-N^\rho}, M=0, X=\Lambda$ to obtain (we hide the dependence on $\xi, \xi'$) 
\begin{align*}
\|G_\Lambda(\sigma)\|&\leq 2e^{N^\rho}<e^{N^{\frac34}}\ ({\rm since}\ 0<\rho<\frac{1}{10^4}),\\
|G_\Lambda(\sigma)(\bm x, \bm x')|&\leq e^{N^{\rho}}e^{-\gamma|\bm x-\bm x'|}. 
\end{align*} 
Then for $|\bm x-\bm x'|\geq N^{\frac89},$
\begin{align*}
|G_\Lambda(\sigma)(\bm x, \bm x')|&\leq  e^{-|\bm x-\bm x'|(\gamma-\frac{N^\rho}{|\bm x-\bm x'|})}\\
%&\leq e^{-|\bm x-\bm x'|(\gamma-N^{-\frac89-\rho})}\\
%&\leq e^{-|\bm x-\bm x'|(\gamma-N^{-\frac89+\rho})}\\
&\leq e^{-|\bm x-\bm x'|(\gamma-N^{-\frac{7}{8}})}\ \ ({\rm since}\ 0<\rho<\frac{1}{10^4})\\
&:=e^{-\gamma'|\bm x-\bm x'|},
\end{align*}
which implies 
$$\gamma'>\gamma_0= \gamma-\log^{-2}\frac{1}{\varepsilon+\delta}.$$
Then both  \eqref{inieq2} and \eqref{inieq3}  hold true   if we assume that  $$0<\varepsilon+\delta\leq \delta_1=\min\{\delta_1^{(1)}(b,d,C_2, \rho), \delta_1^{(2)}(b,d,C_2, \rho)\},$$ and $\sigma\notin\Sigma_N.$

%from the standard  Neumann series  argument  Lemma \ref{lwlem}(cf. the proof of Theorem~4.3 in \cite{JLS20}) since $e^{-2N^{\rho}}\geq (\varepsilon+\delta)^{\frac 15}$ if  $N\leq 10^{-\frac1\rho}\log^{\frac{1}{\rho}}\frac{1}{\varepsilon+\delta}$. 

\end{proof}

\subsubsection{The intermediate scales LDE}
This section and the next deal with intermediate scales LDE.  We will show that the LDE hold  true for all $\bm\omega\in\Omega$  and scales in the range
$$10^{-\frac1\rho}\log^{\frac{1}{\rho}}\frac{1}{\varepsilon+\delta}\leq N\leq (\varepsilon+\delta)^{-c_1},$$  
where $c_1$ is given by Theorem \ref{clusthm}. 
In this case,  the analysis in Section~\ref{NRomega0} becomes essential.  %This part contains one of the central new ingredients of present work.  %So we assume the conditions of Theorem \ref{} are satisfied and set 
%\begin{align}
%L=2e^{-\log^{1/2}|\varepsilon+\delta|},\ \eta=(\varepsilon+\delta)^{\frac{1}{8b}}. 
%\end{align}
%This then yields the existence of $\mathcal{M}\subset [2, 3]$ so that 
%\begin{align}
%{\rm meas}([2,3]\setminus \mathcal{M})&\leq L^{50db^2}\eta^{\frac{1}{b+2}}\leq (\varepsilon+\delta)^{\frac{1}{10b(b+2)}}
%\end{align}
%assuming $0<\varepsilon+\delta\leq \delta_0(b,d)\ll1.$
\begin{lem}\label{intlem}
%Let  $\alpha\in{\rm DC}_\nu\subset {\rm DC}(L;L^{-3d})$ and let $\theta_0$,  $\mathcal{M}\subset [2, 3]$ be given by Theorem \ref{clusthm} with 
%$$ \eta=(\varepsilon+\delta)^{\frac{1}{8b}}, \ L=(\varepsilon+\delta)^{-\frac{1}{10^4db^4}}.$$
Let $0<\rho<\frac{1}{10^4}$. There is some $\delta_2=\delta_2(b,d, K, V, \rho, c_1, C_2)>0$ so that the following holds for $0<\varepsilon+\delta\leq\delta_2.$
Assume $(\bm \alpha, \bm\theta)\in \mathcal{M}$ (cf. Theorem \ref{clusthm})  and fix  $$ 10^{-\frac1\rho}\log^{\frac{1}{\rho}}\frac{1}{\varepsilon+\delta}\leq N\leq (\varepsilon+\delta)^{-c_1}.$$ Then for each $\bm\omega\in\Omega$, 
there exists a set $\Sigma_N\subset\R$ with 
\begin{align*}
{\rm meas}(\Sigma_N)\leq (\varepsilon+\delta)^{-\frac{9}{10}}e^{-\frac{1}{2b}N^{\frac{3}{8}}}<e^{-N^{\rho}},
\end{align*}
so that if $\sigma \notin\Sigma_N$, then   for all   $\Lambda\in (\bm 0, \bm n)+\mathcal{ER}_{\bm 0}(N)$ with  $|\bm n|\leq 10N$,
\begin{align*}
\|G_{\Lambda}(\sigma)\|&\leq e^{N^{\frac34}},\\
|G_{\Lambda}(\sigma)((\bm x,\xi);(\bm x',\xi'))|&\leq e^{-\gamma_N|\bm x-\bm x'|}\ {\rm for}\ |\bm x-\bm x'|\geq N^{\frac89},
\end{align*}
where 
$$\gamma_N=\gamma_0-N^{-\vartheta}\in[\frac\gamma2,\gamma],$$
with $\gamma_0$  given by \eqref{gamma0} and $\vartheta>0$ by  Lemma \ref{Liulem1}. 
\end{lem}
\begin{rem}
Note that this lemma also holds for all  $\bm\omega\in\Omega$. 
\end{rem}
\begin{proof}
The proof is based on clustering property (cf. (2) of Theorem \ref{clusthm}), Schur's completement argument,  the Rouch\'e  theorem and the iteration of  resolvent identities. So we first choose $0<\varepsilon+\delta\leq c(b,d,K,V, \max\limits_{1\leq \ell\leq b}|\bm n_\ell|)\ll1$ so that Theorem \ref{clusthm}  holds true. Then  the proof   can be decomposed into two steps. \\

{\noindent{\bf Step 1.}  
We  define %for $ (\frac{1}{16b})^{\frac{1}{b_1}}\log^{\frac{1}{b_1}}|\varepsilon+\delta|\leq N\leq L,$
\begin{align}\label{intsgmn}
\Sigma_N=\bigcup_{\sqrt{N}\leq N'\leq N,\ \Lambda\in (\bm 0, \bm n)+\mathcal{ER}_{\bm 0}(N'),\ |\bm n|\leq 10N'}\Sigma_\Lambda,%\left\ %{\sigma\in\R:\  \|G_\Lambda(\sigma)\|\geq  e^{|\Lambda|^{\rho_2}}\right\}
\end{align}
%We are going to show in this step %for each $\Lambda\in \mathcal{ER}_0(N')$, 
%\begin{align*}
%{\rm meas}\left(\Sigma_\Lambda\right)\leq 4b (\varepsilon+\delta)^{-\frac{1}{3b}}e^{-\frac{1}{2b}{|\Lambda|}^{\rho_2}}.
%\end{align*}
%where 
where
\begin{align*}
\Sigma_\Lambda=\left\{\sigma\in\R:\  \|G_\Lambda(\sigma)\|\geq  e^{({\rm diam}\ \Lambda)^{\frac34}}\right\}.
\end{align*}
We first show that 
\begin{align}\label{set}
\Sigma_\Lambda\subset \bigcup_{\xi=\pm 1, \bm i=(\bm k, \bm n)\in\Lambda} I_{\xi, \bm  i},
\end{align}
where 
\begin{align*}
I_{\xi, \bm i}=\left\{\sigma\in\R:\ |\xi(\sigma+\bm k\cdot\bm \omega^{(0)})+\mu_{\bm n}|\leq \frac{1}{100}(\varepsilon+\delta)^{\frac{1}{7b}}\right\}.
\end{align*}
In fact, if $\sigma\notin  \bigcup\limits_{(\xi,\bm i)\in\Lambda} I_{\xi, \bm i},$ then for each $\bm \omega\in\Omega,$
\begin{align*}
&\ \ \ \ \inf_{(\xi,\bm i)\in\Lambda} |\xi(\sigma+\bm k\cdot\bm \omega)+\mu_{\bm n}|\\
&\geq \inf_{(\xi,\bm i)\in\Lambda} |\xi(\sigma+\bm k\cdot\bm\omega^{(0)})+\mu_{\bm n}|-\sup_{\bm k\in\Lambda}|\bm k\cdot(\bm \omega-\bm \omega^{(0)})|\\
%&\geq \frac{1}{100}(\varepsilon+\delta)^{\frac{1}{8b}}-C(\varepsilon+\delta)^{-\frac{1}{10^4db^4}}\delta\\
&\geq \frac{1}{200}(\varepsilon+\delta)^{\frac{1}{7b}}.%\geq e^{-\frac12N^{\rho_2}}
\end{align*}
Since  $N\geq 10^{- \frac1\rho} \log^{\frac{1}{\rho}} \frac{1}{\varepsilon+\delta}$, using Neumann series argument  implies 
\begin{align*}
\|G_\Lambda(\sigma)\|\leq C(\varepsilon+\delta)^{-\frac{2}{7b}}<e^{({\rm diam}\ \Lambda)^{\frac34}}
\end{align*}
when  $\sigma\notin  \bigcup\limits_{(\xi,\bm i)\in\Lambda} I_{\xi, \bm i}.$ This proves \eqref{set} and 
\begin{align*}
\Sigma_\Lambda=\bigcup_{(\xi,\bm i)\in\Lambda}(\Sigma_\Lambda\cap I_{\xi,\bm i}).
\end{align*}

In the following,  we estimate ${\rm meas}(\Sigma_\Lambda\cap I_{\xi, \bm i}).$ For this purpose, we  fix any $\xi^*\in\{+,-\}$, $\bm i^*=(\bm k^*, \bm n^*)$ and denote 
$$ I^*=I_{\xi^*, \bm i^*}=\left[\sigma^*-\frac{1}{100}(\varepsilon+\delta)^{\frac{1}{7b}}, \sigma^*+\frac{1}{100}(\varepsilon+\delta)^{\frac{1}{7b}}\right],$$
where 
\begin{align}\label{sgm*}
\sigma^*=-\bm k^*\cdot\bm \omega^{(0)}-{\xi^*}\mu_{\bm n^{*}}.
\end{align}
We will study $G_\Lambda(\sigma)$ for $\sigma\in I^*.$ From (2) of Theorem \ref{clusthm}, we have for every  $\sigma^*\in\R$,  there exists $B_*=B_*^+\cup B_*^-\subset \Lambda$ with $\# B_*^{\pm}\leq b$ so that
\begin{align*}
\inf_{(\xi,\bm i)\in\Lambda\setminus B_*}|\xi(\sigma^*+\bm k\cdot\bm\omega^{(0)})+\mu_{\bm n}|\geq \frac{(\varepsilon+\delta)^{\frac{1}{8b}}}{4}.
\end{align*}
Similar to the proof of \eqref{set},  if (we denote $\mathbb D_r(z):=\{w\in\C:\ |w-z|\leq r\}$)
$z\in \mathbb{D}_{(\varepsilon+\delta)^{\frac{1}{8b}}/10}(\sigma^*),$ %=\{z\in\C:\ |z-\sigma^*|\leq (\varepsilon+\delta)^{\frac{1}{8b}}/10\},$$
then
\begin{align*}
\inf_{(\xi,\bm i)\in\Lambda\setminus B_*}|\xi(z+\bm k\cdot\bm \omega)+\mu_{\bm n}|\geq \frac{(\varepsilon+\delta)^{\frac{1}{8b}}}{8}.
\end{align*}
%In fact,  we have
%\begin{align*}
%&\ \ \ \ \inf_{\xi=\pm 1, i=(k,n)\in\Lambda\setminus B_*} |\xi(z+k\cdot\omega)+\mu_n|\\
%&\geq \inf_{\xi=\pm 1, i=(k,n)\in\Lambda\setminus B_*} |\xi(\sigma^*+k\cdot\omega^{(0)})+\mu_n|-\sup_{k\in\Lambda\setminus B_*}|k\cdot(\omega-\omega^{(0)})|-|z-\sigma^*|\\
%&\geq \frac{1}{2}(\varepsilon+\delta)^{\frac{1}{8b}}-C(\varepsilon+\delta)^{-\frac{1}{10^4db^4}}\delta-\frac{1}{4}(\varepsilon+\delta)^{\frac{1}{8b}}\\
%&\geq \frac{1}{5}(\varepsilon+\delta)^{\frac{1}{8b}}.%\geq e^{-\frac12N^{\rho_2}}
%\end{align*}
So we have an analytic extension of ${H}_\Lambda(z)$ to $z\in\mathbb{D}_{(\varepsilon+\delta)^{\frac{1}{8b}}/10}(\sigma^*)$  and 
\begin{align*}
\sup_{(\bm x, \xi),(\bm x', \xi')\in \Lambda,\ z\in\mathbb{D}_{(\varepsilon+\delta)^{\frac{1}{8b}}/10}(\sigma^*)}|{D}_\Lambda(z)((\bm x,\xi); (\bm x', \xi'))|\leq C(b)(\varepsilon+\delta)^{-c_1}. 
\end{align*}
 More importantly, if we denote $\Lambda^c=\Lambda\setminus B_*$, then by the  Neumann series argument,  we have  for all $z\in \mathbb{D}_{(\varepsilon+\delta)^{\frac{1}{8b}}/10}(\sigma^*),$
\begin{align*}
\|G_{\Lambda^c}(z)\|\leq (\varepsilon+\delta)^{-\frac{1}{4b}}.
\end{align*}
In the following,  we estimate $G_{\Lambda}(z)$ using the Schur complement reduction.  We define 
\begin{align*}
S(z)=R_{B_*}{H}(z)R_{B_*}-R_{B_*}{H}(z)R_{\Lambda^c}G_{\Lambda^c}(z)R_{\Lambda^c}{H}(z)R_{B_*}.
\end{align*}
Obviously, we have
\begin{align}\label{schureq1}
R_{B_*}{H}(z)R_{B_*}%&=R_{B_*}{D}(z)R_{B_*}+R_{B_*}(\varepsilon\Delta+\delta{T}_\phi )R_{B_*}\\
&=R_{B_*}{D}(z)R_{B_*}+O(\varepsilon+\delta) 
\end{align}
and 
\begin{align}\label{schureq2}
R_{B_*}{H}(z)R_{\Lambda^c}G_{\Lambda^c}(z)R_{\Lambda^c}{H}(z)R_{B_*}=O((\varepsilon+\delta)^{2-\frac{1}{4b}}).
\end{align}
Now we consider $\det S(z)$. Combining \eqref{schureq1} and \eqref{schureq2} yields 
\begin{align*}
&\ \ \ \det S(z)\\
&= \prod_{ \bm i=(\bm k, \bm n)\in B_*^+}(z+\bm k\cdot\bm\omega+\mu_{\bm n})\prod_{\bm i=(\bm k,\bm n)\in B_*^-}(-z-\bm k\cdot\bm\omega+\mu_{\bm n})\\
&\ \ \ +O((\varepsilon+\delta)^{3/4})\\
&=(-1)^{\#B_*^-}\prod_{\ell=1}^{\# B_*}(z-\sigma_\ell)+O((\varepsilon+\delta)^{3/4}),
\end{align*}
where $\sigma_\ell\in\{-\bm k\cdot \bm \omega- \mu_{\bm n}\}_{(\bm k, \bm n)\in B_*^+}\cup\{-\bm k\cdot \bm \omega+\mu_{\bm n}\}_{(\bm k, \bm n)\in B_*^-}.$ 
Recalling the definition of $B_*$, we must have 
\begin{align*}
|z-\sigma_\ell|\leq \frac45(\varepsilon+\delta)^{\frac{1}{8b}}.
\end{align*}
From \eqref{sgm*}, there exists at least one $\sigma_\ell$ so that $|\sigma_\ell-\sigma^*|=O((\varepsilon+\delta)^{1-c_1}).$
%\begin{align}
%|\sigma^*-\sigma_l|\leq \frac14(\varepsilon+\delta)^{\frac{1}{8b}}.
%\end{align}
At this stage,  we need a useful result. 
\begin{lem}\label{pigelem}
{ Let $\varphi=\left(\frac{14}{15}\right)^{\frac{1}{5b}}<1$. 
There is $0\leq \ell_*\leq 5b-1$ so that
\begin{align*}
\left(\mathbb{A}_{\ell_*+1}\setminus\mathbb{A}_{\ell_*}\right)\cap \{\sigma_\ell\}_{1\leq \ell\leq \# B_*}=\emptyset, 
\end{align*}
where for $0\leq \ell\leq 5b$,
\begin{align*}
\mathbb{A}_\ell=\mathbb{D}_{(\varepsilon+\delta)^{\frac{\varphi^{\ell}}{7b}}}(\sigma^*).
\end{align*}
}
\end{lem}
\begin{proof}[Proof of Lemma \ref{pigelem}]
It suffices to note that 
\begin{align*}
\mathbb{A}_{5b}\setminus\mathbb{A}_{0}=\bigcup_{\ell=0}^{5b-1}\left(\mathbb{A}_{\ell+1}\setminus \mathbb{A}_{\ell}\right). 
\end{align*}
The proof then follows from the pigeonhole principle since $\#B_*\leq 2b$. 
\end{proof}

Let $r_*=(\varepsilon+\delta)^{\frac{\varphi^{\ell_*}}{7b}}$, where $\varphi$ and $\ell_*$ are defined in Lemma \ref{pigelem}. From Lemma \ref{pigelem}, we can restrict our considerations on  $\mathbb{A}_*$ with 
\begin{align*}
\mathbb{A}_*=\mathbb{D}_{\frac{r_*+r_*^\varphi}{2}}(\sigma^*). 
\end{align*}
Then
\begin{align*}
\mathbb{D}_{(\varepsilon+\delta)^{\frac{1}{7b}}}(\sigma^*)\subset\mathbb{A}_{\ell_*}\subset\mathbb{A}_*\subset \mathbb{D}_{(\varepsilon+\delta)^{\frac{1}{7.5b}}}(\sigma^*).
\end{align*}
Let 
\begin{align*}
\mathcal K=\{1\leq \ell\leq \# B_*:\  \sigma_\ell\in \mathbb{D}_{10r_*}(\sigma^*)\}.
\end{align*}
We have $\mathcal K\neq \emptyset$ and for $\ell\not\in \mathcal K$,
\begin{align*}
\sigma_\ell\notin\mathbb{A}_*,\ {\rm dist}(\sigma_\ell,\partial \mathbb{D}_{20r_*}(\sigma^*))\geq{\rm dist}(\sigma_\ell, \partial\mathbb{A}_*)\geq \frac{r_*^\varphi}{4}\gg(\varepsilon+\delta)^{\frac{1}{7b}}.
\end{align*}
As a result, we can write for $z\in\mathbb{A}_*,$
\begin{align}\label{Rouceq1}
\frac {\det S(z)}{\prod_{\ell\notin \mathcal K}(z-\sigma_\ell)}=\pm{\prod_{\ell\in \mathcal K}(z-\sigma_\ell)}+O((\varepsilon+\delta)^{\frac13}).
\end{align}
It holds that 
\begin{align*}
{\prod_{\ell\in \mathcal K}|z-\sigma_\ell|}\big|_{\partial \mathbb{D}_{20r_*}(\sigma^*)}\geq (\varepsilon+\delta)^{\frac{\#\mathcal K}{7b}}\geq (\varepsilon+\delta)^{2/7}\gg (\varepsilon+\delta)^{\frac13}.
\end{align*}
Applying Rouch\'e  theorem shows that  the function $\frac {\det S(z)}{\prod_{\ell\notin \mathcal K}(z-\sigma_\ell)}$ has exactly $\#\mathcal K$ zeros (denoted by $z_\ell,\ \ell\in \mathcal K$)  in $ \mathbb{D}_{20r_*}(\sigma^*).$   From  $\sigma_\ell\in \mathbb{D}_{10r_*}(\sigma^*)$, we also get
\begin{align}\label{Rouceq2}
|\sigma_\ell-z_\ell|\le 30r_*.
\end{align}
From similar analysis, we can obtain that \eqref{Rouceq1} also has exactly $\#\mathcal{K}$ zeros in $\mathbb{A}_*$, which implies that $\{z_\ell\}_{\ell\in\mathcal{K}}$ are all zeros of \eqref{Rouceq1} in $\mathbb{A}_*$. So on $\mathbb{A}_*$, we have 
\begin{align}\label{Rouceq3}
\frac {\det S(z)}{\prod_{\ell\notin \mathcal K}(z-\sigma_\ell)}=g(z)\prod_{\ell\in \mathcal K}(z-z_\ell),
\end{align}
where $g$ is analytic on $\mathbb{A}_*$ with $\inf\limits_{z\in\mathbb{A}_*}|g(z)|>0.$ Using \eqref{Rouceq2} and \eqref{Rouceq3}, we have
\begin{align*}
g\big|_{\partial \mathbb{A}_*}=1+O((\varepsilon+\delta)^{\frac{1}{21}}).
\end{align*}
Using the maximum principle then leads to 
\begin{align*}
g\big|_{ \mathbb{A}_*}=1+O((\varepsilon+\delta)^{\frac{1}{21}}).
\end{align*}
We have established on $\mathbb{A}_*,$
\begin{align*}
\det S(z)=g(z)\prod_{\ell\notin \mathcal K}(z-\sigma_\ell)\prod_{\ell\in \mathcal K}(z-z_\ell).
\end{align*}
In particular, one has for all $z\in\mathbb{A}_*,$
\begin{align*}
|\det S(z)|&\geq \frac 45 (\varepsilon+\delta)^{\frac{\#B_*-\#\mathcal K}{7b}}\prod_{\ell\in \mathcal K}|z-z_\ell|\\
&\geq  (\varepsilon+\delta)^{\frac{2}{7}}\prod_{\ell\in \mathcal K}|z-z_\ell|
\end{align*}
assuming $0<\varepsilon+\delta\leq c(b)\ll1.$ Thus combining Hadamard's inequality and Cramer's rule implies for $z\in\mathbb{A}_*,$
\begin{align*}
\|S^{-1}(z)\|&\leq C(b)    (\varepsilon+\delta)^{-{2bc_1}}(\varepsilon+\delta)^{-\frac{2}{7}}\prod_{\ell\in \mathcal K}|z-z_\ell|^{-1}\\
&\leq (\varepsilon+\delta)^{-\frac{1}{3}}\prod_{\ell\in \mathcal K}|z-z_\ell|^{-1}\ ({\rm since}\  c_1<\frac{1}{100b}).
\end{align*}
Applying the Schur's complement argument   shows that for $z\in \mathbb{A}_*,$
\begin{align*}
\|G_\Lambda(z)\|&\leq 4(1+\|G_{\Lambda^c}(z)\|)^2(1+\|S^{-1}(z)\|)\\
&\leq (\varepsilon+\delta)^{-\frac{7}{8}}\prod_{\ell\in \mathcal K}|z-z_\ell|^{-1}.
\end{align*}
Recall that $I^*\subset \mathbb{A}_*$. So for $\sigma\in I^*,$ we have
\begin{align*}
\|G_\Lambda(\sigma)\|&\leq (\varepsilon+\delta)^{-\frac{7}{8}}\prod_{\ell\in \mathcal K}|\sigma-z_\ell|^{-1} 
\end{align*}
and 
\begin{align*}
\Sigma_\Lambda\cap I^*\subset \bigcup_{\ell\in \mathcal K}\left\{\sigma\in I^*:\ |\sigma-z_\ell|\leq (\varepsilon+\delta)^{-\frac{7}{8}}e^{-\frac{1}{2b}({\rm diam}\ \Lambda)^{\frac34}}\right\}.
\end{align*}
This implies 
\begin{align*}
{\rm meas}(\Sigma_\Lambda\cap I^*)\leq C(\varepsilon+\delta)^{-\frac{7}{8}}e^{-\frac{1}{2b}({\rm diam}\ \Lambda)^{\frac34}},
\end{align*}
and subsequently
\begin{align*}
{\rm meas}(\Sigma_\Lambda)\leq (\varepsilon+\delta)^{-\frac{8}{9}}e^{-\frac{1}{2b}({\rm diam}\ \Lambda)^{\frac34}},
\end{align*}
which combined with \eqref{intsgmn} shows
\begin{align*}
{\rm meas}(\Sigma_N)\leq (\varepsilon+\delta)^{-\frac{9}{10}}e^{-\frac{1}{2b}N^{\frac{3}{8}}}.
\end{align*}}\\

\noindent {\bf Step 2.} In this step,  we prove the off-diagonal  exponential  decay of $G_\Lambda(\sigma)((\bm x,\xi);(\bm x', \xi'))$ assuming $\sigma\not\in\Sigma_{N}$ and $\Lambda\in(\bm 0, \bm n)+\mathcal{ER}_{\bm 0}(N)$ with $|\bm n|\leq 10N.$  By Theorem \ref{clusthm}, we know that there exists $B\subset \Lambda$ so that 
\begin{align*}
\inf_{ (\xi, \bm k, \bm n)\in\Lambda\setminus B}|\xi(\sigma+\bm k\cdot\bm\omega)+\mu_{\bm n}|\geq \frac{(\varepsilon+\delta)^{\frac{1}{8b}}}{8}.
\end{align*}
Using the Neumann series argument  (cf. Lemma \ref{lwlem}) yields that for any $\Lambda'\subset\Lambda$ with $\Lambda'\cap B=\emptyset,$
\begin{align}
\label{intexp1}\|G_{\Lambda'}(\sigma)\|&\leq 50(\varepsilon+\delta)^{-\frac{1}{4b}},\\
\label{intexp2}|G_{\Lambda'}(\sigma)((\bm x,\xi);(\bm x', \xi'))|&\leq 50(\varepsilon+\delta)^{-\frac{1}{4b}}e^{-\gamma|\bm x-\bm x'|}\ {\rm for}\ \bm x\neq \bm  x'.
\end{align}
We say that $\mathcal{ER}(\sqrt{N})\ni Q,\ Q\subset \Lambda$ is $\sqrt{N}$-regular if both \eqref{intexp1} and \eqref{intexp2} hold with $\Lambda'=Q.$  Otherwise, $Q$ is called $\sqrt{N}$-singular.  Let $\mathcal{F}$ be any family of pairwise disjoint $\sqrt{N}$-singular sets   contained in $\Lambda$. We obtain 
\begin{align*}
\#\mathcal{F}\leq 2b. 
\end{align*}
We have established the sublinear bound. Then applying the coupling lemma, i.e.,  Lemma \ref {Liulem1}  of   \cite{Liu22}   shows 
\begin{align*}
|G_\Lambda(\sigma)((\bm x,\xi);(\bm x', \xi'))|\leq e^{-\gamma_N|\bm x-\bm x'| } \ {\rm for}\  |\bm x-\bm x'|>N^{\frac89},
\end{align*}
where  
\begin{align*}
\gamma_N&\geq \gamma- N^{-\frac89}\cdot \log\frac{1}{\varepsilon+\delta}-N^{-\vartheta}\\
&\geq \gamma- {\log^{-\frac{8}{9\rho}+1} \frac{1}{\varepsilon+\delta}}-N^{-\vartheta}\\
&\geq  \gamma_0-N^{-\vartheta}\ ({\rm since}\ 0<\rho<\frac{1}{10^4}). 
\end{align*}
This proves the off-diagonal  exponential decay estimates. 

We  have completed the proof of Lemma \ref{intlem}.

\end{proof}
%\end{proof}

	\subsubsection{The large scales LDE}
	In this section,  we will prove the LDE for 
	\begin{align*}
	(\varepsilon+\delta)^{-c_1}<N<\infty.
	\end{align*}
To perform the multi-scale scheme, we need two small scales $N_1<N_2$ with
\begin{align*}
N_2=N_1^{\frac{2}{\rho}},\  N=N_1^{\frac{1}{\rho^2}}. 
\end{align*}
We assume the LDE hold  for $\bm\omega\in \Omega_{N_2}\subset\Omega_{N_1}$. %   for these two  scales. 

%where $\rho>0$ will be specified below. 

Then we have 
\begin{lem}\label{lsclem}
There are some $0<c_2=c_2(b,d)\leq \frac{1}{10^4(b+d)^2}$ and  some  $$\delta_3=\delta_3(b, d, K, V, \max\limits_{1\leq \ell\leq b}|\bm n_\ell|, \kappa_1, \kappa_2, \rho,  c_1, C_2)>0$$
so that the following holds true for $0<\varepsilon+\delta\leq \delta_3$, $\varepsilon\leq \delta\leq \log^{-1}\frac{1}{\varepsilon}$ and $0<\rho\leq c_2\kappa_1$ with $\kappa_1$ given by Lemma \ref{Boulem}. 
Let $(\bm \alpha, \bm \theta)\in\mathcal W$  (cf. \eqref{ALset})  and $N$ satisfy  %and let  $\alpha\in(\cap_{N_0\leq N'\leq N}\mathcal{A}_{N'})\cap {\rm DC}_\alpha(N)$ with $\mathcal{A}_{N'}$ being given by Lemma \ref{Boulem}.  Assume that  
\begin{align*}
	(\varepsilon+\delta)^{-c_1}<N<\infty.
	\end{align*}
 Then there is a semi-algebraic set \  $\widetilde\Omega_N\subset\Omega$  (cf. Lemma \ref{mellem})  with ${\rm meas}(\Omega\setminus\widetilde \Omega_N)\leq e^{-\frac15N^{\frac{3\kappa_1\rho^2}{4}}}$  and $\deg \ \widetilde \Omega_N\leq N^{5(b+d)}$ so that,  for $\bm\omega\in{\rm DC}_{\bm\omega}(N)\cap\widetilde\Omega_N\cap{\Omega_{N_2}}$ (with ${\rm DC}_{\bm\omega}(N)$ given by \eqref{dcomg}),     there is some 
$\Sigma_N\subset\R$ with ${\rm meas}(\Sigma_N)\leq e^{-N^{\rho}}$, so that,  if $\sigma\notin\Sigma_N$, then for all $\Lambda\in(\bm 0, \bm n)+\mathcal{ER}_{\bm 0}(N)$ with $|\bm n|\leq 10N,$
\begin{align*}
\|G_{\Lambda}(\sigma)\|&\leq e^{N^{\frac{3}{4}}},\\
|G_{\Lambda}(\sigma)((\bm x,\xi);(\bm x', \xi'))|&\leq e^{-\gamma_N|\bm x-\bm x'|}\ {\rm for}\ |\bm x-\bm x'|\geq N^{\frac{8}{9}},
\end{align*}
where $[\frac\gamma2, \gamma]\ni\gamma_N\geq\gamma_{N_1}-N^{-\kappa}$   (cf. \eqref{kappa}). 
\end{lem}

We begin with a useful definition. 

\begin{defn}[{\bf $\sigma$-good}]\label{sgmbd}
For  $\Lambda\in \mathcal{ER}(L)$, we say $\Lambda$ is $\sigma$-{\it good} if 
\begin{align*}
\|G_{\Lambda}(\sigma)\|&\leq e^{L^{\frac{3}{4}}},\\
|G_{\Lambda}(\sigma)(\bm x,\xi);(\bm x', \xi'))|&\leq e^{-\gamma_L |\bm x-\bm x'|}\ {\rm for}\ |\bm x-\bm x'|\geq L^{\frac{8}{9}},
\end{align*}
where $\gamma_L\in[\frac \gamma 2, \gamma]$, and $\sigma$-{\it bad}, if one of the inequalities is violated. 
\end{defn}
The proof of Lemma \ref{lsclem} is based on matrix-valued Cartan's lemma (Lemma \ref{mcl}) and iterations of resolvent identities.  To apply Cartan's lemma, one needs to establish   the sublinear bound on $\sigma$-bad  $\Lambda\in (\bm k, \bm n)+\mathcal{ER}_{\bm 0}(N_1)$  contained in $\Lambda_N$: If  $|\bm n|\leq 10N_1$, the sublinear bound can be obtained via the T\"oplitz property of $H$ in the $\bm k$-direction   together with the semi-algebraic sets theory (cf. Lemma \ref{lsclem1}); for $|\bm n|>10N_1$, we will use the short-range property of $S$ and Bourgain's geometric lemma (Lemma \ref{Boulem}) to prove the sublinear bound. Once the sublinear bound is derived, one can use the Cartan's lemma and resolvent identities to finish the proof of large scales LDE. 

We need the following ingredients to prove Lemma~\ref{lsclem}. 

\begin{lem}\label{lsclem1} There is  some $0<c_2=c_2(b,d)\leq \frac{1}{10^4(b+d)^2}$   so that the following holds true. Assume that LDE hold  true at scale $N_1$ with $$N_1\geq N_0(b,d,\rho,K,V).$$ If  $0<\rho \leq c_2$,  then for all $\sigma\in\R$,  the number of  $\sigma$-bad   $\Lambda\in (\bm k,\bm n)+\mathcal{ER}_{\bm 0}(N_1)$ satisfying $|\bm k|\leq N^2$ and $|\bm n|\leq 10N_1$ is at most $N^{\frac{1}{10^3(b+d)^2}}.$ 
 \end{lem}
\begin{proof}
First, we define $\widetilde \Sigma_{N_1}$ to be the set of  $\sigma\in\R$ so that,  there is    $\Lambda\in (\bm0, \bm n)+\mathcal{ER}_{\bm 0}(N_1)$ satifying  $|\bm n|\leq 10N_1$,  which is $\sigma$-bad.  We then obtain using LDE at scale $N_1,$
\begin{align*}
{\rm meas}(\widetilde \Sigma_{N_1})\leq C(b,d)N_1^{C(b,d)}e^{-N_1^{\rho}}<e^{-\frac12 N_1^{\rho}}
\end{align*}
for  $N_1\geq N_0(b,d,\rho).$
Using the Hilbert-Schmidt norm and  Cramer's rule, we can identify  $\widetilde \Sigma_{N_1}$ with  a semi-algebraic set of $\deg \widetilde \Sigma_{N_1}\leq N_1^{C(b,d)}$.  Then using Basu-Pollack-Roy Theorem \cite{BPR96} on Betti numbers of a semi-algebraic set (cf. also Proposition 9.2 of \cite{Bou05}), we have a decomposition of $$\widetilde \Sigma_{N_1}\subset\bigcup_{1\leq \ell\leq N_1 ^{C(b,d)}}I_\ell,$$
where each $I_\ell$ is an interval of length $|I_\ell|\leq e^{-\frac12 N_1^{\rho}}.$  Now assume that $\Lambda\in(\bm k,\bm n)+\mathcal{ER}_{\bm 0}(N_1)$ and $\Lambda'\in(\bm k',\bm n)+\mathcal{ER}_{\bm 0}(N_1)$, which  are all $\sigma$-bad for  some  $|\bm n|\leq 10N_1, \bm k\neq \bm k'$ and $|\bm k|, |\bm k'|\leq N^2$.  Then from  the T\"oplitz property of ${H}(\sigma)$ in the $\bm k$-direction, we obtain 
$$\sigma+\bm k\cdot\bm \omega,\ \sigma+\bm k'\cdot\bm\omega\in\widetilde \Sigma_{N_1}. $$
As a result, by $\bm k\neq \bm k'$ with $|\bm k|, |\bm k'|\leq N^2$, $\bm\omega\in{\rm DC}_{\bm\omega}(N)$ (cf. \eqref{dcomg}) and $N\sim N_1^{\frac{1}{\rho^2}}$,  we have for $N_1\geq N_0(b,d, \rho, K,V),$
\begin{align*}
{e^{-\frac12 N_1^{\rho}}}=e^{-\frac12 N^{\rho^3}}< e^{-N^{\rho^4}}\leq |(\bm k-\bm k')\cdot\bm\omega|,
\end{align*}
which implies $\bm k$ and $\bm k'$ can not stay in the same interval $I_\ell$ for   $1\leq \ell\leq N_1^{C(b,d)}.$  Thus we have shown that the number of $\sigma$-bad   $\Lambda\in (\bm k,\bm n)+\mathcal{ER}_{\bm 0}(N_1)$ with $|\bm n|\leq 10N_1$ and $|\bm k|\leq N^2$ is at most
\begin{align*}
N_1^{C(b,d)}\leq N^{\rho^2 \cdot C(b,d)} \leq N^{\frac{1}{10^3(b+d)^2}},
\end{align*}
assuming $0<\rho\leq c_2(b,d)\leq \frac{1}{10^4(b+d)^2}. $ This proves Lemma \ref{lsclem1}. %{\bf Claim 1}.
\end{proof}

%{\bf Claim 2.}
 
 Next, we will deal with $|\bm n|>10N_1$ by combining short-range  property  of $S$ and Bourgain's geometric lemma (cf. Lemma \ref{Boulem}). For  $\Lambda\subset\Z^{b+d}$,  denote  by %$$\Lambda_*=\Pi_{b,d}\Lambda\subset \Z^{b+d}.$$ Write  further  
 $\Pi_b\Lambda $ (resp. $\Pi_d\Lambda$)    the projection of $\Lambda$ on $\Z^b$ (resp. $\Z^d$). For each $\bm k\in \Pi_b\Lambda, $ define 
$$\Lambda(\bm k)=\{\bm n\in\Z^d:\ (\bm k,\bm n)\in\Lambda\}.$$
We have 
 \begin{lem}\label{lsclem2} Assume that $(\bm\alpha, \bm\theta)\in\mathcal{\mathcal W'}$ (cf. Lemma \ref{Boulem}) and  $L\geq 10^{-\frac{32}{\kappa_1^2}}\log^{\frac{4}{\kappa_1}}(\log \frac1\varepsilon)$. We have
 { 
\begin{itemize}
\item[(1)] Fix $\Lambda\in (\bm k', \bm n')+\mathcal{ER}_{\bm 0}(L) $ with  $|\bm n'|>10L$.  There is a collection of $\{\lambda_\ell=\lambda_\ell(\bm \alpha, \bm\theta)\}_{1\leq \ell\leq L^{b+2d}}$ depending only on $\Lambda, \bm\alpha, \bm\theta, V$ so that,  if 
%Let $k\in\Pi_b\Lambda$
%and  let $\{\zeta_l=\zeta_l(k)\geq \frac12\} _{1\leq l\leq \#\Pi_d\Lambda(k)}$ denote the set of all eigenvalues of  $$R_{\Pi_d\Lambda(k)}(V^2+\varepsilon \Delta)R_{\Pi_d\Lambda(k)}\ {\rm on}\ \Z^d.$$
%Then %there is a sequence $\{\zeta_l\}_{\leq l\leq N_1^C}\subset \R$ with  $C=C(b,d)>0$ and $\zeta_l$ depending only on $\Lambda,\alpha,\theta,m $ (but not on $\sigma,\omega$) so that, if
%we have for 
$$\min_{\bm k\in\Pi_b\Lambda, 1\leq \ell\leq L^{b+2d}, \xi=\pm1}|\sigma+\bm k\cdot\bm\omega+\xi{\lambda_\ell}|> e^{-\frac14 L^{\frac{3\kappa_1}{4}}},$$
 then $\Lambda$ is $\sigma$-good with $\gamma_L=\gamma$, where $\kappa_1$ is given in Lemma \ref{Boulem}. %where $\Pi_b\Lambda$ denotes the projection of $\Lambda$ on $\Z^b.$
%\begin{rem}
%We have since $m\geq 2$ that $V^2+\varepsilon \Delta\geq \frac12 {\rm id}$ if $0<\varepsilon \ll1$, where ${\rm id}$ denotes the identity operator. Obviously, $\zeta_l$ depends  on $\Pi_d\Lambda(k), \alpha,\theta, m$, but not on $\sigma, \omega.$
%\end{rem}
\item[(2)]Let $L\leq N\leq e^{\log^2 L}$. There exists some $\tilde\Sigma_{L}\subset \R$ with ${\rm meas}(\tilde\Sigma_{L})\leq  e^{-\frac15 L^{\frac{3\kappa_1}{4}}}$ so that,  if $\sigma\notin\tilde\Sigma_{L}$,  then each   $\Lambda\in (\bm k, \bm n)+\mathcal{ER}_{\bm 0}(L) $ satisfying $|(\bm k,\bm n)|\leq 100N^2$ and $|\bm n|>10L$ is $\sigma$-good with $\gamma_L=\gamma$. 
\end{itemize}}
  \end{lem}
\begin{proof}
%The proof consists of two steps. \\
%\begin{itemize}
(1)  %For simplicity, we %Fix any $\Lambda\in (k',n')+\mathcal{ER}_0(N_1) $ with $\Lambda\subset [-N, N]^b\times [-3N, 3N]^d$ and $|n'|>2N_1$ (the case $l=2$ can be handled similarly).  
%We want to show 
%\b\b\begin{align}\label{sgmlk}
%{\rm meas}(\Sigma_{\Lambda,k})\leq e^{-\frac23N_1^{\frac{3}{4}}}.
%\end{align}
%where 
%\begin{align*}
%\Sigma_{\Lambda,k}=\left\{\sigma\in\R:\ \|(R_{\Pi_{d}(k)}(\mathcal{D}(\sigma)+\varepsilon\mathcal{T})R_{\Pi_{d}(k)})^{-1}\|\geq e^{-N_1^{\frac{3}{4}}}\right\}.
%\end{align*} 
Let  $\bm k\in\Pi_b\Lambda$.  {\it It is important that  either $\Lambda({\bm k})\in \mathcal{ER}_{\Z^d}(L),$  or  $\Lambda({\bm k})$ is a rectangle  of  width   at least $L$ (cf. Lemma \ref{projlemz})}.   In each of the two cases, we have  $\Lambda({\bm k})\in\mathcal E_{2L}^{L}$. 

For   $Q\subset \Lambda({\bm k}), $ %for some $\sqrt{N_1}\leq L\leq N_1.$
write 
\begin{align*}
 A_{\bm k, Q}=A_{\bm k, Q}(\sigma)&=R_{Q}({D}(\sigma)+\varepsilon(\Delta\oplus\Delta))R_{Q}\\
 &=(-\sigma-\bm k\cdot\bm \omega+\mathcal{L}_{Q}(\bm\theta)) \oplus (\sigma+\bm k\cdot\bm \omega+\mathcal{L}_{Q}(\bm\theta)), 
%&=  R_{\Pi_{d}\Lambda(k)}\left(
	%\begin{array}{cc}
 %-x V+V^2+\frac\varepsilon 2\Delta&  \frac\varepsilon2 \Delta\\
%  \frac\varepsilon2 \Delta&  x V+V^2+\frac\varepsilon2\Delta
%\end{array}\right)R_{\Pi_{d}\Lambda(k)}\\
%&:=\left(
%	\begin{array}{cc}
 %-x V+V^2+\frac\varepsilon 2\Delta&  \frac\varepsilon2 \Delta\\
 % \frac\varepsilon2 \Delta&  x V+V^2+\frac\varepsilon2\Delta
%\end{array}\right)_{{\Pi_{d}\Lambda(k)}}.
\end{align*}
where $\mathcal{L}_{Q}(\bm\theta):=\mathcal H_Q(\bm \alpha, \bm \theta)$ is defined in Lemma \ref{Boulem}. 
Then  the set of all eigenvalues of $A_{\bm k, Q}$ is given by 
$$\{\pm (\sigma+\bm k\cdot\bm \omega)+\lambda_{\ell, Q}(\bm \alpha, \bm\theta)\}_{\ell=1}^{\#Q},$$
where $\{\lambda_\ell:=\lambda_{\ell, Q}(\bm \alpha, \bm\theta)\}_{\ell=1}^{\#Q}$ denotes the spectrum  of $\mathcal{L}_{Q}(\bm\theta).$ 
So we have 
\begin{align*}
\|A_{\bm k, Q}^{-1}\|&\leq \max_{\xi=\pm1, 1\leq \ell\leq \#Q}|\xi (\sigma+\bm k\cdot\bm\omega)+\lambda_\ell|^{-1}.
%&= \max_{1\leq l\leq \#\Lambda'}|\sigma+k\cdot\omega-\sqrt{\zeta_l(k)}|^{-1}\cdot|\sigma+k\cdot\omega+\sqrt{\zeta_l(k)}|^{-1}.
\end{align*}
%\begin{align*}
%&\ \ \ \{\sigma+k\cdot\omega:\ \|A^{-1}(\sigma+k\cdot\omega)\|\leq e^{\frac12N_1^{\frac{3}{4}}}\}\\
%&\subset\bigcup_{1\leq l\leq  \Pi_d\Lambda(k)} [\zeta_l-\frac12 e^{-\frac 13 N_1^{\frac{3}{4}}},\  \zeta_l+\frac12 e^{-\frac 13 N_1^{\frac{3}{4}}}].
%\end{align*}
We now can  use Lemma \ref{Boulem}  and  resolvent identities. Since $\bm \alpha\in\mathcal{A}\subset\mathcal{A}_{L}$,   applying  (1)  of  Lemma \ref{Boulem} with $E_{\bm k}=\sigma+\bm k\cdot\bm \omega$  similar to \cite{JLS20} (cf. pages 471--472, the proof of Theorem 3.7) gives the following:  For each $\bm m\in\Lambda(\bm k)$,   there exist $\frac{1}{4}L^{\kappa_1}\leq \tilde{L} \leq L^{\kappa_2}$,  $Q_{\bm m} \in \mathcal{ER}_{\Z^d}(\tilde L)$ and $\tilde{Q}_{\bm m}$, such that
\begin{align*}
  &Q_{\bm m} \subset \bm m+[-\tilde{L}, \tilde{L}]^d, \   \bm m+[-\tilde{L}^{\frac{1}{10d}}, \tilde{L}^{\frac{1}{10d}}]^d\subset\tilde{Q}_{\bm m},\\
& \bm m\in Q_{\bm m} \subset \Lambda(\bm k), {\rm dist}(\bm m,  \Lambda({\bm k})\setminus Q_{\bm m})\geq \frac{ \tilde{L}}{2},\ {\rm diam}\ \tilde Q_{\bm m}  \leq 4 \tilde{L}^{\frac{1}{10d}}.
\end{align*}
Also  for any $ \bm n\in  Q_{\bm m}\setminus  \tilde{Q}_{\bm m}$,
there exist $L_1\sim (\log L)^{\frac{4}{\kappa_1}}$ and  some $\mathcal{ER}_{\Z^d}(L_1)\ni W \subset Q_{\bm m}\setminus  \tilde{Q}_{\bm m}$ such that
\begin{equation*}
 \bm n\in W,\  {\rm dist}(\bm n, Q_{\bm m}\setminus  \tilde{Q}_{\bm m}\setminus W)\geq \frac{L_1}{2},
\end{equation*}
and 
\begin{align*}
\|\mathcal{H}^{-1}_{W}(E_{\bm k};\bm\theta)\|&\leq e^{\sqrt{L_1}},\\
|\mathcal{H}^{-1}_{W}(E_{\bm k};\bm\theta)(\bm n; \bm n')|&\leq e^{-\frac12|\log\varepsilon|\cdot|\bm n-\bm n'|}\ {\rm for}\ |\bm n-\bm n'|\geq {L_1^{\frac89}}.
\end{align*}
Next,  assume that (set $Q=Q_{\bm m}$)%for all $\bm m\in\Lambda_{\bm k}$ and $Q_{\bm m}\subset \Lambda_{\bm k}$
\begin{align*}
\min_{\bm m\in\Lambda(\bm k), \xi=\pm1, 1\leq \ell\leq \#Q_{\bm m}}|\sigma+\bm k\cdot\bm \omega+\xi{\lambda_\ell}|>e^{-\frac 14 \tilde L^{\frac{3}{4}}}.
%\sigma+k\cdot\omega\notin \bigcup_{1\leq l\leq \#\Pi_d\Lambda(k)} [\zeta_l-\frac12 e^{-\frac 23 N_1^{\frac{3}{4}}},\  \zeta_l+\frac12 e^{-\frac 23 N_1^{\frac{3}{4}}}],
\end{align*}
Then applying  Lemma \ref{JLSlem2} of \cite{JLS20}  yields  for $$A_{\bm k}:=A_{\bm k}(\sigma)=A_{\bm k, \Lambda(\bm k)}(\sigma),$$
the following estimates 
\begin{align*}
\|A^{-1}_{\bm k}\|&\leq e^{\frac 12 L^{\frac{3}{4}}},\\
|A_{\bm k}^{-1}(\bm n;\bm n')|&\leq e^{-\frac14|\log\varepsilon|\cdot|\bm n-\bm n'|}\ {\rm for}\ |\bm n-\bm n'|\geq {L}^{\frac{8}{9}}.
\end{align*}
%the sublinear bound on the number of bad blocks (of size  $M\sim N_1^c$ for some $0<c\ll1$) contained in $\Pi_d\Lambda(k)$.    Next, applying  the resolvent identity of \cite{HSSY}  (cf. Lemma 4.2) yields for $B_k=R_{\Pi_d\Lambda(k)}({D}(\sigma)+\varepsilon\Delta)R_{\Pi_d\Lambda(k)},$  %$A(\sigma+k\cdot\omega)$ has good estimates
%\begin{align*}
%\|B^{-1}_k\|&\leq e^{\frac 12 N_1^{\frac{3}{4}}},\\
%|B_k^{-1}(n;n')|&\leq e^{-\frac12|\log\varepsilon|\cdot|n-n'|}\ {\rm for}\ |n-n'|\geq {N_1}^{\frac{8}{9}}.
%\end{align*}
Note that 
\begin{align*}
\tilde G_{\Lambda}(\sigma)=\bigoplus_{\bm k\in\Pi_b\Lambda}A_{\bm k}^{-1},
\end{align*}
where 
\begin{align*}
\tilde G_{\Lambda}(\sigma)=\left(R_{\Lambda} ({D}(\sigma)+\varepsilon(\Delta\oplus\Delta))R_{\Lambda}\right)^{-1}.
\end{align*}

By taking account of all $\bm k\in\Pi_b\Lambda,$ we obtain  that if 
\begin{align}\label{2mel}
 \min_{\bm k\in\Pi_b\Lambda,  \bm m\in\Lambda(\bm k), \xi=\pm1, 1\leq \ell\leq \#Q_{\bm m}}|\sigma+\bm k\cdot\bm \omega+\xi{\lambda_\ell}|> e^{-\frac 14 L^{\frac{3\kappa_1}{4}}},
%\sigma+k\cdot\omega\notin \bigcup_{1\leq l\leq \#\Pi_d\Lambda(k)} [\zeta_l-\frac12 e^{-\frac 23 N_1^{\frac{3}{4}}},\  \zeta_l+\frac12 e^{-\frac 23 N_1^{\frac{3}{4}}}],
\end{align}
then 
\begin{align*}
\|\tilde G_{\Lambda}(\sigma)\|&\leq e^{\frac12L^{\frac{3}{4}}},\\
|\tilde G_{\Lambda}(\sigma)((\bm k, \bm n, \xi);(\bm k', \bm n', \xi'))|&\leq \delta_{\bm k, \bm k'}e^{-\frac14|\log\varepsilon| \cdot|\bm n-\bm n'|}\ {\rm for}\ |\bm n-\bm n'|\geq L^{\frac{8}{9}}.
\end{align*}
Note that the  total number of $\lambda_\ell$ in \eqref{2mel}   is at most $$C(b, d)L^b L^dL^{\kappa_2}L^{\kappa_2d}\ll L^{b+2d}.$$

Finally, since  $|\bm n|>10L$,  we know that $S$ has the decay estimate
\begin{align*}
|S((\bm x,  \xi);(\bm x', \xi'))|\leq C_2(1+|\bm x-\bm x'|)^{C_2}e^{-6\gamma L}e^{-2\gamma |\bm x-\bm x'|}. 
\end{align*}
%where $\bm x=(\bm k, \bm n), \bm x'=(\bm k', \bm n').$
Then using the perturbation   Lemma \ref{lwlem}   implies  that $\Lambda$  is $\sigma$-good  assuming \eqref{2mel} holds true.  {In fact, we can apply Lemma \ref{lwlem} (we hide the  dependence on $\xi, \xi'$) with $B=\delta S_\Lambda(\bm x, \bm x')$  , $A^{-1}(\bm  x, \bm x')=\tilde G_{\Lambda}(\bm x, \bm x')$ and $c=2\gamma, \epsilon_1=e^{-\frac 12 L^{\frac34}}, M=L^{\frac89}, \epsilon_2=C\delta e^{-6\gamma L}$. Since   $\gamma\in(\frac12, 10)$, we obtain  for $L\geq N_0(b,d, C_2)\gg1,$
$$ C(b,d, C_2)L^{2(b+d)+C_2}\delta e^{-6\gamma L}e^{2\gamma L^{\frac89}}e^{\frac 12 L^{\frac34}} \leq \frac12.$$
This verifies all assumptions  of Lemma \ref{lwlem}. So using Lemma \ref{lwlem}, we get 
\begin{align*}
	\|G_{\Lambda}(\sigma)\|&\leq 2 e^{\frac12L^{\frac{3}{4}}}<e^{L^{\frac34}},
\end{align*}
and  for $|\bm x-\bm x'|\geq L^{\frac{8}{9}}$ (since $\frac 14 |\log \varepsilon|>2\gamma>1$), 
\begin{align*}
	|G_{\Lambda}(\sigma)((\bm x, \xi);(\bm x',  \xi'))|&\leq |\tilde G_{\Lambda}(\sigma)((\bm x, \xi);(\bm x',  \xi'))|+
	e^{\frac 12 L^{\frac34}}e^{-2\gamma |\bm x-\bm x'|} \\
	&\leq  e^{L^{\frac34}-L^{\frac89}}e^{-\gamma |\bm x-\bm x'|}\\ %=e^{-|\bm x-\bm x'|(\gamma-N_1^{\frac34}}\ {|\bm x-\bm x'|})}\\
	&\leq e^{-\gamma |\bm x-\bm x'|}. 
\end{align*}}
%where $\gamma_{N_1}=\gamma-N_1^{-\frac{1}{36}}.$
 This proves (1) of Lemma \ref{lsclem2}.  
\begin{rem}\label{remn1}
\begin{itemize}
\item Let $L\leq N\leq e^{(\log L)^2}$. We would also like to remark that if we take into account of all these $\Lambda\in (\bm k, \bm n)+\mathcal{ER}_{\bm 0}(L)$ with $|(\bm k, \bm n)|\leq 100N^2$ and $|\bm n|>10L$ (the total number of  these $\Lambda$ is at most $C(b,d)N^{2(b+d)}$), then we get a sequence  of real numbers $\{{\lambda_\ell(\bm \alpha, \bm \theta)}\}_{1\leq \ell\leq N^{4(b+d)}}$. As a result,   there is a set $\tilde\Sigma_{L}\subset \R$  satisfying  $${\rm meas}(\tilde\Sigma_{L})\leq e^{C(b,d)\log^2 L}e^{-\frac14 L^{\frac{3}{4}\kappa_1}}\leq e^{-\frac15 L^{\frac{3}{4}\kappa_1}}$$ so that,  for  $\sigma\not\in\tilde\Sigma_{L}$, all   $\Lambda\in(\bm k, \bm n)+\mathcal{ER}_{\bm 0}(L)$ satisfying  $|(\bm k, \bm n)|\leq 100N^2$ and $|\bm n| >10L$ are $\sigma$-good. 
\item In the application, we typically choose  $L=N_1\geq 10^{-\frac{32}{\kappa_1^2}}\log^{\frac{4}{\kappa_1}}(\log \frac1\varepsilon)$. It then follows from $  N_1\sim N^{\rho^2}\geq (\varepsilon+\delta)^{-c_1\rho^2}$  that  a restriction  like $\delta \leq \log^{-1}\frac1\varepsilon$ is needed. 
\end{itemize}
\end{rem}
%{\bf Step 2.}  In this step we need an important result originated from Bourgain \cite{},  which together with estimates in {\bf Step 1}  will lead to the exponential off-diagonal decay of $A^{-1}(\sigma+k\cdot\omega)$ for $\sigma+k\cdot\omega\notin \Sigma_{\Lambda, k}.$  

%{\bf Step 2.} Now assume $\sigma+k\cdot\omega\notin \Sigma_{\Lambda,k}$ for all $k\in\Pi_b\Lambda$. Using estimates of {\bf Step 1} provides the upper bound of $\|A^{-1}(\sigma+k\cdot\omega)\|$.  
%\begin{align}
%&\ \ \ \ \tilde G_{\Lambda}((k,n);(k',n'))\\
%& =G_{\Pi_d(k)}(n;n^*)
%-\frac\varepsilon2\sum_{n^*,n^{**}}G_{\Pi_d(k)}(n;n^*)\Delta(n^*;n^{**})\tilde G_\Lambda((k,n^{**});(k',n'))
%\end{align}
%we obtain % the good estimates of $\tilde G_\Lambda(\sigma),$  i.e., 
%\begin{align*}
%\|\tilde G_{\Lambda}(\sigma)\|&\leq e^{N_1^{\frac{3}{4}}},\\
%|\tilde G_{\Lambda}(\sigma)((k,n);(k',n'))|&\leq \delta_{k,k'}e^{-\frac12|\log\varepsilon| \cdot|n-n'|}\ {\rm for}\ |n-n'|\geq N_1^{\frac{8}{9}}.
%\end{align*}
%Note that $|n|>2N_1$. Then $\mathcal{T}_{\phi}$ has decay estimates
%\begin{align*}
%\mathcal{T}_{\phi}((k,n);(k',n'))\leq Ce^{-2\gamma N_1}.
%\end{align*}
%So applying the Neumann series argument implies the $\sigma$-good of $\Lambda$.  %good estimates of  $G_\Lambda (\sigma).$  
%Finally, it suffices to take account of all these $\Lambda\subset [-N,N]^b\times [-3N, 3N]^d$.  This proves {\bf Claim 2}.
(2) The conclusion is just  that   in Remark~\ref{remn1}. 
%\end{itemize}

This completes the proof of  Lemma \ref{lsclem2}.%{\bf Claim 2}.
\end{proof}
By combining Lemma \ref{lsclem1} and Lemma \ref{lsclem2},  we are ready to prove LDE at the scale $N$ (indeed all scales in $[N, N^2]$).  We first recall an important lemma. % (i.e.,  Proposition 14.1 in \cite{Bou05}). 

\begin{lem}[Matrix-valued Cartan's lemma,  cf. Proposition 14.1 in \cite{Bou05}]\label{mcl}
Let $T(\sigma)$ be a self-adjoint $N\times N$ matrix-valued function of a parameter $\sigma\in[-\beta, \beta]$  satisfying the following conditions:
\begin{itemize}
\item[(i)] $T(\sigma)$ is real analytic in $\sigma$ and has a holomorphic extension to
\begin{align*}
\mathbb{D}=\left\{z\in\mathbb{C}: \ |\Re z|\leq \beta,\ |\Im z|\leq \beta_1\right\}
\end{align*}
satisfying $\sup\limits_{z\in \mathbb{D}}\|T(z)\|\leq K_1,\  K_1\geq 1.$

\item[(ii)] For each $\sigma\in[-\beta, \beta]$, there is a subset $Y\subset [1,N]$ with $\#Y\leq M$ such that 
\begin{align*}
\|(R_{[1,N]\setminus Y}T(\sigma)R_{[1,N]\setminus Y})^{-1}\|\leq K_2, \ K_2\geq 1.
\end{align*}
\item[(iii)] Assume %for some $a\in[-\frac12,\frac12],$
\begin{align*}
%\|T^{-1}(a)\|\leq K_3,\ K_3\geq 1.
\mathrm{meas}\left(\{\sigma\in[-{\beta}, {\beta}]: \ \|T^{-1}(\sigma)\|\geq K_3\}\right)\leq 10^{-3}\beta_1(1+K_1)^{-1}(1+K_2)^{-1}.
\end{align*}
\end{itemize}
Let $0<\epsilon\leq (1+K_1+K_2)^{-10 M}.$ 
 Then we have
\begin{align}\label{mc5}
\mathrm{meas}\left(\left\{\sigma\in\left[-{\beta}/{2}, {\beta}/{2}\right]:\  \|T^{-1}(\sigma)\|\geq \frac1\epsilon\right\}\right)\leq Ce^{\frac{-c\log \frac1\epsilon}{M\log( M+K_1+K_2+K_3)}},
\end{align}
where $C, c>0$ are some absolute constants.
\end{lem}
%In contrast, we introduce 
%\begin{lem}[Lemma 2 in \cite{Bou09}]%\label{mcl}
%Let $T(\sigma)$ be a self-adjoint $N\times N$ matrix function of a parameter $\sigma\in[-\frac12, \frac12]$  satisfying the following conditions:
%\begin{itemize}
%\item[(i)] $T(\sigma)$ is real analytic in $\sigma\in [-1,1]$ and has a holomorphic extension to
%\begin{align*}
%\mathbb{D}=\left\{z\in\mathbb{C}: \ |z|\leq 1\right\}
%\end{align*}
%satisfying $\sup\limits_{z\in \mathbb{D}}\|T(z)\|\leq K_1,\  K_1\geq 1.$

%\item[(ii)] There is a subset $V\subset [1,N]$ with $|V|\leq M$ such that 
%\begin{align*}
%\sup_{z\in\mathbb{D}}\|(R_{[1,N]\setminus V}T(z)R_{[1,N]\setminus V})^{-1}\|\leq K_2, \ K_2\geq 1.
%\end{align*}
%\item[(iii)] Assume for some $a\in[-\frac12,\frac12],$
%\begin{align*}
%\|T^{-1}(a)\|\leq K_3,\ K_3\geq 1.
%\end{align*}
%\end{itemize}
%Then for all $t>0,$
%\begin{align}\label{mc5}
%\mathrm{mes}\left(\left\{\sigma\in\left[-{1}/{2}, {1}/{2}\right]:\  \|T^{-1}(\sigma)\|\geq e^{-t}\right\}\right)\leq Ce^{-\frac{ct}{M\log(K_1K_2K_3)}},
%\end{align}
%where $C, c>0$ are some absolute constants.
%\end{lem}
Now, we can prove Lemma \ref{lsclem}.  
\begin{proof}[Proof of Lemma \ref{lsclem}]
Recalling 
\begin{align*}
N_2=N_1^{\frac 2\rho},\  N=N_1^{\frac 1{\rho^2}}, 
\end{align*}
we assume that for $L=N_1, N_2$ and  $\Lambda\in (\bm 0, \bm n)+\mathcal{ER}_{\bm 0}(L)$ with $|\bm n|\leq 10L$, $G_{\Lambda}(\sigma)$ satisfies the LDE.

Without loss of generality,  we  only  establish LDE for  the  Green's function on $\Lambda_{{\rm pm}, N}=\{(\bm x, \xi)\in\Z_{{\rm pm}, *}^{b+d}:\ |\bm x|\leq N \}= (\Lambda_{N}\times\{+, -\})\cap \Z_{{\rm pm}, *}^{b+d}$,
the other cases  (i.e., $\Lambda\in (\bm 0, \bm n)+{\mathcal{ER}_{\bm 0}(N_3)}$ for $|\bm n|\leq 10N_3, N_3\in[N, N^2]$) can be dealt with similarly. We define three intermediate scales between $[N_2, N]$: 
\begin{align}\label{L123}
L_1=N^{\frac{1}{100(b+d)^2}},\ L_2=N^{\frac{1}{50(b+d)^2}},\ L_3= N^{\frac{1}{10(b+d)}}. 
\end {align}
Note that from  $0<\rho\leq c_2(b,d)\leq {\frac{1}{10^4(b+d)^2}}$ (cf. Lemma \ref{lsclem1}), we have 
$$\log N_1\ll\log N_2\ll \log L_1= \frac{1}{100(b+d)^2}\log N.$$

We first cover  the $\Lambda_{{\rm pm}, N}$  with two overlapped subregions, $\Lambda_{{\rm pm}, N}=\mathcal{R}_1\cup \mathcal{R}_2$ satisfying 
\begin{align*}
\mathcal{R}_1=\{(\bm k, \bm n, \xi )\in\Lambda_{{\rm pm}, N}:\ |\bm n|\leq L_2\},\\
\mathcal{R}_2=\{(\bm k, \bm n, \xi )\in\Lambda_{{\rm pm}, N}:\ |\bm n|\geq L_1\}.
\end{align*}
We will apply matrix-valued Cartan's lemma (cf. Lemma \ref{mcl}) and Bourgain's geometric lemma \cite{Bou07} (cf. Lemma \ref{Boulem}) to obtain the sub-exponential  growth estimates on Green's functions.   Then we apply the  coupling lemmas  of  \cite{BGS02, JLS20, Liu22} (cf. e.g., Lemmas \ref{Liulem1}, \ref{JLSlem1}, \ref{JLSlem2})  based on the iteration of  resolvent identities    to prove the off-diagonal exponential decay of  $G_{\Lambda_{{\rm pm}, N}}(\sigma)$   if desired  $\sigma$ are removed.  

 Before we enter the details, let us first explain why we need to divide $\Lambda_{{\rm pm}, N}$ into the above two regions. This is quite different from that in \cite{SW23}. Based on Lemma~\ref{lsclem2}, we can impose the so called weak second Melnikov's  condition on $\bm \omega$ so that  all  $\sigma$-bad (cf. Definition \ref{sgmbd})   $\Lambda\in (\bm k, \bm n)+\mathcal{ER}_{\bm 0}(N_1)\subset \Lambda_{{\rm pm}, N}$ can stay only   in  a {\bf narrow strip region} (cf. Figure \ref{graph}), namely, there is some $\bm k_*\in [-N, N]^b$ so that 
$$|\bm k-\bm k_*|\leq 10N_1\leq N^{0+}.$$ 
This gives the desired control on $\sigma$-bad $\Lambda$ in the $\bm k$-direction. To establish the sublinear bound, it remains to handle the $\bm n$-direction. In \cite{SW23}, the $\bm n$-direction estimates can be obtained directly using  the fact that $\theta\in\R$, Diophantine $\bm\omega$ and LDT for the corresponding quasi-periodic Schr\"odinger operators. Such arguments inevitably fail  for $\bm\theta\in \R^d$: There is simply no one-dimensional interval decomposition  of $d$-dimensional  semi-algebraic sets, and more essentially,  the Diophantine  property of $\bm \omega$ cannot ensure the sublinear estimate in the $\bm n$-direction.  This motivates us to  restrict  $|\bm n|\leq L_2$ (i.e., the region $\mathcal{R}_1$). We can use the Cartan's lemma and resolvent identities to deal with  $\mathcal{R}_1$. For $\mathcal{R}_2$, we  take advantage of Bourgain's geometric lemma and the short-range property of $H$  using Lemma \ref{lsclem2}. Finally, it suffices to cover  $\Lambda_{{\rm pm}, N}$ with  {\bf four types} of   regions of sizes  $N_1, L_1^{\frac{9}{10}}, L_1, L_3$ (will be specified below) contained in $\mathcal{R}_1$, $\mathcal{R}_2$,  and apply the resolvent identities again.  See Figure \ref{graph}.  

We first deal with $\mathcal R_2$, and then $\mathcal R_1$.  
\ 
\\ 

\noindent{\bf \Large Analysis in $\mathcal{R}_2$}\\

In this region, we can apply directly  the conclusion  of Lemma \ref{lsclem2}.  In this case, we make no restriction on $\bm\omega.$  We let 
\begin{align*}
L_0 =L_1^{\frac {9}{10}}= N^{\frac{9}{1000(b+d)^2}}.
\end{align*}
For any $\Lambda\subset \mathcal R_2$ satisfying $\Lambda=(\bm k, \bm n)+\mathcal{ER}(L_0)$, we  have $|\bm n|>L_1>10L_0$.  Then applying  Lemma \ref{lsclem2} with $L=L_0$ and  since $0<\rho<c_2(b,d)\kappa_1\leq \frac{\kappa_1}{10^4(b+d)^2},$
we  get a set   $\tilde\Sigma_{N,  1}\subset \R$  satisfying 
$$ \mathrm{meas}(\tilde\Sigma_{N,1})\leq e^{-\frac15 L_0^{\frac{3\kappa_1}{4}}}\leq e^{-\frac15N^{\frac{27\kappa_1}{4000(b+d)^2}}}<e^{-N^{3\rho}},$$
so that, if $\sigma \notin \tilde\Sigma_{N,1}$,  all $\Lambda \in \mathcal{ER}(L_0)$ contained in  $ \mathcal R_2$  is $\sigma$-good (cf. Definition \ref{sgmbd}).  %These $\sigma$-good elementary regions are called {\bf Type 1} regions. 

We remark that in this case, we do not use Cartan's lemma. 
\ 
\\ 
{\noindent{\bf \Large Analysis in $\mathcal{R}_1$}} 
\\ 

In this region, we need to make further restriction on $\bm\omega$ called the weak second Melnikov's condition. More precisely, we have 
{\begin{lem}\label{mellem}
There is a sequence of real numbers $\{\lambda_{\ell}=\lambda_\ell(\bm \alpha, \bm \theta)\}_{1\leq \ell \leq N^{4(b+d)}}$ depending only on $V, \bm \alpha, \bm \theta$  so that the following holds true.  If $\bm \omega\in \widetilde\Omega_{N}$ with
\begin{align}
\nonumber&\ \ \ \widetilde\Omega_N\\
\label{tilomg}&:=\bigcap_{1\leq \ell,\ell'\leq N^{4(b+d)}, \xi=\pm1, \xi'=\pm1, 0<|\bm k|\leq 2N^2}\left\{\bm\omega\in\Omega:\ |\bm k\cdot\bm \omega+\xi {\lambda_\ell}-\xi'{\lambda_{\ell'}}|>10 e^{-\frac 14 N^{\frac{3\kappa_1\rho^2}{4}}}\right\},
\end{align}
then for any $\sigma\in\R,$  there is some $\bm k_*\in[-N^2, N^2]^b$  so that, all $\sigma$-bad $\Lambda\in(\bm k, \bm n)+\mathcal{ER}_{\bm 0}(N_1)$ satisfying $|\bm k|\leq N^2$ and $10N_1\leq |\bm n|\leq 100N^2$ must obey 
\begin{align}\label{krestr}
\Pi_b \Lambda\subset [\bm k_*-10N_1, \bm k_*+10N_1]^b.
\end{align}
In particular, we have 
\begin{align}\label{mtilomg}
{\rm meas}(\Omega\setminus\widetilde\Omega_N) \leq e^{-\frac 15 N^{\frac{3\kappa_1\rho^2}{4}}}. 
\end{align}

\end{lem} }

\begin{rem}
Indeed, in this lemma,  $\lambda_\ell(\bm \alpha, \bm \theta)$ (${1\leq \ell \leq N^{4(b+d)}})$ are  eigenvalues of the operator  $\mathcal H_{Q}(\bm \alpha, \bm\theta)$ for  all $Q\subset[-100N^2, 100N^2]^d$ satisfying   $Q\in\mathcal{ER}_{\Z^d}(\tilde N_1)$ for some $\tilde N_1\in [\frac14N_1^{\kappa_1}, N_1^{\kappa_2}]$. 
\end{rem}
\begin{proof}[Proof of Lemma \ref{mellem}]
The proof follows from Lemma \ref{lsclem2}. 
%We first prove the sublinear bound on the number of  $\sigma$-bad elementary $N_1$-regions with centers in 
%$[-N, N]^b\times( [-N^{\frac{\frac{3}{4}}{10d}}, N^{\frac{\frac{3}{4}}{10d}}]^d \setminus [-2N_1, 2N_1]^d)$ for all $\sigma\in \R.$ For this purpose, an additional restriction on $\bm\omega$ (may be called the weak second Melnikov's condition) is required. 
Applying  (1) of Lemma \ref{lsclem2}   (cf. also Remark \ref{remn1})  with $L=N_1$ and taking into account of all $\Lambda\in (\bm k, \bm n)+\mathcal{ER}_{\bm 0}(N_1)$ with $|(\bm k,\bm n)|\leq 100N^2$ and $|\bm n|>10N_1$ yield   a sequence  $\{{\lambda_\ell}\}_{\leq \ell\leq N^{4(b+d)}}\subset \R$  and each $\lambda_\ell$ depending only on $\bm\alpha, \bm \theta$ (but not on $\sigma,\bm\omega$), so that if $\Lambda=(\bm k, \bm n)+\mathcal{ER}_{\bm 0}(N_1)$ satisfies $|(\bm k, \bm n)|\leq 100N^2$,  $|\bm n|>10N_1$ and 
$$\min_{\bm k\in\Pi_b\Lambda, \xi=\pm1, 1\leq \ell\leq N^{4(b+d)}}|\sigma+\bm k\cdot\bm \omega+\xi {\lambda_\ell}|>e^{-\frac 14 N_1^{\frac{3\kappa_1}{4}}}= e^{-\frac 14 N^{\frac{3\kappa_1\rho^2}{4}}},$$
then $\Lambda$ is $\sigma$-good. So recalling \eqref{tilomg}, %$\wideltilde\Omega_N,$
%\begin{align*}
%\widetilde\Omega_N=\bigcap_{1\leq l,l'\leq N^C, \xi=\pm1, \xi'=\pm1, 0<|\bm k|\leq 2N}\left\{\bm\omega\in\Omega:\ |\bm k\cdot\bm \omega+\xi {\lambda_l}-\xi'{\lambda_{l'}}|>2 e^{-\frac 14 N_1^{\kappa_1\frac{3}{4}}}\right\}. 
%\end{align*}
we have $${\rm meas}(\Omega\setminus \widetilde\Omega_N)\leq N^{C(b,d)}e^{-\frac 14 N^{\frac{3\kappa_1\rho^2}{4}}}\leq e^{-\frac 15 N^{\frac{3\kappa_1\rho^2}{4}}},$$
which implies \eqref{mtilomg}. 

In the following,   we  assume $\bm\omega\in \widetilde\Omega_N.$ 

Suppose now $\Lambda'$ with $\Lambda'\in(\bm k', \bm n')+\mathcal{ER}_{\bm 0}(N_1)$ satisfying  $|\bm k'|\leq N^2$ and  $10N_1<|\bm n'|\leq 100N^2$  is $\sigma$-bad.  Then there are some $\bm k_*\in \Pi_b\Lambda'\subset [-N^2, N^2]^b$, $\xi_*=\pm1$  and some $\ell_*\in[1, N^{4(b+d)}]$ so that 
\begin{align*}
|\sigma+\bm k_*\cdot\bm\omega+\xi_*{\lambda_{\ell_*}}|\leq e^{-\frac14 N_1^{\frac{3\kappa_1}{4}}}. 
\end{align*}
Now  let  $\Lambda''$ (satisfying  $\Lambda''\in(\bm k'', \bm n'')+\mathcal{ER}_{\bm 0}(N_1)$ with  $|\bm k''|\leq N^2$ and  $10N_1<|\bm n''|\leq 100N^2$)   be another  $\sigma$-bad region. 
  %be another $\sigma$-bad elementary $N_1$-region. 
  From $\bm\omega\in \widetilde\Omega_N,$ we must have  
\begin{align*}
\bm k_*\in \Pi_b\Lambda''.
\end{align*}
Otherwise, there must be some $\bm k'''\in \Pi_b\Lambda'', \xi'\in\{\pm1\}, \ell'$ so that,  $\bm k_*\neq \bm k'''$ and $$|\sigma+\bm k'''\cdot\bm\omega+\xi'{\lambda_{\ell'}}|\leq e^{-\frac14 N_1^{\frac{3\kappa_1}{4}}},$$ which contradicts $\bm\omega\in\widetilde \Omega_N$. 
We have established   that all $\sigma$-bad elementary $N_1$-regions $\Lambda$  with centers $(\bm k, \bm n)$ satisfying $|\bm k|\leq N^2$ and $10N_1<|\bm n|\leq 100N^2$  must satisfy  \eqref{krestr}. 
%\begin{align}
%\Pi_b\Lambda\subset [\bm k_*-10N_1, \bm k_*+10N_1]^b. 
%\end{align}

\end{proof}
In the following, we always assume $\bm\omega\in\widetilde \Omega_N$.  % so that all $\sigma$-bad $\Lambda\in (\bm k, \bm n)+\mathcal{ER}_{\bm 0}$
Under this condition, we  first  give a geometric construction of the  a region $\Lambda^\dagger$,  on which the worst  resonances appear.  It turns out that this region   has the {\bf annulus structure} similar to   \cite{Bou07, JLS20}. 
\begin{lem}\label{ty2lem}
Let $\bm\omega\in\widetilde\Omega_N$ and $\sigma\in\R$. Let $\bm k_*\in[-N, N]^b$ (if it exists)  be defined by  Lemma \ref{mellem}. Then there are some $\Lambda^{\dagger}_1\subset\Lambda^{\dagger}\subset \mathcal R_1$ satisfying ${\rm diam} \ \Lambda_1^\dagger\leq 4L_1,  \bm k_*\in\Lambda_1^\dagger$ and 
\begin{align*} 
\Lambda^{\dagger}&=((\bm k_*+[-L_3+O(L_1), L_3+O(L_1)]^b)\times [-L_2, L_2]^d)\cap \mathcal R_1,
%\Lambda^{\dagger}_1&= ((\bm k_*+[-L_1, L_1]^b)\times [-L_1, L_1]^d\times \{\pm\})\cap \mathcal R_1,
\end{align*}
so that the following holds true (recalling $L_1\ll L_2\ll L_3$ given by \eqref{L123}).    For any $\bm x\in \Lambda^{\dagger}\setminus \Lambda^{\dagger}_1$, there are some $L\in \{N_1, L_1, L_0\}$  and some $\mathcal{ER}(L)\ni W\subset\Lambda^{\dagger}\setminus \Lambda^{\dagger}_1$  so that  $$\bm x\in W,\  {\rm dist}(\bm x, \Lambda^{\dagger}\setminus \Lambda^{\dagger}_1\setminus W)\geq \frac{L}{2}.$$
%so that $W$ is $\sigma$-good. 

\end{lem}
\begin{proof}

The proof is similar to that of Theorem 3.7 of \cite{JLS20} (cf. pages 471--472).  We initially set   $Q^\dagger=(\bm k_*+[-L_3, L_3]^b)\times [-L_2, L_2]^d$ and $Q_1=(\bm k_*+[-L_1, L_1]^b)\times [-L_1, L_1]^d$.  
Then we  slightly  change   $Q^\dagger, Q_1^\dagger$ to $\Lambda^\dagger, \Lambda_1^\dagger$ with  the diameter  modulations  of length  $2L_1$  if the boundary   of $Q^\dagger, Q_1^\dagger$   approaches to the boundary of $\mathcal R_1$ with  the distance of $0<L\leq 2L_1$. 

% in the $\Z^b\ni \bm  k$-direction  when $|\bm k_*| \leq $ it  approaches the boundary of $\mathcal R_1$ with the distance of $2L_1$. 

Next, let  $\bm x=(\bm k, \bm n)\in \Lambda^{\dagger}\setminus\Lambda_1^{\dagger}$.   For $|\bm n|\leq 10N_1$, we use $L_1$  size regions  to cover $\bm x$; for $10N_1<|\bm n|\leq L_1$, we use the  $N_1$ size regions  to cover $\bm x$; for $L_1< |\bm n|\leq L_2$, we use the $L_0$ size regions to cover $\bm x.$
%when The key point is that we need to trim the $\$
\end{proof}

Next, we  show that  all $Q\in (\bm k, \bm 0)+\mathcal{ER}_{\bm 0}(L_1)$   satisfying $|\bm k|\leq N$ and $|\bm k-\bm k_*|>10N_1$ are $\sigma$-good for ``most'' $\sigma$, by using Cartan's lemma and coupling lemma of  \cite{Liu22}. More precisely, we have 
\begin{lem}\label{ty3lem}
Let $\bm\omega\in\widetilde \Omega_N\cap\Omega_{N_2}$. There is some $\tilde\Sigma_{N,2}\subset \R$ satisfying $${\rm meas}(\tilde\Sigma_{N,2})\leq e^{-N^{3\rho}},$$ so that for $\sigma\notin\tilde\Sigma_{N, 2}$,  the following holds true. If $Q\subset  \mathcal R_1$ satisfies 
$$Q\in \mathcal Q:=\bigcup_{|\bm k-\bm k_*|>10N_1, |\bm k|\leq N}\left\{(\bm k,\bm 0)+Q:\ Q\in\mathcal{ER}_{\bm 0}(L_1)\right\},$$
then $Q$ is $\sigma$-good with $\gamma_{L_1}=\gamma_{N_1}-N_1^{-\vartheta}$, where $\vartheta>0$ is an absolute constant defined in Lemma \ref{Liulem1}. 
\end{lem}
\begin{proof}[Proof of Lemma \ref{ty3lem}]
We choose  any $Q\in\mathcal{ER}(L_1)$ with $Q\in\mathcal Q$.   Then we take any $Q'\in \mathcal R_{L_*}^{{\sqrt{L_1}}}$ satisfying  $Q'\subset Q$  and $\sqrt{{L_1}}\leq L_*\leq 2L_1$. From Lemma \ref{lsclem1} and  $\bm\omega\in\widetilde\Omega_N$, we have for any $\sigma\in\R,$
$$\#\{\Lambda\in \mathcal{ER}(N_1):\ \Lambda\  {\rm is}\ \sigma{-\rm bad}, \  \Lambda\subset Q'\}\leq N^{\frac{1}{10^3(b+d)^2}}.$$
Denote by $Y_1(\sigma):=\bigcup\limits_{1\leq \ell \leq k} Q_{\bm x_\ell}, k\leq N^{\frac{1}{10^3(b+d)^2}}$ all those $\sigma$-bad $N_1$ size elementary regions. We claim that, there is some $Y=Y(\sigma)\subset Q'$ satisfying $\#Y\leq C(b,d)N_1^{b+d} N^{\frac{1}{10^3(b+d)^2}}$,  so that  if  $\bm x\in Q'\setminus Y$, then there is some $\sigma$-good $W\in\mathcal{ER}(N_1), W\subset Q'\setminus Y$ satisfying% the following: 
$$\bm x\in W, \ {\rm dist}(\bm x, Q'\setminus Y\setminus W)\geq \frac{N_1}{2}.$$
Indeed, we pave $Q'$ with $\Lambda_\alpha$  for  $\Lambda_\alpha=Q'\cap \Lambda_{10N_1}(\bm x), \bm x\in 10N_1\Z^{b+d}$. Then we can enlarge $\Lambda_\alpha$  to $\tilde \Lambda_\alpha\subset Q'$ so that $\tilde\Lambda_\alpha$ has width  at least $N_1$ and diameter at most $C(b,d)N_1$, and $\bigcup\limits_{\alpha }\tilde\Lambda_\alpha$ remains a tiling of $Q'$. It suffices to take $$Y(\sigma)=\bigcup\limits_{\alpha:\ \tilde \Lambda_\alpha\cap Y_1(\sigma)\neq\emptyset}\tilde \Lambda_\alpha.$$  Since each $Q_{\bm x_\ell}$ can only intersect with  at most $C(b,d)$ many $\tilde\Lambda_\alpha$, we have $$\#Y(\sigma)\leq C(b,d)N_1^{b+d}N^{\frac{1}{10^3(b+d)^2}}.$$ 

We are ready to apply Cartan's lemma (cf. Lemma \ref{mcl}). From Lemma \ref{lwlem},  any $\sigma$-good $N_1$ region   remains essentially  $\sigma'$-good  if $|\sigma'-\sigma|\leq e^{-10\gamma N_1}$. So on an interval $I$ of length at most  $e^{-10\gamma N_1}$, we can  fix   $Y$ to associate  with  the midpoint of $I$.  % independent of  $\sigma\in I.$  
In addition, we can assume $|\sigma|\leq C(b,d)N$ since if $|\sigma|>CN$,  then $\|D_{Q'}^{-1}(\sigma)\|\leq 1$. In this case,  $Q'$ becomes  $\sigma$-good via the Neumann series argument.   
%Next, we estimate the number of $\sigma$-bad elementary  $N_1$-regions satisfying \eqref{krestr} contained in $\mathcal{R}_1$.  This needs to control the $\bm n$-directions and it is ensured by $|\bm n |\leq N^{\frac{\frac{3}{4}}{10d}}$.  
%In conclusion, assuming  $\bm\omega\in\widetilde\Omega_N$,  we  have proved that for all $\sigma\in\R$,   the number of $\sigma$-bad elementary $N_1$-regions of centers  $(\bm k, \bm n)$ satisfying  $|\bm n |\leq N^{\frac{\frac{3}{4}}{10d}}$ and $|\bm k|>2N_1$ %$\Lambda$  with centers $(k,n)$ satisfying  $|k|\leq N$ and $2N_1<|n|\leq 10N$  
%is at most 
%\begin{align*}
%CN_1^{2b}N^{\frac{3}{4}/10}\leq N^{\frac{3}{4}/8}.
%\end{align*}
%if $\widetilde C^2\geq C/\frac{3}{4}$ for some $C=C(b,d)>0.$  
%Combining with  Lemma~\ref{lsclem1},  we get  actually  the sublinear bound  $N^{\frac{3}{4}/4}$ of  all $\sigma$-bad elementary $N_1$-regions  with centers belonging to $\mathcal{R}_1$.
%We are ready to apply Lemma \ref{mcl}.  % and assume $\sigma\notin\Sigma_{N_2}$ with $\Sigma_{N_2}$ being given by (2) of {\bf Claim 2}. 
%We first let 
%\begin{align*}
%\widetilde C^2\rho<\frac{3}{4},
%\end{align*}
%which implies 
%\begin{align*}
%e^{-\frac 23 N_1^{\frac{3}{4}}}< e^{-N^{\rho}}.
%\end{align*}
%So  it suffices to   fix $\Lambda(\tilde N)\in (0,n)+\mathcal{ER}_0(\tilde N)$ with $|n|\leq 2N_1.$ 
%Now let $\tilde N\in[\sqrt{N}, N]$ and $\Lambda(\tilde N)\in (0,n)+\mathcal{ER}_0(\tilde N)$ with $|n|\leq 3N$. 
We  apply  Lemma~\ref{mcl}
with  \begin{align*}
T(\sigma)={{H}}_{Q'}(\sigma), \ \beta=\beta_1=e^{- 10\gamma N_1}.
\end{align*}
It remains to verify the assumptions of Lemma \ref{mcl}.
%Obviously, $K_1=O(1)$.
First, one has by the definition of the operator $H(\sigma),$
\begin{align*}
K_1=CN, C=C(V,b,d,C_2)>0. %,  M=|\bar\Lambda|\leq C(b,d)\widetilde N^{1/{10}},
%&& ||G_{\Lambda\setminus \Lambda_1}(E;x)||\leq
%K_2=e^{2\sqrt{N_1}}.
\end{align*}
We let $Y$ be defined in the  above claim. 
%the union of all $\sigma$-bad $\Lambda$ satisfying  $\Lambda_*\in\mathcal{ER}(N_1)$ with centers belonging to $\mathcal{R}_1$. 
Then   using sublinear bound conclusion and  Lemma \ref{JLSlem1}, we have %by choosing  $\widetilde C\geq C/\rho$, 
\begin{align*}\label{mb}
&M=\#Y\leq   C(b,d)N^{\rho^2(b+d)}N^{\frac{1}{10^3(b+d)^2}}<N^{\frac{1}{999(b+d)^2}}\leq  L_*^{ \frac14}\\
& ({\rm since}\  \rho\leq  c_2\leq \frac{1}{10^4(b+d)^2},\  L_*\geq N^{\frac{1}{200(b+d)^2}}), \\ 
%&& ||G_{\Lambda\setminus \Lambda_1}(E;x)||\leq
&K_2=e^{2{N_1}^{\frac{3}{4}}}.
\end{align*}
To apply the matrix-valued Cartan's  lemma, we need the scale $$N_2=N_1^{\frac2\rho}.$$  Recall that the LDE hold  at  scale $N_2$. %for $y$ being outside a set  of measure at most $e^{-{N_2^{c_1}}}$.
%It implies  \eqref{ldt1newi} and \eqref{ldt2newi} hold at scale $N_2$ for all $\theta+k\omega$ and $|k| \leq N_3$ except a    set
%with measure less than $(2N_3+1)^de^{-{N_2^{c_1}}}$.
Applying the resolvent identity  Lemma \ref{JLSlem1} and using  (2) of Lemma \ref{lsclem2} with $L=N_2$ yield 
\begin{align*}
 \|T^{-1}(\sigma)\|  \leq e^{2{N_2}^{\frac{3}{4}}}=K_3 
\end{align*}
for $\sigma$ away from a set  of measure at most (since $0<\rho<c_2\kappa_1$)
$$CN^{C(b,d)}e^{-{N_2^{\rho}}}+e^{-\frac15{N_2^{\frac{3\kappa_1}{4}}}}\leq e^{-{N_2^{\rho}}/{2}}.$$
It follows from  $N_1=N_2^{\frac\rho2}$
 %assuming $$\widetilde C\rho>\frac{3}{4}$$ 
that
\begin{align*}
10^{-3}\beta_1(1+K_1)^{-1}(1+K_2)^{-1}&\geq e^{-10\gamma N_1}N^{-2} e^{-2N_1^{\frac34}}\\
&> e^{-{N_2^{\rho}}/{2}}. 
\end{align*}
%which is not sufficient to imply (iii). The argument does no work   !!!
This verifies (iii) of Lemma  \ref{mcl}.
If $\epsilon=e^{- {L_*}^{\frac{1}{2}}}$,  then one has   $\epsilon<(1+K_1+K_2)^{-10M}.$  Cover $[-CN,  CN]$ with disjoint intervals of length $e^{- 10\gamma N_1}$. % and take into account of all $\Lambda(\tilde N)$ with $\tilde N\in [\sqrt{N}, N]$.
Define 
\begin{align*}
\tilde \Sigma_{Q'}=\{\sigma\in\R:\ \|G_{Q'}(\sigma)\|\geq e^{L_*^{\frac{3}{4}}}\}.
\end{align*}
Then by \eqref{mc5} of Lemma \ref{mcl}, one obtains 
%since $\rho<\frac{3}{4}/5$ that 
\begin{equation*}
\mathrm{meas}(\tilde \Sigma_{Q'})\leq CN^Ce^{ 10\gamma N_1}e^{-\frac{c{{L_*}^{\frac{3}{4}}}}{CN_1N_2  L_*^{\frac{1}{4}}\log   N}}\leq e^{-{L_*}^{\frac13} }\ll  e^{-N^{3\rho}}.
\end{equation*}
Note that the number of $Q'\subset Q$ with   $Q'\in \mathcal{R}_{L_*}^{\sqrt{L_1}}, \sqrt{{L_1}}\leq L_*\leq 2L_1$ is at most $L^{C(b,d)}$ and the number of $Q\in\mathcal Q$ is at most $N^{(b,d)}$.  We take the  union on all these $Q'\subset Q, Q\in \mathcal Q$  leading to the desired $\tilde \Sigma_{N,2}\subset \R$ with (since $0<\rho<c_2\leq \frac{1}{10^4(b+d)^2}$)
\begin{align*}
{\rm meas} (\tilde \Sigma_{N,2})\leq N^{C(b,d)}e^{-L_1^{\frac{1}{6}}}<e^{-N^{\frac{1}{600(b+d)^2}}}\ll e^{-N^{3\rho}}. 
\end{align*}
%so that, if $\sigma\notin $\tilde \Sigma_{N,2}$,  all $QQ\in \mathcal Q$ 

To finish the proof, it suffices to prove the off-diagonal exponential  decay of $G_Q(\sigma)$ for $Q\in\mathcal Q$  and $\sigma\notin \tilde \Sigma_{N,2}$. This follows  from  using the coupling lemma of  \cite{Liu22}.  Indeed, fix $Q\in\mathcal Q$ (then $Q\in \mathcal{ER}(L_1)$). Let $\mathcal F$ be any family of pairwise disjoint $\sqrt{L_1}$ size elementary regions in $Q$. Then using Lemma \ref{lsclem1} and $\bm\omega\in\widetilde\Omega_N,$ we get 
$$ \# \{\Lambda\in\mathcal F:\  \Lambda \ {\rm is }\ \sigma{-\rm bad}\}\leq N^{\frac{1}{10^3(b+d)^2}}\leq L_1^{\frac19}.$$
So applying Lemma \ref{Liulem1}   leads   to the off-diagonal exponential decay of $G_Q(\sigma)$ with the decay rate 
$$\gamma_{L_1}=\gamma_{N_1}-N_1^{-\vartheta}.$$

%In addition, if $\sigma\notin\Sigma_N$, then  for all $\tilde N\in [\sqrt{N}, N]$,  $\Lambda\in (0,n)+\mathcal{ER}_0(\tilde N)$ with $|n|\leq 3 N$,
%\begin{align}\label{sublg}
%\|G_{\Lambda}(\sigma)\|&\leq e^{\tilde N^{\frac{3}{4}}}.
%|G_{\Lambda}(\sigma)(j;j')|&\leq e^{-\gamma_N|j-j'|}\ {\rm for}\ |j-j'|\geq N^{\frac{8}{9}},
%\end{align}

%Next,  for $\sigma\notin\Sigma_{N, 1}$,  it suffices to prove the exponential off-diagonal decay of $G_{\mathcal{R}_1}(\sigma)$.   This will be completed by using the coupling lemma of \cite{HSSY} (cf. Lemma 4.2).  %with $\tau=\frac12, b=\frac{3}{4}$.  The sublinear bound can be established as follows: Assume that  $\mathcal{F}$ is a (any) family  of pairwise disjoint bad $M$-regions in  $\Lambda$ with $\sqrt{N} + 1 \leq  M\leq  2\sqrt{N} + 1$,  then from Lemma~\ref{lsclem1}
%{\color {red} What is claim 1 here?? }{\bf Claim 1},  
%$\#\mathcal{F}\leq N^{\frac{3}{4}/4}\leq \frac{N^{\frac{3}{4}}}{M}  $ since $\frac{3}{4}>\frac23.$
%On the other hand, the sup-exponential growth of $\|G_\Lambda(\sigma)\|$ has been  given by \eqref{sublg}.  As a result, we have the exponential off-diagonal decay estimates with the decay rate $\gamma_N=\gamma-N^{-\kappa}. $

%Collecting all the restrictions on $\widetilde C, \rho, \frac{3}{4}$ yields 
%\begin{align*}
%1>\frac{3}{4}>2/3,\, 0<\rho<\min\{\frac15,\kappa_1\}\cdot\frac{3}{4},\ \widetilde C^2\geq \frac{C(b,d)}{\frac{3}{4}}, \ \widetilde C>\frac{\frac{3}{4}}{\rho}.
%\end{align*}

\end{proof}
%So it suffices to first  let  ${\frac{\frac{3}{4}}{\rho}}>C/\sqrt{\frac{3}{4}}$, i.e., 
%$$\rho<c\frac{3}{4},$$
%and then choose $\widetilde C \in({\frac{C}{\sqrt{\frac{3}{4}}}}, \frac{\frac{3}{4}}{\rho}).$

We then deal with  $G_{\Lambda^{\dagger}}(\sigma)$  with $\Lambda^{\dagger}$ given by Lemma \ref{ty2lem}.  We have 
\begin{lem}\label{ty2est}
Let $\bm\omega\in\widetilde\Omega_N\cap\Omega_{N_2}$. Then there is some $\widetilde\Sigma_{N, 3}\subset\R$ satisfying 
$${\rm meas}(\tilde\Sigma_{N,3})<e^{-N^{3\rho}} $$
so that,  if $\sigma\notin \bigcup\limits_{\ell=1, 2,3}\tilde\Sigma_{N,\ell},$ then $G_{\Lambda^{\dagger}}(\sigma)$ is $\sigma$-good in the sense of Definition \ref{sgmbd} (replacing $\Lambda$ with $\Lambda^\dagger$) with $\gamma_{L_2}=\gamma_{L_1}-N_1^{-\frac{1}{10}}$. 
\end{lem}
\begin{proof}[Proof of Lemma \ref{ty2est}]
The proof is similar to that of Lemma \ref{ty3lem},  but is more in the spirit of \cite{JLS20} when dealing with off-diagonal exponential decay estimate, as  $\Lambda ^{\dagger}$ has the  special annulus structure.

To get the sub-exponential growth estimate of $\|G_{\Lambda^\dagger}(\sigma)\|$, the key point is  to apply   Cartan's lemma only involving     $N_1, N_2$ sizes regions similar to the proof of Lemma \ref{ty3lem}.  This means  precisely that we do not  choose  $Y=\Lambda_1^\dagger$  when applying  Lemma \ref{mcl}. 

From  $\bm\omega\in\widetilde\Omega_N$, the definition  of $\Lambda^\dagger$ (cf. Lemma \ref{ty2lem}) and  Lemma \ref{lsclem1}, we have for any $\sigma\in\R,$
\begin{align*}
&\ \ \ \#\{\Lambda\in \mathcal{ER}(N_1):\ \Lambda\  {\rm is}\ \sigma{-\rm bad}, \  \Lambda\subset \Lambda^\dagger\}\\
&\leq N^{\frac{1}{10^3(b+d)^2}}+ C(b,d)N_1^{b}L_2^{d}\\
&\leq N^{\frac{1}{10^3(b+d)^2}}+C(b,d)N^{\rho^2b+\frac{1}{50(b+d)}}\\
&\leq 2N^{\rho^2b+\frac{1}{50(b+d)}}. 
\end{align*}
Similar to the proof of Lemma \ref{ty3lem}, we can find some $Y=Y(\sigma)\subset \Lambda^\dagger$ satisfying $\#Y\leq C(b,d)N^{2\rho^2(b+d)} N^{\frac{1}{50(b+d)}}<N^{\frac{1}{49(b+d)}}$ so that,   if  $\bm x\in \Lambda^\dagger\setminus Y$, then there is some $\sigma$-good $W\in\mathcal{ER}(N_1), W\subset \Lambda^\dagger\setminus Y$ satisfying% the following: 
$$\bm x\in W,\  {\rm dist}(\bm x, \Lambda^\dagger\setminus Y\setminus W)\geq \frac{N_1}{2}.$$ 
%Indeed, we pave $Q'$ with $\Lambda_\alpha$  for  $\Lambda_\alpha=Q'\cap \Lambda_{10N_1}(\bm x), \bm x\in 10N_1\Z^{b+d}$. Then we can enlarge $\Lambda_\alpha$  to $\tilde \Lambda_\alpha\subset Q'$ so that $\tilde\Lambda_\alpha$ has width of at least $N_1$ and diameter at most $C(b,d)N_1$, and $\bigcap_{\alpha: }\tilde\Lambda_\alpha$ remains a tiling of $Q'$. It suffices to take $Y=\bigcup_{\alpha:\ \tilde \Lambda_\alpha\cap Y_1(\sigma)\neq\emptyset}\tilde \Lambda.$ Since each $\Lambda(\bm x_\ell)$ can only intersect with $C(b,d)$ many $\tilde\Lambda_\alpha$, we have $\#Y\leq C(b,d)N_1^{b+d}N^{\frac{1}{10^3d^2}}.$ 

Let $\bm\omega\in\widetilde\Omega_N$. Again, both the site $\bm k_*\in[-N, N]^b$ (given in Lemma \ref{mellem})  and  $Y(\sigma)$  depend  on $\sigma$. From Lemma \ref{lwlem}, any $\sigma$-good $N_1$ size region  remains essentially  $\sigma'$-good  if $|\sigma'-\sigma|\leq e^{-10\gamma N_1}$.  So on an interval $I$ of length at most  $e^{-10\gamma N_1}$, we can choose $Y,  \bm k_*$ to  associate  with the midpoint of $I$.   Also, we can assume $|\sigma|\leq C(b,d)N$ since if $|\sigma|>CN$,  then $\|D_{\Lambda^\dagger}^{-1}(\sigma)\|\leq 1$, and  $\Lambda^\dagger$ becomes  $\sigma$-good via  the Neumann series argument.

We  now apply  Lemma \ref{mcl}
with  \begin{align*}
T(\sigma)={{H}}_{\Lambda^\dagger}(\sigma), \ \beta=\beta_1=e^{- 10\gamma N_1}.
\end{align*}
%It remains to verify the assumptions of Lemma \ref{mcl}.
%Obviously, $K_1=O(1)$.
We have $
K_1=C N.%,  M=|\bar\Lambda|\leq C(b,d)\widetilde N^{1/{10}},
%&& ||G_{\Lambda\setminus \Lambda_1}(E;x)||\leq
%K_2=e^{2\sqrt{N_1}}.
$ 
%We let $Y=\Lambda_1^\dagger$ be given by Lemma \ref{ty2lem}.    
%the union of all $\sigma$-bad $\Lambda$ satisfying  $\Lambda_*\in\mathcal{ER}(N_1)$ with centers belonging to $\mathcal{R}_1$. 
Using  the sublinear bound conclusion  and   the resolvent identity Lemma \ref{JLSlem1}, we have %by choosing  $\widetilde C\geq C/\rho$, 
\begin{align*}\label{mb}
M=\#Y&\leq   N^{\frac{1}{49(b+d)}}< L_3^{\frac14},\ K_2=e^{2{N_1}^{\frac{3}{4}}}. %N^{\rho^2(b+d)}N^{\frac{1}{10^3d^2}}\\
\end{align*}
%From Lemma \ref{ty2lem} and Lemma \ref{ty3lem},  we can cover $\Lambda^\dagger\setminus \Lambda_1^\dagger$ with $\sigma$-good regions of sizes  $N_1, L_0, L_1, L_1/2$ and use  the resolvent identity Lemma 3.2 of \cite{JLS20}  to  get 
%\begin{align*}
%&<N^{\frac{1}{999d^2}}\leq (L')^{ \frac14}\ ({\rm since}\  \rho\leq  c_2\ll \frac{1}{10^3d^2}, L'\geq N^{\frac{1}{200d^2}}), \\ 
%&& ||G_{\Lambda\setminus \Lambda_1}(E;x)||\leq
%\end{align*}
We again use the scale $N_2=N_1^{\frac2\rho}$, at which the LDE hold. %for $y$ being outside a set  of measure at most $e^{-{N_2^{c_1}}}$.
%It implies  \eqref{ldt1newi} and \eqref{ldt2newi} hold at scale $N_2$ for all $\theta+k\omega$ and $|k| \leq N_3$ except a    set
%with measure less than $(2N_3+1)^de^{-{N_2^{c_1}}}$.
Similar to  the proof of Lemma \ref{ty3lem}, we can verify all assumptions of Lemma \ref{mcl}. 
%Applying the resolvent identity  and using  (2) of Lemma \ref{lsclem2} with $L=N_2$ yield 
%\begin{align*}
% \|T^{-1}(\sigma)\|  \leq e^{2{N_2}^{\frac{3}{4}}}=K_3
%\end{align*}
%for $\sigma$ away from a set  of measure at most (since $0<\rho<c_2\kappa_1$)
%$$CN^{C(b,d)}e^{-{N_2^{\rho}}}+e^{-\frac15{N_2^{\frac{3\kappa_1}{4}}}}\leq e^{-{N_2^{\rho}}/{2}}.$$
%It follows from  $N_1=N_2^{\frac\rho2}$
 %assuming $$\widetilde C\rho>\frac{3}{4}$$ 
%that
%\begin{align*}
%10^{-3}\beta_1(1+K_1)^{-1}(1+K_2)^{-1}&\geq ce^{-10\gamma N_1}N^{-1} e^{-2N_1^{\frac34}}\\
%&> e^{-{N_2^{\rho}}/{2}}. 
%\end{align*}
%which is not sufficient to imply (iii). The argument does no work   !!!
%This verifies (iii) of Lemma \ref{mcl}.
%If $\epsilon=e^{- L_3^{\frac{1}{2}}}$,  then one has   $\epsilon<(1+K_1+K_2)^{-10M}.$  Cover $[-CN,  CN]$ with disjoint intervals of length $e^{- 10\gamma N_1}$. % and take into account of all $\Lambda(\tilde N)$ with $\tilde N\in [\sqrt{N}, N]$.
Define 
\begin{align*}
\tilde \Sigma_{N,3}=\{\sigma\in\R:\ \|G_{\Lambda^\dagger}(\sigma)\|\geq e^{L_3^{\frac{3}{4}}}\}.
\end{align*}
Then by \eqref{mc5} of Lemma \ref{mcl}, one obtains 
%since $\rho<\frac{3}{4}/5$ that 
\begin{equation*}
\mathrm{meas}(\tilde \Sigma_{N,3})\leq CN^Ce^{ 10\gamma N_1}e^{-\frac{c{{L_3}^{\frac{3}{4}}}}{CN_1N_2  L_3^{\frac{1}{4}}\log   N}}\leq e^{-{L_3}^{\frac13} }\ll  e^{-N^{3\rho}}.
\end{equation*}
%Note that the number of $Q'\subset Q$ with  and $Q'\in \mathcal{R}_{L'}^{\sqrt{L_1/2}}, \sqrt{\frac{L_1}{2}}\leq L'\leq 2L$ is at most $L^{C(b,d)}$ and the number of $Q\in\mathcal Q$ is at most $N^{(b,d)}$.  We take union on all these $Q'\subset Q, Q\in \mathcal Q$ leading to the desired $\tilde \Sigma_{N,2}\subset \R$ with (since $\rho<c_2\ll\frac{1}{10^3d^2}$)
%\begin{align*}
%{\rm meas} (\tilde \Sigma_{N,2})\leq N^{C(b,d)}e^{L_1^{\frac{1}{8}}}<e^{-N^{\frac{1}{900d^2}}}\ll e^{-2N^{\rho}}. 
%\end{align*}
%so that, if $\sigma\notin $\tilde \Sigma_{N,2}$,  all $QQ\in \mathcal Q$ 

Finally, it suffices to prove the  off-diagonal exponential decay of $G_{\Lambda^\dagger}(\sigma).$   This follows  directly from the resolvent identity of \cite{JLS20}, i.e.,   Lemma \ref{JLSlem2}.  In fact,   by the analysis in $\mathcal R_2$, Lemma \ref{ty3lem} and Lemma \ref{ty2lem}, all elementary regions contained in $\Lambda^\dagger\setminus\Lambda_1^\dagger$ with sizes $N_1, L_0, L_1$ are $\sigma$-good. Then applying Lemma \ref{JLSlem2} with $\Lambda=\Lambda^\dagger, \Lambda_1=\Lambda_1^\dagger$ (since ${\rm diam}\ \Lambda^\dagger\sim L_3>100^{2(b+d)} L_1^{2(b+d)}\geq  ({\rm diam}\ \Lambda_1^\dagger)^{2(b+d)}$)  leads to the desired off-diagonal exponential decay estimate of $G_{\Lambda^\dagger}(\sigma)$ with 
$$\gamma_{L_2}=\gamma_{L_1}-N_1^{-\frac{1}{10}}=\gamma_{N_1}-N_1^{-\vartheta}-N_1^{\frac {1}{10}}.$$

% using the coupling lemma of Liu \cite{Liu22} combined with a small modification (i.e., changing the  decay scale from $\frac {L} {10}$ to $L^{\frac89}$, similar to Lemma 4.2 of \cite{HSSY}).  The condition (4) of Theorem 2.3 \cite{Liu22} is satisfied since $\sigma\notin\tilde\Sigma_{N, 2}$. It remains to verify the  sublinear bound condition (5)  of Theorem 3.2 \cite{Liu22} with $\zeta =\frac{11}{18}$. Indeed, fix $Q\in\mathcal Q$ with $Q\in \mathcal{ER}(L)$ for some $L=L_1, L_1/2$. Let $\mathcal F$ be any family of pairwise disjoint elementary regions in $Q$. Then using Lemma \ref{lsclem1} and $\bm\omega\in\widetilde\Omega_N,$ we get 
%$$ \# \{\Lambda\in\mathcal F:\  \Lambda \ {\rm is }\ \sigma{\rm\  -bad}\}\leq N^{\frac{1}{10^3d^2}}\leq L^{\frac19}= \frac{L^{\frac {11}{18}}}{\sqrt{L}}.$$

%So applying Theorem 2.3 of \cite{Liu22} with  $\xi=\frac12$   leads   to the exponential decay of $G_Q(\sigma)$.

\end{proof}

\ \\

\noindent{\large \bf Application of the resolvent identity again} 
\ \\ 
\begin{figure}[htbp]
		\centering
		\vspace*{1cm}
		\begin{tikzpicture}[scale=1.05, every node/.style={transform shape}]
			
			% ================= 定义样式 =================
			% 坏块（阴影填充）
			\tikzset{bad block/.style={pattern=north west lines, pattern color=red, draw=red, very thick}}
			% 空心，黄色边框
			\tikzset{yellow block/.style={draw=orange, very thick, fill=none}}
			% 空心，红色边框
			\tikzset{hollow red/.style={draw=red, very thick, fill=none}}
			%蓝色边框
			\tikzset{hollow blue/.style={draw=blue, very thick, fill=none}}
			%蓝色阴影
			\tikzset{blue block/.style={pattern=north west lines, pattern color=blue, draw=blue, very thick}}
			%绿色边框
			\tikzset{green block/.style={draw=green!50!black, very thick, fill=none}}
			
			% ================= 背景与网格 =================
			% 整体大外框
			\draw[very thick, black] (-6, -7) rectangle (6, 8);
			% 隐约的网格线（可选，用来辅助视觉，模拟方格纸）
%			\draw[step=0.2cm, gray!20, very thin] (-6,-5) grid (6,5);
			
			% ================= 核心区域划分 =================
			%竖直条带 (宽 2N_1)
			\draw[-, thick, black] (-0.2, -7) -- (-0.2, 8);
			\draw[-, thick, black] (0.2, -7) -- (0.2, 8);
			% 水平条带 (宽 2N_1)
			\draw[-, thick, black] (-6, -0.2) -- (6, -0.2);
			\draw[-, thick, black] (-6, 0.2) -- (6, 0.2);
			% 右侧 2N_1 标注
			\draw[decorate, decoration={brace, amplitude=6pt, mirror}, thick, black] 
			(6.1, -0.2) -- (6.1, 0.2) 
			node[midway, right=4pt] {$10N_1$};
			
			% 垂直条带 (宽 L_1)
			\draw[dashed, very thick, orange] (-1, -7) -- (-1, 8);
			\draw[dashed, very thick, orange] (1, -7) -- (1, 8);
			
			% 底部 L_1 标注
			\draw[decorate, decoration={brace, amplitude=6pt, mirror}, thick, green!60!black] 
			(-1, -7.1) -- (1, -7.1) node[midway, below=8pt] {$L_1 \sim N^{\frac{1}{100(b+d)^2}}$};

			% 垂直条带 (宽 L_2)
			\draw[dashed, thick, green!60!black] (-3, -7) -- (-3, 8);
			\draw[dashed, thick, green!60!black] (3, -7) -- (3, 8);
			% 顶部绿色大括号
			\draw[decorate, decoration={brace, amplitude=6pt}, thick, blue, overlay] 
			(-3, 8.1) -- (3, 8.1)
			node[midway, above=8pt, overlay] {$L_2 \sim N^{\frac{1}{50(b+d)^2}}$};
			
			% ================= 中心蓝框与阴影区 (L_3) =================
			% 中心蓝色阴影区（十字交叉处）
			\fill[blue block] (-1, -1) rectangle (1, 1);
			
			% L_3 蓝色大外框
			\draw[hollow blue] (-3, -4) rectangle (3, 4);
%			\draw[->, blue, thick] (-3, -3) node[below] {$L_3 \sim N^{\frac{1}{100d}}$} -- (-2.5, -2);

			% 左侧 L_3标注
			\draw[dashed, thick, blue] (-6, -4) -- (-3, -4);
			\draw[dashed, thick, blue] (-6, 4) -- (-3, 4);
			\draw[decorate, decoration={brace, amplitude=6pt}, thick, blue, overlay] 
			(-6.1, -4) -- (-6.1, 4) 
			node[midway, left=2pt, overlay] {$L_3 \sim N^{\frac{1}{10(b+d)}}$};
			
			% 绿色框 (垂直条带上半部分)
			\draw[green block] (-1, 1) rectangle (1, 3);
			\draw[green block] (-1, 4.5) rectangle (1, 6.5);
			\draw[green block] (-1, -3.5) rectangle (1, -1.5);
			\draw[green block] (-1, -4.7) rectangle (1, -6.7);
			
			% --- 垂直条带 ---
			\fill[bad block] (-0.1, 7) rectangle (0.1, 7.2);
			\draw[hollow red] (0.45, 7) rectangle (0.65, 7.2);
			\fill[bad block] (-0.1, 5) rectangle (0.1, 5.2);
			\fill[bad block] (-0.1, 2) rectangle (0.1, 2.2);
			\fill[bad block] (-0.1, -0.1) rectangle (0.1, 0.1);
			\fill[bad block] (-0.1, -2.4) rectangle (0.1, -2.2);
			\fill[bad block] (-0.1, -4.5) rectangle (0.1, -4.3);
			\fill[bad block] (-0.1, -6.5) rectangle (0.1, -6.3);
			\fill[bad block] (0.2, -0.1) rectangle (0.4, 0.1);
			\fill[bad block] (0.5, -0.1) rectangle (0.7, 0.1);
			\fill[bad block] (0.2, -0.1) rectangle (0.4, 0.1);
			\fill[bad block] (0.75, -0.1) rectangle (0.95, 0.1);
			\fill[bad block] (-0.4, -0.1) rectangle (-0.2, 0.1);
			\fill[bad block] (-0.7, -0.1) rectangle (-0.5, 0.1);
			\fill[bad block] (-0.95, -0.1) rectangle (-0.75, 0.1);
			% --- 交叉区域周边及内部 ---
%			\draw[yellow block] (-2.0, 0.1) rectangle (-1.0, 1.1);
%			\draw[yellow block] (1.0, 0.1) rectangle (2.0, 1.1);
			
			% 交叉区上边缘的三个红点
			\fill[red] (-0.68, 1) circle (1.5pt);
			\draw[hollow red] (-0.78, 1) rectangle (-0.58, 1.2);
			\fill[red] (0, 1) circle (1.5pt);
			\fill[red] (1, 1) circle (1.5pt);
			%基本块
			\draw[hollow red] (0.9, 1.1) -- (1.1, 1.1) -- (1.1, 0.9) -- (1.0, 0.9) -- (1.0, 1.0) -- (0.9, 1.0) -- cycle;

			% ================= 绘制好块与坏块 =================
			
			% --- 水平条带左侧 ---
			\draw[yellow block] (-5, -0.7) rectangle (-3.6, 0.7);
			\fill[bad block] (-5.5, -0.1) rectangle (-5.3, 0.1);
			\fill[bad block] (-4.9, -0.1) rectangle (-4.7, 0.1);
			\fill[bad block] (-4.5, -0.1) rectangle (-4.3, 0.1);
			\fill[bad block] (-4, -0.1) rectangle (-3.8, 0.1);
			\fill[bad block] (-3.5, -0.1) rectangle (-3.3, 0.1);
			\fill[bad block] (-2.5, -0.1) rectangle (-2.3, 0.1);
			\fill[bad block] (-2, -0.1) rectangle (-1.8, 0.1);
			\fill[bad block] (-1.5, -0.1) rectangle (-1.3, 0.1);

			\draw[yellow block] (-2.65, -0.7) rectangle (-1.25, 0.7);
			
			% --- 水平条带右侧 ---
			\draw[yellow block] (5, -0.7) rectangle (3.6, 0.7);
			\fill[bad block] (4.9, -0.1) rectangle (4.7, 0.1);
			\fill[bad block] (4.5, -0.1) rectangle (4.3, 0.1);
			\fill[bad block] (4, -0.1) rectangle (3.8, 0.1);
			\fill[bad block] (3.5, -0.1) rectangle (3.3, 0.1);
			\fill[bad block] (2.5, -0.1) rectangle (2.3, 0.1);
			\fill[bad block] (2, -0.1) rectangle (1.8, 0.1);
			\fill[bad block] (1.5, -0.1) rectangle (1.3, 0.1);
			\draw[yellow block] (2.65, -0.7) rectangle (1.25, 0.7);
			
			\fill[bad block] (5.2, -0.05) rectangle (5.4, 0.15);
			\fill[bad block] (5.7, -0.15) rectangle (5.9, 0.05);

			% ================= 右上角图例与底部公式 =================
			% 右上角 L_* 模板好块
			\draw[yellow block] (3.6, 4.8) rectangle (5, 6.2);
			\fill (4.3, 5.5) circle (1pt); % 中心点
			\draw[decorate, decoration={brace, amplitude=4pt, mirror}, thick, orange] 
			(3.6, 4.7) -- (5, 4.7) node[midway, below=5pt] {$L_0\sim L_1^{\frac{9}{10}}$};
			
			% 图例
			\draw[->,thick,overlay] (-8.5, 6) -- (-6.5,6) node[right] {$\bm n$};
			\draw[->,thick,overlay] (-7.5,5) -- (-7.5,7) node[above] {$\bm k$};

			\node[text=white!20!black, font=\Large] at (4, -8) {$\log N_1 \ll \log N$};

		\end{tikzpicture}
		\caption{The distribution of resonant blocks}
		\label{graph}
	\end{figure}

Finally, we define $\Sigma_{N}=\bigcup\limits_{\ell=1,2,3}\tilde\Sigma_{N, \ell}$ with $\tilde\Sigma_{N, 1}, \tilde\Sigma_{N, 2}, \tilde\Sigma_{N, 3}$ given by the analysis in $\mathcal R_2$, Lemma \ref{ty3lem}, Lemma \ref{ty2est}, respectively. Then $${\rm meas}(\Sigma_N)\leq 3e^{-2N^{3\rho}}\ll e^{-N^{2\rho}},$$   and  for $\sigma\notin\Sigma_N$, we can cover $\Lambda_{{\rm pm}, N}$ with $\sigma$-good regions of sizes $N_1, L_0, L_1, L_3$. More precisely, we have  (cf. Figure \ref{graph})
\begin{itemize}
\item For $\bm x\in \mathcal R_2$,  cover $\bm x$ with $\Lambda\in \mathcal{ER}(L_0)$ with $L_0\sim N^{\frac{9}{1000(b+d)^2}}$.
\item For $\bm x\in\Lambda_1^\dagger$, cover $\bm x$ with $\Lambda^\dagger$. 
\item For $\bm x\in \mathcal R_1\setminus\Lambda_1^\dagger$, cover $\bm x$ with regions in 
$$\bigcup_{L=N_1, L_0, L_1}\mathcal{ER}(L).$$
%\item For $\bm x=(\bm k, \bm n)\in\mathcal R_1$ with $|\bm k-\bm k_*|>10N_1$, we cover $\bm x$ with regions in 
\end{itemize}
As a result, for $\sigma\notin\Sigma_N$, we can apply Lemma \ref{JLSlem1} and Lemma \ref{JLSlem2} to show that  $\Lambda_{{\rm pm}, N}$ is $\sigma$-good. Indeed,  we first apply Lemma \ref{JLSlem1} with $$\Lambda=\Lambda_{{\rm pm}, N}, W\in \bigcup_{L=N_1, L_0, L_1}(\mathcal{ER}(L)\cup\{\Lambda^\dagger\})$$ to obtain 
$$\|G_{\Lambda_{{\rm pm}, N}}(\sigma)\|\leq e^{N^{\frac34}}.$$
Next, for the off-diagonal exponential decay, we use Lemma \ref{JLSlem2} with $\Lambda=\Lambda_{{\rm pm}, N},\ \Lambda_1=\emptyset$, $M_0=N_1$ to obtain (recalling the decay  rate\ estimates   in  Lemmas \ref{ty3lem}, \ref{ty2est}) for $|\bm x-\bm x'|>N^{\frac89},$
\begin{align*}
|G_{\Lambda_{{\rm pm}, N}}(\sigma)((\bm x,\xi);(\bm x',\xi'))|&\leq e^{-(\gamma_{L_2}-N_1^{-\frac{1}{10}})|\bm x-\bm x'|}\\
&\leq e^{-(\gamma_{N_1}-N_1^{-\vartheta}-2N_1^{-\frac{1}{10}})  |\bm x-\bm x'|}\\
% ({\rm recalling\  the\  decay\  rate\  estimates \  in \ Lemmas} \ \ref{ty2est},\  \ref{ty3lem})\\
&=e^{-\gamma_{N} |\bm x-\bm x'|},
\end{align*}
where 
\begin{align}\label{kappa}
\gamma_N\geq \gamma_{N_1}-N^{-\kappa},\ \kappa={\frac12}\min\{\vartheta, \frac{1}{10}\}\leq \frac{1}{20}.
\end{align}
Since $N\sim N_1^{\frac{1}{\rho^2}}$,  we choose $\ell_*$ so that $N^{\rho^{2\ell_*}}:=N_*\sim \log^{\frac1\kappa}\frac{1}{\varepsilon+\delta}$. It  follows from applying  Lemma \ref{inilem}   (i.e., $\gamma_{N^{\rho^{2\ell_*}}}=\gamma_0=\gamma-\log^{-2}\frac{1}{\varepsilon+\delta}$)  and iterating \eqref{kappa} that 
\begin{align}
\nonumber\gamma_N&\geq \gamma_{N_*}-\sum_{\ell=1}^{\ell_*}N_*^{-\frac{1}{\rho^{2\ell}}\kappa}\\
\nonumber&\geq \gamma_{N^{\rho^{2\ell_*}}}-N_*^{-\frac{\kappa}{\rho^2}}-\sum_{\ell\geq 2}N_*^{-\frac{1}{\rho^{2\ell}}\kappa}\\
\nonumber&\geq \gamma-\log^{-2}{\frac{1}{\varepsilon+\delta}}-2N_*^{-\frac{\kappa}{\rho^2}}\\%C(\rho, \kappa)(\varepsilon+\delta)^{-c_1\kappa}\\
\label{gamitr}&\geq   \gamma-3\log^{-2}{\frac{1}{\varepsilon+\delta}}\geq\frac\gamma 2
\end{align}
provided $0<\varepsilon+\delta\leq c(\rho, c_1, \kappa).$
%where in the last  $\gamma_0$

%Finally, we cover  points $(\bm k, \bm n, \xi )\in\Lambda_{{\rm pm}, N}$ satisfying $|\bm n|\leq N^{\frac{\frac{3}{4}}{20d}}$ with $\mathcal{R}_1$, points satisfying $|\bm n|> N^{\frac{\frac{3}{4}}{20d}}$  with $\sigma$-good $L$-regions.  Applying Theorem 3.3 of \cite{JLS20} yields $\Lambda_{{\rm pm}, N}$ is $\sigma$-good for $\sigma$ outside a set $\Sigma_N=\Sigma_{N,1}\cup\Sigma_{N,2}$ of measure at most $e^{-N^{\rho}}.$\\

The above arguments hold indeed for all scales in $[N, N^2]$. Then we can   propagate the LDE from small scales interval $[N_1, N_2]$ to large  scales one $[N, N^2]$ similar to  the proof of Theorem 4.1 in   \cite{JLS20}.   This completes the  proof of Lemma \ref{lsclem}. % at scale $N$. 

 %(and then at all $N\geq (\varepsilon+\delta)^{-c_1}$ :  ).  
\end{proof}

\begin{proof}[Proof of Theorem~\ref{ldtthm}]
%The above arguments hold indeed for all scales in $[N, N^2]$. Then we can   propagate the LDE from small scales interval $[N_1, N_1^{\frac{2}{\rho}}]$ to large  scales one $[N, N^2]$ similar to  the proof of Theorem 4.1 in   \cite{JLS20}.  
The proof is just a combination of Lemmas~\ref{inilem},~\ref{intlem} ~and~\ref{lsclem} by taking 
$$0<\varepsilon+\delta\leq\delta_0=\min\{\delta_1, \delta_2, \delta_3\}.$$

 %and Lemma \ref{lsclem}.
\end{proof}

\section{Nonlinear Analysis} \label{nonsect}
In this section, we focus on nonlinear analysis.  
%We use a Lyapunov-Schmidt decomposition to decompose the nonlinear equation into the $P$ and the $Q$ equations. 
%We iteratively solve the $Q$ and then the $P$-equations.
	We work directly with the  nonlinear equation \eqref{NLS}.   Recall that
\begin{align*}
%D={\rm diag}(\mu^{(0)}_{\bm n})_{\bm n\in\Z^d}\ {\rm with}\ \
\mu_{\bm n}=V(\bm n\bm \alpha+\bm \theta),\ \bm\omega^{(0)}=\left(\omega^{(0)}_\ell=\mu_{\bm n_\ell}\right)_{\ell=1}^b. 
\end{align*}
We aim to construct solutions to \eqref {NLS} of the form
\begin{align*} 
	u(t,\bm n)=\sum_{(\bm k, \bm n)\in\Z^b\times\Z^d}\hat u (\bm k,\bm n) e^{i\bm k\cdot\bm \omega t}.
	\end{align*} 
Using the Ansatz, we get  the nonlinear lattice equation 
\begin{align}\label{nle1}
{\rm diag}_{(\bm k, \bm n)}(-\bm k\cdot\bm\omega+\mu_{\bm n})\hat u +\varepsilon\Delta \hat u +\delta (\hat u*\hat v)^{*p} *\hat u=0,
\end{align}
where $\hat v(\bm k,\bm n)=\overline{\hat u(-\bm k, \bm n)}$,  the convolution is only in the $\bm k$-variable: 
$$\hat u*\hat v(\bm k, \bm n)=\sum_{\bm k_1+\bm k_2=\bm k}\hat u(\bm k_1, \bm n)\hat v(\bm k_2, \bm n)$$
and 
$$ (\hat u*\hat v)^{*p} =\underbrace{(\hat u*\hat v)*\cdots*(\hat u*\hat v)}_{p\ {\rm times}}.$$
%\begin{align}\label{involu}
%q_*^{p+1}(k,n)=\sum_{k^{(1)}+\cdots+ k^{(p+1)}=k}\prod_{l=1}^{p+1}q(k^{(l)}, n),
%\end{align} 
Note that  $\Delta$ acts only on the $\bm n$-variable.  To solve for $\hat u, \hat v$, we also need a conjugate equation of \eqref{nle1}:
\begin{align*}
{\rm diag}_{(\bm k, \bm n)}(\bm k\cdot\bm\omega+\mu_{\bm n})\hat v +\varepsilon\Delta \hat v +\delta (\hat u*\hat v)^{*p} *\hat v=0.
\end{align*}
As a result, we obtain the following vector-valued nonlinear lattice equation: 
\begin{align}\label{nonlatt}
F(\vec u)=0,\  \vec u=\binom{\hat u}{\hat v}, 
\end{align}
where 
\begin{align*}
F(\vec u)= H_0\vec u+\delta P(\vec u)
\end{align*}
with 
\begin{align}\label{Ldef}
\notag H_0&={\rm diag}_{(\bm k,\bm n)\in\Z^{b+d}}\left(\begin{array}{cc}
	{-\bm k\cdot\bm\omega+\mu_{\bm n}+\varepsilon \Delta} & {0} \\
	{0}& {\bm k\cdot\bm\omega+\mu_{\bm n}+\varepsilon \Delta} \\
	\end{array} \right) \\
	&:=D(0)+\varepsilon(\Delta\oplus\Delta)
\end{align}
and 
\begin{align*}
P(\vec u)= \binom{(\hat u*\hat v)^{*p} *\hat u}{(\hat u*\hat v)^{*p} *\hat v}. 
\end{align*}
The linearized operator of $ P$ at  $\vec u$  is then given by 
\begin{align}\label{toep}
&{T}_{\vec u }=\left(\begin{array}{cc}
	{(p+1)(\hat u*\hat v)^{*p} *} & {p(\hat u*\hat v)^{*(p-1)} *\hat u*\hat u*} \\
	{p(\hat u*\hat v)^{*(p-1)} *\hat v*\hat v*}& {(p+1)(\hat u*\hat v)^{*p} *} \\
	\end{array} \right).
\end{align}

\subsection{The Lyapunov-Schmidt decomposition}
We look for quasi-periodic  solutions to \eqref{NLS} near $$u^{(0)}=u^{(0)}(t,\bm n)=\sum_{\ell=1}^ba_\ell e^{i\omega_\ell^{(0)}t}\delta_{\bm n_{\ell}, \bm n},$$
using the Lyapunov-Schmidt decomposition.  
%In view of \eqref{nonlatt}, we  study on   the lattice $\Z^{b+d}_{{\rm pm}}.$
The Fourier support of $u^{(0)}$  is  $\mathcal{S}_+=\{(\bm e_\ell, \bm n_\ell)_{\ell=1}^b\}\subset\Z^{b+d}$   (we  identify $ \mathcal{S}_+$ with $\{(\bm e_\ell, \bm n_\ell)_{\ell=1}^b\}\times\{+\}\subset \Z^{b+d}_{{\rm pm}}$). Similarly, let $\mathcal{S}_-=\{(-\bm e_\ell, \bm n_\ell)_{\ell=1}^b\}\subset \Z^{b+d}$ be the Fourier support of $v^{(0)}:=\overline{u^{(0)}}$, which is then identified with $\{(-\bm e_\ell, \bm n_\ell)_{\ell=1}^b\}\times\{-\}\subset \Z^{b+d}_{{\rm pm}}$.  
Let $\mathcal{S}=\mathcal{S}_+\cup \mathcal{S}_-$ 
and write  $\Z^{b+d}_{{\rm pm}, *}=\Z^{b+d}_{{\rm pm}}\setminus\mathcal{S}:=\mathcal S^c.$ 

We divide the equation \eqref{nonlatt} into the $P$-equations 
\begin{align}\label{Peq}
{F}\left(\vec u\right)\big |_{\mathcal S^c}=0,\
%text{ where } \mathcal{S}^c=\Z^{b+d}_{{\rm pm}}\setminus \mathcal{S}, 
\end{align}
and the $Q$-equations 
\begin{align}\label{Qeq}
{F}(\vec u)\big |_{\mathcal{S}}=0.
\end{align}

%the domain of the $Q$-equations, 
%Then we solve the finitely dimensional $Q$-equation  \eqref{Qeq} for $q\big |_{\mathcal{S}}$ via the implicit function theory.  

	%We prove the existence of nonlinear Anderson localized states using the Lyapunov-Schmidt decomposition as written
	%in sect. 2.1, which divides the nonlinear matrix equation in \eqref {nonlatt},
	%\begin{align*}
%F(\vec u)=0,\  \vec u=\binom{\hat u}{\hat v}, 
%\end{align*}
%into the $P$-equations in \eqref{Peq},
%\begin{align*}
%{F}\big |_{\Z^{b+d}_{{\rm pm}, *}}\left(\vec u\right)=0,\ \Z^{b+d}_{{\rm pm}, *}=\Z^{b+d}\times\{+,-\}\setminus \mathcal{S}, 
%\end{align*}
%and the $Q$-equations in \eqref{Qeq}, 
%\begin{align*}
%{F}\big |_{\mathcal{S}}(\vec u)=0.
%\end{align*}
	We iteratively solve the $Q$-equations and then the $P$-equations.

	\subsection{The $Q$-equations and extraction of parameters}
	The $Q$-equations are of $2b$-dimensional. 
	
	On $\mathcal S$, $\vec u$ is held fixed, and the $Q$-equations \eqref{Qeq}
are viewed instead as equations for $\bm\omega$, and will be solved using the implicit function theorem. The $Q$-equations are used to relate the frequencies $\bm \omega$ to the amplitudes $\bm a$, permitting amplitude-frequency modulation to solve the $P$-equations.

	% % in the symmetric form %coordinates as 
	Note that 
\begin{align}\label{lattq0}
u^{(0)}(t, \bm n)=\sum_{(\bm k, \bm n)\in\mathcal{S}_{+}}\hat u^{(0)}(\bm k, \bm n)e^{i\bm k\cdot\bm \omega^{(0)}t}, %\cos(e_l\cdot\omega^{(0)}t),%e^{i\omega_l^{(0)} t}.
%{\bar z}^{(0)}(t,n)=
\end{align} 
where 
%$
%\mathcal{S}_+=\{(\bm e_l, \bm n_{l})\}_{l=1}^b%\cup\{(-e_l,n^{(l)})\}_{l=1}^b
%$ is given by Theorem \ref{mthm} and 
\begin{align*}
\hat u^{(0)}(\bm e_\ell, \bm n_\ell)=\hat v^{(0)}(-\bm e_\ell,\bm n _\ell)=a_\ell\ {\rm for}\ 1\leq \ell\leq b.
\end{align*}
The solutions $\vec u=(\hat u, \hat v)^t$  (with $(\cdot)^t$ denoting  the transpose) on $\mathcal S$ are held {\it fixed}:  $\vec u \equiv(\hat u^{(0)}, \hat v^{(0)})^t=\vec u^{(0)}$ on $\mathcal S$, the 
$Q$-equations are used instead to solve for the frequencies. Due to symmetry, the $Q$-equations for 
$\hat v$ are the same as that for $\hat u$. So we only need to consider that for the $\hat u$.
Solving  the equations at the initial step
\begin{align*}
{F}(\vec u^{(0)})\big|_{\mathcal{S_+}}=0
\end{align*}
leads to, for $1\leq \ell\leq b$, the initial modulated frequencies
\begin{align}\label{omgdef}
\omega_\ell=\omega_\ell^{(1)}&=\omega_\ell^{(0)}+\delta a_\ell^{2p}
%\sqrt{(\omega_l^{(0)})^2+C_{p+1}^{p/2}2^{-p}a_l^p\delta}\\
%\label{omgdef}&=\omega_l^{(0)}+\frac{C_{p+1}^{p/2}2^{-p}a_l^p\delta}{\sqrt{(\omega_l^{(0)})^2+C_{p+1}^{p/2}2^{-p}a_l^p\delta}+%\omega_l^{(0)}},
\end{align}
satisfying $$\det \left(\frac{\partial \bm \omega^{(1)}}{\partial \bm a}\right)=\delta^b\prod_{\ell=1}^b(2p a_\ell^{2p-1}).$$

Along the way, 
%and denoting the final modulated frequency by $\omega:=\omega^{\infty}$,
we will show that $\bm\omega^{(r)}=\bm\omega^{(0)}+O(\delta)$, for all $r=0, 1, 2, \cdots$. 
% If we denote the range of $\omega=\omega(a)$ to be $\Omega$.  
So $\bm\omega^{(r)}\in \Omega$, the frequency set in Theorem~\ref{ldtthm},
hence Theorem~\ref{ldtthm} is at our disposal to solve the $P$-equations.
To construct approximate solutions to the $P$-equations, it is convenient to view $\bm \omega\in\Omega$
as  an independent parameter, and work in the $(\bm\omega, \bm a)$-variable, 
$(\bm\omega, \bm a)\in\Omega\times [1,2]^b:=(\bm\omega^{(0)}+[\delta, 2^{2p}\delta]^b)\times [1,2]^b$.
The approximate solutions to the nonlinear matrix equation \eqref{nle1} will then be obtained by taking into account 
the solutions to the $Q$-equations as well, which restricts  $(\bm\omega, \bm a)$ to the $b$-dimensional hypersurface $\bm \omega=\bm\omega^{(r)} (\bm a)$, 
at the $r$-th iteration.

	%The key point is definitely the estimate of Green's functions for linearized operators when solving the $P$-equation. It turns out the LDT (cf. Theorem \ref{ldtthm}) plays the central role, which %combined with the projection lemma of Bourgain (based on Yomdin-Gromov triangulation theorem,  cf. Lemma \ref{proj} below) will lead to converting estimates in $\sigma$ into that of $\omega$, %and finally $a.$
	%\subsection{}
	\subsection{The $P$-equations} 
	The  $P$-equations are of infinitely dimensional. 
	
	We use \eqref{Peq} to solve for $\vec u |_{\mathcal S^c}$.   Solving the $P$-equations requires analyzing the invertibility of the linearized operator $H=H_0+\delta T_{\vec u}$ restricted to $\Z^{b+d}_{{\rm pm}, *}$, which we do
in the previous section. Then the resolution of the $P$-equations combines the  Newton scheme and the linear analysis in Section~\ref{ldtsec}. 
	The large deviation estimates in $\sigma$ in Theorem~\ref{ldtthm} will be converted into estimates in the amplitudes $\bm a$ by using 
	a semi-algebraic projection lemma  (cf. Lemma~\ref{proj} below), and amplitude-frequency modulation, $\bm\omega=\bm\omega(\bm a)$.
	This is an established scheme, which has recently been extensively elaborated on with detailed proofs in sects.~IV-VI, \cite{KLW23} and \cite{LW22}.  They  will serve as 
	the basic reference point for the present section.

	%We solve the $P$-equations using a Newton scheme.
	Recall first the formal Newton scheme. Let $\vec u$ be an approximate solution to the $P$-equations \eqref{Peq}.
	%Let $\vec u^{(l)}$ be an approximate solution to the $P$-equations \eqref{Peq} after $\ell$ steps of iteration.
%Assuming $\vec u^{(l)}$ have been constructed  for $1\leq l\leq r$, 
%we construct $\vec u^{(r+1)}$. Toward that purpose, 
For the next approximation ${\vec u}'$, write 
$$ {\vec u}'=\vec u+\Delta_{\rm cor}\vec u'.$$
Then  $\Delta_{\rm cor}\vec u'$ is set to be 
%Then  $\Delta_{r+1} \vec u$ is a solution to the linearized equation at $\vec u^{(r)}$
\begin{align*}
  \Delta_{\rm cor}\vec u'=-(H(\vec u))^{-1}F( \vec u),
\end{align*}
where the linearized operator $H(\vec u)=H_0+\delta T_{\vec u}(\bm \omega, \bm a)$ (cf. \eqref{Ldef} and \eqref{toep}). 
%\begin{align*}
%H={\rm diag}_{(\bm k,\bm n)\in\Z^{b+d}}&\left(\begin{array}{cc}
%	{-\bm k\cdot\bm\omega+\mu_{\bm n}+\varepsilon \Delta} & {0} \\
	%{0}& {\bm k\cdot\bm\omega+\mu_{\bm n}+\varepsilon \Delta} \\
	%\end{array}\right)\\
	%&+\delta \left(\begin{array}{cc}
	%{(p+1)(\hat w*\hat {\frak w})^{*p} *} & {p(\hat w*\hat {\frak w})^{*(p-1)} *\hat w*\hat w*} \\
	%{p(\hat w*\hat {\frak w})^{*(p-1)} *\hat {\frak w}*\hat {\frak w}*}& {(p+1)(\hat w*\hat {\frak w})^{*p} *} \\
	%\end{array} \right), 
%\end{align*}
%as in \eqref{Ldef} and \eqref{toep}, with $\hat w$ in lieu of $\hat u$ and $\hat {\frak w}$ in lieu of $\hat v$.
This is, however, only {\it indicative}  since  we assume  that $H(\vec u)$ is invertible. Due to the small divisor difficulties in the inversion process,
we need to regularize the linearized operator $H(\vec u)$ and do a {\it multi-scale} Newton iteration instead. 

Toward that purpose, let $M$ be a large integer and consider the geometric sequence of scales $M^\ell$, $\ell=1, 2, \cdots$. 
Denote by  $\vec u^{(r)}$, the $r$-th approximate solution to the $P$-equations \eqref{Peq}. 
For the $(r+1)$-th approximation, write 
$$\vec u^{(r+1)}=\vec u^{(r)}+\Delta_{\rm cor} \vec u^{(r+1)}.$$
We define (if it exsits) 
$$ \Delta_{\rm cor} \vec u^{(r+1)}=-(R_{N}{{H(\vec u^{(r)})}}R_{N})^{-1}F( \vec u^{(r)}),$$
where $N=M^{r+1}$ and  $R_{N}$ is the restriction operator to the cube 
$\Lambda_{N}\cap \Z^{b+d}_{\rm pm, *}.$
 %and $\tilde r =r+r_0$, for some $r_0>1$. The reason to consider a shifted sequence $\tilde r$ 
%is technical, see sect.~V, E.~Step~2-F.~Step~4, \cite{KLW23}.
%solve instead the following regularized linearized equation 
%\begin{align*}
% {{F}}( \vec u^{(r)})+R_{\Lambda_{M^{r+1}}^*}{{H}}R_{\Lambda_{M^{r+1}}^*}{{}}\Delta_{r+1} \vec u=0,
%\end{align*}
This leads
%gives 
%\begin{align}\label{udef}
% \Delta_{r+1} \vec u=-{{F}}( \vec u^{(r)})[R_{\Lambda_{M^{r+1}}^*}{{H}}R_{\Lambda_{M^{r+1}}^*}]^{-1}{{}},
%\end{align}
%assuming invertibility. This 
to estimate the inverse (i.e., the Green's function),
\begin{align*}
{{G}}_{N}(\vec u ):= \left(R_{{N}}{{H(\vec u^{(r)})}} R_{N}\right)^{-1}, \  r=0, 1, 2, \cdots.
\end{align*}
The following estimates are crucial  for the convergence of the approximate solutions with exponential decay: 
\begin{align*}
\| {{G}}_{N}(\vec u^{(r)})\|&\leq  e^{(\log N)^C},  \\
| {{G}}_{N}(\vec u^{(r)})((\bm x, \xi);(\bm x',\xi'))|&\leq e^{-c|\bm x-\bm x'|}\ {\rm for} \ |\bm x-\bm x'|\geq (\log N)^C,
\end{align*}
for $\bm x=(\bm k, \bm n)$ and some $c, C>0$, see sect.~V, \cite{KLW23}. 

The LDT Theorem \ref{ldtthm} in the previous section plays an essential role  in the above Green's function estimates.  The operator  $T_{\vec u}$ plays the role of  $S$ in \eqref{hsigm} (cf. Remark \ref{LDTrem2} for details).  The following lemma implies that $T_{\vec u}$ satisfies the decay assumption of $S$ (cf. also  Remark \ref{LDTrem2}).  This will then enable us to use  Theorem \ref{ldtthm}. 
\begin{lem}\label{powerf}
\begin{itemize}
	\item[(1)]     For all $\bm k,\bm k', \bm k''\in\Z^b$, $\bm n,\bm n'\in\Z^d$ and $\xi,\xi'\in\{+, -\}$,  we have the T\"oplitz property in the  $\bm k$-variable: 
	\begin{align*}
	T_{\vec u}((\bm k'+\bm k, \bm n, \xi); (\bm k''+\bm k, \bm n', \xi'))=T_{\vec u}((\bm k', \bm n, \xi ); (\bm k'', \bm n', \xi')). 
	\end{align*}

	\item [(2)]  Assume that $|\vec u(\bm k, \bm n)|\leq e^{-c(|\bm k|+|\bm n|)}, c>0$. Then we have for some  $C=C(b, p)>0,$
	\begin{align*}
	|T_{\vec u}((\bm k, \bm n, \xi); (\bm k', \bm n', \xi'))|\leq  C(1+|\bm k-\bm k'|)^{C} e^{-c|\bm k-\bm k'|-c|\bm n|}\delta_{\bm n,\bm n'}. 
	\end{align*}
	%\[|T_{\vec u}((\bm k'+\bm k, \bm n, \xi); (\bm k''+\bm k, \bm n', \xi'))|\leq C_1(|n-n'|+1)^{C_1} e^{-c|n-n'|-\frac{\kappa}{2}|j-j'|-2pc\max\{|j|,|j'|\}},\ C_1=C_1(p,b)>1.\]
\end{itemize}
\end{lem}
\begin{proof}
The proof is similar to  that of Lemma 5.1 in \cite{LW22} and Proposition 4.1 in \cite{LSZ25} (our case is easier since the convolution here  involves  only  the $\bm k$-variable). We omit the details. 
\end{proof}

\subsection{The induction hypothesis and the proof}

 Let $M$ be a large integer and recall that $\Lambda_{R}\subset\Z^{b+d}$ denotes the $\ell^\infty$-norm induced cube of radius $R$ and centered at the origin.  
 %Set 
% $$r_0=\left[ \frac{\log^{\frac34}\frac{1}{\varepsilon+\delta}}{\log M}\right].$$

In the following, $C>0$ is a large constant and $c>0$ is a small one. 
 
 We begin with the analysis of initial induction steps, which relies on the Neumann series argument and Diophantine estimates \eqref{mk1thm}  of  Theorem \ref{clusthm}. 

%The initial induction steps  were not fully addressed  in  \cite{KLW23, LW22} since  they  used   the bound   $(\varepsilon+\delta)^{-c},  c\in (0,1)$ in the Neumann series argument to handle Green's function estimates of size $N\leq M^{r_0}$ ,  cf. e.g., \cite{LW22},  % (in \cite{KLW23}, they choose )
%where 
%$r_0=\left[ \frac{\log^{\frac34}\frac{1}{\varepsilon+\delta}}{\log M}\right].$  The estimates  become  invalid if $1\leq N\leq \log\frac{1}{\varepsilon+\delta}$.  
%In \cite{KLW23}, the initial scales interval is too short to perform  large scales induction. 
%For these reasons, we  first  fix    this minor  issue  using  the power-law (in scales)  bound   \eqref{mk1thm}. 

\subsubsection{The initial $2r_\star-r_0$  steps}\label{subsin}
	We choose $1<r_0<r_\star$ such that 
	\begin{align*}
M^{r_0}\leq \log \frac{1}{\varepsilon+\delta}\leq M^{r_0+1} 
	\end{align*}
	and 
	\begin{align}\label{r0}
	M^{2r_\star}\leq  (\varepsilon+\delta)^{-c_1}<M^{2r_\star+1},
	\end{align}
	where $c_1>0$ is given in Theorem \ref{ldtthm}. 
	%cf., sect.~V, E.~Step~2.-F.~Step~4, \cite{KLW23}.

	%For any $r\geq 1$, we define 
	%$$\tilde r=r+r_0.$$

	Let 
	%$M$ be a large integer, %Recall that  $F=\mathcal{F}_\omega\big|_{\mathcal{S}^c}.$
%From \eqref{mk1thm} of Theorem \ref{clusthm}, we can set $L=(\varepsilon+\delta)^{-\frac{1}{10^3b^2d^2R_{b+2}^2}}$ and $\eta=(\varepsilon+\delta)^{\frac{1}{8b}}$. 
	 %and  
	 $\vec u^{(0)}(\bm \omega,\bm a)=\vec u^{(0)}(\bm \omega^{(0)},\bm a)$ with $\bm \omega=\bm\omega^{(1)}$ given by \eqref{omgdef}.
	 %=\vec u^{(0)}(\bm a)$  with $(\bm\omega, \bm a)\in\Omega\times [1,2]^b=(\bm\omega^{(0)}+[-C\delta, C\delta]^b)\times [1,2]^b$.
	 %and $\Omega$ being given in the previous subsection.  % \eqref{omgset}.  

	We start by constructing $\vec u_{\rm in}^{(1)}(\bm\omega, \bm a)$. Obviously, 
	$${F}(\vec u^{(0)})=O(\varepsilon+\delta)$$%O(\sqrt{\varepsilon+\delta})$$
	%for $|k|\leq (\varepsilon+\delta)^{-\frac{1}{10^4db^4}},$
	 and (the support)   ${\rm supp}\  {F}(\vec u^{(0)})\subset \Lambda_{M_0}$ for some $M_0=M_0(p, \mathcal{S})>0.$
	  Let $L_1=M^{r_0+1}>M_0.$  It suffices to estimate 
	\begin{align*}
	H_{L_1}^{-1}=(R_{\Lambda_{L_1}}({D}(0)+\varepsilon(\Delta\oplus\Delta)+\delta{T}_{\vec u^{(0)}})R_{\Lambda_{L_1}})^{-1}.
	\end{align*}
	It follows from  \eqref{mk1thm} of Theorem \ref{clusthm}, $\bm\omega=\bm\omega^{(0)}+O(\delta)$ and the Neumann series argument (cf. Lemma \ref{lwlem})  that 
	\begin{align*}
	\|H_{L_1}^{-1}\|&\leq 2L_1^{C_1}, \\
	|H_{L_1}^{-1}((\bm x, \xi); (\bm x', \xi'))|&\leq L_1^{C_1}e^{-\frac{\log \frac{1}{\varepsilon+\delta}}{2}|\bm x-\bm x'|}\  {\rm for}\ \bm x\neq \bm x'\in \Lambda_{L_1},
	\end{align*}
	since $L_1=M^{r_0+1}\ll (\varepsilon+\delta)^{-c_1}$. Then we have 
	\begin{align*}
	\Delta_{\rm cor}\vec u_{\rm in}^{(1)}=-H_{L_1}^{-1}{F}(\vec u^{(0)}),
	\end{align*}
	and consequently
	\begin{align*}
	\vec u_{\rm in}^{(1)}=\vec u^{(0)}+\Delta_{\rm cor}\vec u_{\rm in}^{(1)}=\vec u^{(0)}-H_{L_1}^{-1}{F}(\vec u^{(0)}).
	\end{align*}
	Substituting $\vec u_{\rm in}^{(1)}$ into the $Q$-equations using 
        $$ \hat u_{\rm in}^{(1)}\big |_{\mathcal{S}+}=\hat v_{\rm in}^{(1)}\big |_{\mathcal{S}-}=\bm a,$$  we obtain
	\begin{align*}
	\omega_\ell&=\omega_\ell^{(0)}+\varepsilon\left(\frac{1}{a_\ell}(\Delta \hat u_{\rm in}^{(1)})(\bm e_\ell, \bm n_\ell)\right)(\bm\omega, \bm a)\\
	&\ \ +\delta\left(\frac{1}{a_\ell}(\hat u_{\rm in}^{(1)}*\hat v_{\rm in}^{(1)})^{*p} *\hat u_{\rm in}^{(1)}(\bm e_\ell, \bm n_\ell)\right)(\bm\omega, \bm a), \, \ell= 1, 2, \cdots,  b.
	\end{align*}
	%where 
	%$\omega=(\omega_l)_{l=1}^b$  is defined  by \eqref{omgdef} and $\varphi_l^{(1)}(\omega,a)$ is  smooth in  $\omega,a$.
	Since the second and the third terms are smooth in $\bm\omega$ and $\bm a$,  using the implicit function theorem yields $\bm \omega_{\rm in}^{(2)}=\bm \omega_{\rm in}^{(2)}(\bm a)=(\omega_{{\rm in}, \ell}^{(2)}(\bm a))_{\ell=1}^b$, written in the form:
	%\begin{align*}
	%\omega_l^{(2)}(a)=\omega_l^{(0)}+\frac{(\varepsilon+\delta) \varphi_l^{(1)}(\omega_l^{(1)}(a), a)}{\sqrt{(\omega_l^{(0)})^2+(\varepsilon+\delta) \varphi_l^{(1)}(\omega_l^{(1)}(a), a)}+\omega_l^{(0)}} \ (1\leq l\leq b).
	%\end{align*}}
	\begin{align*}
	\omega_{{\rm in}, \ell}^{(2)}(\bm a)=\omega_\ell^{(0)}+(\varepsilon+\delta) \varphi_{{\rm in}, \ell}^{(2)}(\bm a), \,  \ell= 1, 2, \cdots, b,
	%{\sqrt{(\omega_l^{(0)})^2+(\varepsilon+\delta) \varphi_l^{(1)}(\omega_l^{(1)}(a), a)}+\omega_l^{(0)}} \ (1\leq l\leq b).
	\end{align*}
	with a smooth function $\varphi_{{\rm in}, \ell}^{(2)}$  (comparing  with \eqref{omgdef}). 
	%Define
	%$$\Delta_r\vec u=-H_{M^{r}}^{-1}{F}(\vec u^{(r-1)}),$$
	%and
	%$$\vec u^{(r)}=\vec u^{(r-1)}+\Delta_r\vec u.$$
	
	The above constructions can be inductively performed for $2r_\star-r_0$ steps,  with $r_\star$ satisfying \eqref{r0}.
	%\begin{align*}
	%M^{r_0}\leq  (\varepsilon+\delta)^{-\frac{1}{10^3b^2d^2\tau_{b+2}^2}}<M^{r_0+1}.
	%\end{align*}
	So we have obtained  $\vec u_{\rm in}^{(r)}=\vec u_{\rm in}^{(r)}(\bm\omega, \bm a)$ ($1\leq r\leq 2r_\star-r_0$)  for all $(\bm \omega, \bm a)\in \Omega\times [1,2]^b$. 
	%\begin{rem}
%The presentation here follows the original scheme 
%present nonlinear iteration scheme follows that 
%of Bourgain, Chap.~18,  \cite{Bou05}. For an updated version with complete proofs and more, see sects. V and VI, \cite{KLW23}.
% and thus may be a bit schematic, or an overview:  
%and aims to provide the essential points in the analysis. 
%While the detailed arguments are commonly assumed to be straightforward, 
%The recent beautiful work \cite{KLW23}, however, has indeed established a complete proof,  and more (cf. sect.  V).  So we refer the reader to \cite{KLW23} for further submersion.
	%\end{rem}
	
	\subsubsection{The Inductive Theorem}
Recall that 
	$$\Omega=\bm\omega^{(0)}+[\delta, 2^{2p}\delta]^b\subset\R^b.$$
	We first state the induction hypothesis, which can be verified if $0<\varepsilon\leq \delta\leq \log^{-1}\frac{1}{\varepsilon}\leq \delta_0\ll1$ and $(\bm \alpha, \bm \theta)\in\mathcal W$ (cf. \eqref{ALset}). 
 Given its significance, we refer to it as the induction theorem. 
	
	%We now introduce the Inductive Theorem. 
	  \begin{thm}
	 For $r\geq r_\star,$ we have
	%  {\bf Inductive Assumptions ($r\geq r_0$):}
	\begin{itemize}
	\item[{\bf(Hi)}] ${\rm supp}\ \vec u^{(r)}\subset \Lambda _{M^{\tilde r}}$, $\tilde r=r+r_\star.$
	%\begin{rem}
	%We have $q^{(r)}=(q_1^{(r)},q_2^{(r)})\in\C^2.$
	%\end{rem}
	\item[{\bf (Hii)}] $\|\Delta_{\rm cor} \vec u^{(r)}\|<\delta_r,\ \|\partial \Delta_{\rm cor} \vec u^{(r)}\|<{\bar\delta}_{r}$, where $\partial$ refers to derivation in $\bm\omega$ or $\bm a$,   $$\vec u^{(r)}=\vec u^{(r-1)}+\Delta_{\rm cor} \vec u^{(r)},$$
	 and $\|\cdot\|=\sup_{\bm\omega, \bm a}\|\cdot\|_{\ell^2(\Z_{\rm pm}^{b+d})}.$
	\begin{rem}
	The size of $\delta_r, \bar\delta_r$ will satisfy $\log\log \frac{1}{\delta_r+\bar\delta_r}\sim r.$ %Precisely, we have 
	%\begin{align*}
	%\delta_r<\sqrt{\varepsilon+\delta}e^{-(\frac{4}{3})^r},\ \bar\delta_r<\sqrt{\varepsilon+\delta}e^{-\frac12(\frac{4}{3})^r}.
	%\end{align*}
	\end{rem}
	\item[{\bf (Hiii)}] $|\vec u^{(r)}(\bm k, \bm n)|\leq e^{-c(|\bm k|+|\bm n|)}$ for some  $c>0.$
	\begin{rem}
	The constant $c>0$ will decrease slightly along the iterations but remain bounded away from $0$. This will become clear in the proof, cf. sect.~V, G. Step 4, \cite{KLW23}.
	We also remark that $\vec u ^{(r)}$ can be defined as a $C^1$ function on the entire parameter space $(\bm\omega, \bm a)\in\Omega\times [1,2]^b$ by using an extension argument, cf., \cite{Bou98, BW08} and sect.~V, L. Step 9, \cite{KLW23}.  So one can apply the implicit function theorem to solve the $Q$-equations 
	%by letting 
	%\begin{align*}
	%q^{(r)}(\pm e_l,n^{(l)})=a_l/2\ {\rm for}\ 1\leq l\leq b.
	%\end{align*}
	%This leads to 
	leading to 
		\begin{align}\label{algeq}
	\omega_\ell^{(r)}(\bm a)=\omega_\ell^{(0)}+(\varepsilon+\delta)\varphi_\ell^{(r)}(\bm a) \ (1\leq \ell\leq b),
	%\omega_l+(\varepsilon+\delta) \varphi_l^{(r)}(\omega, a) \ (1\leq l\leq b),
	\end{align}
	%and then 
	%\begin{align*}
	%\omega_l^{(r)}=\omega_l+\frac{(\varepsilon+\delta) \varphi_l^{(r)}(\omega, a)}{\sqrt{\omega_l^2+(\varepsilon+\delta) \varphi_l^{(r)}(\omega, a)}+\omega_l} \ (1\leq l\leq b),
	%\omega_l+(\varepsilon+\delta) \varphi_l^{(r)}(\omega, a) \ (1\leq l\leq b),
	%\end{align*}
	where   $\|\partial \bm\varphi^{(r)}\|=\sup_{1\leq \ell\leq b}\|\partial \varphi_\ell^{(r)}\|\lesssim 1$ and $ \bm\varphi^{(r)}=( \varphi^{(r)}_\ell)_{\ell=1}^b.$ By {\rm ({\bf Hii})}, we have 
	\begin{align*}
	|\bm\varphi^{(r)}-\bm\varphi^{(r-1)}|\lesssim (\varepsilon+\delta)\|\vec u^{(r)}-\vec u^{(r-1)}\|\lesssim {(\varepsilon+\delta)}\delta_r.
	\end{align*}
	  %and $\det (\frac{\partial \omega}{\partial a})\geq c\delta^{b}$ for some $c>0.$ 
	Denote by $\Gamma_r$ the graph of $\bm\omega^{(r)}=\bm\omega^{(r)}(\bm a)$. We have $\|\Gamma_r-\Gamma_{r-1}\|\lesssim{(\varepsilon+\delta)}\delta_r.$ Recall that $\bm\omega$ is given by \eqref{omgdef} and $\bm\varphi^{(0)}=\bm 0$. Thus we  have a diffeomorphism from $\bm\omega^{(r)}\in\Omega$ to $\bm a\in[1,2]^b$. 
	%So we also write $\bm a^{(r)}=\bm a^{(r)}(\bm \omega)$ for $\bm \omega\in\Omega.$
	\end{rem}
	
	\item[{\bf (Hiv)}] There is  a collection $\mathcal{I}_r$ of intervals $I\subset \Omega\times [1,2]^b$ of size $M^{-{\tilde r}^{10C}},$  so that 
	\begin{itemize}
	\item[{\bf(a)}] On each $I\in\mathcal{I}_r,$  both  $\hat u^{(r)}(\bm \omega, \bm a)$ and $ \hat v^{(r)}(\bm \omega, \bm a)$ are given by  rational functions  in $(\bm \omega, \bm a)$
of degree at most $M^{{\tilde r}^3}$.
	\item[{\bf (b)}] For $(\bm\omega, \bm a)\in\bigcup_{I\in\mathcal{I}_r}I$, 
	\begin{align*}
	\|F(\vec u^{(r)})\|\leq \kappa_r,\ \|\partial F(\vec u^{(r)})\|\leq \bar\kappa_r,
	\end{align*}
	where $\partial$ refers to derivation in $\bm\omega$ or $\bm a$, and $\log\log\frac{1}{\kappa_r+\bar\kappa_r}\sim r$. %More precisely, we have 
	%\begin{align*}
	%\kappa_r<\sqrt{\varepsilon+\delta}e^{-(\frac{4}{3})^{r+2}},\ \bar\kappa_r<\sqrt{\varepsilon+\delta}e^{-\frac12(\frac{4}{3})^{r+2}}.
	%\end{align*}
	\item[{\bf (c)}] For $(\bm \omega, \bm a)\in\bigcup_{I\in\mathcal{I}_r}I$ and  $H=H(\vec u^{(r-1)})$, one has
	\begin{align}
	\label{hivc1}\|H_{M^{\tilde r}}^{-1}\|&\leq M^{{\tilde r}^C},\\
	\label{hivc2}|H_{M^{\tilde r}}^{-1}((\bm x,\xi); (\bm x',\xi'))|&\leq e^{-c|\bm x-\bm x'|}\  {\rm for}\ |\bm x-\bm x'|>{\tilde r}^C,	\end{align}
	where $\bm x=(\bm k, \bm n), \bm x'=(\bm k', \bm n')$ and $H_{M^{\tilde r}}$ refers to the restriction of $H$ to $\Lambda_{M^{\tilde r}}$.
	\item[{\bf (d)}] Each $I\in\mathcal{I}_r$ is contained in some $I'\in\mathcal{I}_{r-1}$ and 
	\begin{align*}
	{\rm meas}\left(\Pi_{\bm a}(\Gamma_{r-1}\cap(\bigcup_{I'\in\mathcal{I}_{r-1}}I'\setminus\bigcup_{I\in\mathcal{I}_r}I))\right)\leq M^{-\frac{\tilde r}{C(b)}},
	\end{align*}
	where $C(b)>0$ depends only on $b$, and  $\Pi_{\bm a}$ denotes the projection of the set on the $\bm a$-variable. 
	\end{itemize}
	\item[(\bf Hv)]
		The following precise  relations  hold  true
		\begin{align*}\label{conv}
		\delta_r= {(\varepsilon+\delta)}^{\frac{1}{2}}M^{-(\frac{4}{3})^r}, \, \bar\delta_r=  {(\varepsilon+\delta)}^{\frac{1}{8}} M^{-\frac{1}{2}(\frac{4}{3})^r};\\
		\kappa_r= {(\varepsilon+\delta)}^{\frac{3}{4}} M^{-(\frac{4}{3})^{r+2}}, \, \bar\kappa_r=  {(\varepsilon+\delta)}^{\frac{3}{8}} M^{-\frac{1}{2}(\frac{4}{3})^{r+2}}. \end{align*}
	\end{itemize}
	\end{thm}
	
	%\begin{rem}
%The presentation here follows the original scheme 
%present nonlinear iteration scheme follows that 
%of Bourgain, Chap.~18,  \cite{Bou05}. For an updated version with complete proofs and more, see sects. V and VI, \cite{KLW23}.
% and thus may be a bit schematic, or an overview:  
%and aims to provide the essential points in the analysis. 
%While the detailed arguments are commonly assumed to be straightforward, 
%The recent beautiful work \cite{KLW23}, however, has indeed established a complete proof,  and more (cf. sect.  V).  So we refer the reader to \cite{KLW23} for further submersion.
	%\end{rem}

	\subsection{Proof of the Inductive Theorem}

	From the analysis in Section \ref{subsin}, we have constructed  $\vec u_{\rm in }^{(r)}$ for $1\leq r\leq 2r_\star-r_0$. In particular, we obtain   that the Inductive Theorem  holds  for $r=r_\star$ as we can set $\vec u^{(r_\star)}=\vec u_{\rm in}^{(2r_\star-r_0)}$.  
	
	%We  also remark that  since  our initial induction scale in Section \ref{subsin}  starts from $M^{r_0}\sim \log \frac{1}{\varepsilon+\delta},$  we can deal with  $N=M^{r}$ scales estimates via multi-scale analysis  provided  $ (\log N)^{\frac{4}{\rho^4}}\geq  \log \frac{1}{\varepsilon+\delta}$.  This fills  the gap  of  scales near $M^{2r_\star}$, e.g.,  scales $M^r$ for $2r_\star<r<100r_\star.$

	Assume now  the Inductive Theorem holds for $r>r_\star$.  We will prove it for  $\ell=r+1.$
	
	We first check the degree bound on $\vec u^{(r+1)}$ (i.e., ({\bf Hiv,~a}) for $r+1$). In fact, the Newton scheme gives 
	\begin{equation}\label{Newton multiscale}
	\Delta_{\rm cor} \vec u^{(r+1)}=-H^{-1}_{N}(\vec u^{(r)}) F(\vec u^{(r)}),\ 
	 N=M^{\tilde r+1},\ \tilde r=r+r_\star, 
	\end{equation}
    and by ({\bf Hiv,~a}), $\vec u^{(r)}(\bm\omega, \bm a)$ is  a rational function of degree at most $M^{\tilde r^3}$ (in  $(\bm\omega, \bm a)$). By the convolution structure of $T_{\vec u^{(r)}}$, we have 
	\[{\rm deg}\   F(\vec u^{(r)})\leq (2p+1)M^{\tilde r^3},\ {\rm deg}\   H_{N}(\vec u^{(r)})\leq 2p M^{\tilde r^3}.\]
	Then by  Cramer's rule, we  have 
	\[H_{N}^{-1}(\vec u^{(r)})=\frac{A}{\det   H_{N}(\vec u^{(r)})},\]
	where $A$ is the adjacent matrix of $H_{N}(\vec u^{(r)})$. Hence,
	\[{\rm deg}\  H^{-1}_{N}(\vec u^{(r)})\leq (4p+2)M^{\tilde r+1}M^{\tilde r^3}.\]
	Thus,   using \eqref{Newton multiscale} shows 
	\[{\rm deg}\   \Delta_{\rm cor} \vec u^{(r+1)} \leq (4p+2)M^{\tilde r+1}M^{\tilde r^3}+(2p+1)M^{\tilde r^3} \leq M^{(\tilde r+1)^3}.\]
	This proves the  ({\bf Hiv,~a})  for $r+1$. 
	
	   The other inductive  assumptions can be verified in the usual way, cf.  \cite{LW22} and sect.~V, \cite{KLW23},  e.g., the verification  of   ({\bf Hiv,~b}) can be found in I. Step 6 of  \cite{KLW23}; the verification of the constant relations {(\bf Hv)} can be found in Appendix E of  \cite{LW22}.

	    It remains to establish ({\bf Hiv,~c}) and ({\bf Hiv,~d}) for $r+1$.  For this purpose, we set  (recalling  $\tilde r=r+r_\star$)   $$N=M^{\tilde r+1}$$  and a smaller scale,
	  $$N_1=(\log N)^{\frac{4}{\rho^4}},$$% \ M^{r_1}\leq N_1<M^{r_1+1},$$
	where $0<\rho\ll1$ is given in Theorem \ref{ldtthm}.  
	
	To establish ({\bf Hiv, c}), we first make approximations on $T_{\vec u^{(\ell)}}$  for different $\vec u^{(\ell)}$.    %Below, we use  $H_{N}(\cdot)$ to denote $H_{\Lambda_N}(\cdot)$ for simplicity.  
	%It follows from 
	%\begin{align}
	%\|T_{\vec u^{(r)}}-T_{\vec u^{(r-1)}}\|\lesssim \|\Delta_{\rm cor} \vec u^{(r)}\|\lesssim \delta_r\lesssim\sqrt{\varepsilon+\delta}e^{-(\frac43)^{r}}\ll e^{-M^{2\tilde r}},
	%\end{align} 
	%Lemma \ref{lwlem} and assumptions on  $H_{M^{\tilde r}}({\vec u^{(r-1)}})$ that  $H_{{M^{\tilde r}}}({\vec u^{(r)}})$ also satisfies essentially  the estimates \eqref{hivc1} and \eqref{hivc2} .  
	For the set $$ W=\left(([-N,N]^{b+d}\setminus [-N/2, N/2]^{b+d})\times \{+, -\}\right)\cap \Z_{{\rm pm},*}^{b+d},$$ we use $N_1$ size regions to do estimates, in order to lower the degree of the
	associated semi-algebraic set, which will be essential for the upcoming semi-algebraic projection arguments. We set
	$$r_1=2\left[\frac{\log N_1}{\log \frac43}\right]+1.$$
	Then 
	\begin{align}\label{tq1}
	\|T_{\vec u^{(r)}}-T_{\vec u^{(r_1)}}\|\lesssim\delta_{r_1}\lesssim\sqrt{\varepsilon+\delta}e^{-(\frac43)^{r_1}}\ll e^{-N_1^2}.\end{align}
	So we can estimate  $H_{W}^{-1}(\vec u^{(r)})$ using  $H_Q^{-1}(\vec u^{(r_1)})$ for $Q\subset W$.  More precisely, we  want  to  show that the following estimates hold for every $Q\subset W$ with $Q\in\mathcal{ER}(N_1)$:
	\begin{align}
\label{tqr11}\|H_{Q}^{-1}(\vec u^{(r_1)})\|&\leq e^{N_1^{\frac{3}{4}}},\\
\label{tqr12}|H_{Q}^{-1}(\vec u^{(r_1)})((\bm x,\xi); (\bm x',\xi'))|&\leq e^{-c|\bm x-\bm x'|}\  {\rm for}\ |\bm x-\bm x'|>N_1^{\frac{8}{9}}.	
\end{align}
This requires additional  restrictions on $(\bm\omega, \bm a).$ For this, we divide into the following cases:
\begin{itemize}
	\item[{\bf Case 1.}]  Assume $Q\in (\bm k,\bm n)+\mathcal{ER}_0(N_1)$ with $|\bm n|>10N_1$ and $Q\subset W.$  Denote by $\mathcal{C}_1$ the set of all these $Q$.   To estimate $H_Q^{-1}(\vec u^{(r_1)})$, we can impose as in the proof of Lemma \ref{lsclem2} (with $\sigma=0$) the following condition
	\begin{align*}
	\min_{\bm k\in\Pi_b Q, \xi=\pm1, 1\leq \ell\leq N_1^{b+2d}}|\bm k\cdot\bm \omega+\xi {\lambda_\ell}|>e^{-\frac14N_1^{\frac{3\kappa_1}{4}}},
	\end{align*}
where $\lambda_\ell$ is  given by  Lemma \ref{lsclem2}. For $\bm k\neq \bm 0$, we  can remove $\bm\omega$ directly (i.e., a set in $\bm \omega$ of measure at most  $N^{C}e^{-\frac14N_1^{\frac{3\kappa_1}{4}}}$) to control $A_{\bm k}^{-1}(0)$ (cf. the proof of Lemma \ref{lsclem2} for this notation).    %To remove $\omega$, the above condition makes sense 
%only if $\bm k\neq 0$. 
For the case $\bm k=\bm 0$ of which we cannot remove $\bm\omega$,  we  employ  Green's function estimates in  (2) of Lemma \ref{Boulem}. Indeed, similar to the proof of Lemma \ref{lsclem2}   (1),   the  estimate of   $A_{\bm 0}^{-1}(0)=\mathcal L^{-1}_{Q(\bm 0)}(\bm\theta)\oplus\mathcal L^{-1}_{Q(\bm 0)}(\bm\theta)$  follows directly from  (2) of Lemma \ref{Boulem} since we assume $(\bm \alpha, \bm\theta)\in\mathcal W$ (cf. \eqref{ALset}).  
%\begin{align*}
%R_{\{n:\ (0,n)\in Q\}} \left(
%	\begin{array}{cc}
 %V^2+\frac\varepsilon 2\Delta&  \frac\varepsilon2 \Delta\\
 % \frac\varepsilon2 \Delta&  V^2+\frac\varepsilon2\Delta
%\end{array}\right)R_{\{n:\ (0,n)\in Q\}},
%\end{align*}
%is positive definite since $\inf_nV^2(n)\geq 1$ and $0<\varepsilon\ll1$.  This combined with  arguments  in the proof of {\bf Claim 2} ( cf. the proof of Lemma \ref{lsclem}) yields  $H^{-1}_{k=0, (0,n)\in Q}(q^{(r_1)})$ has desired estimates. So  it is not necessary to remove $\omega$ for the case $k=0.$ 
So by taking account of all $\bm k\in \Pi_bQ$,  all  $Q\in\mathcal{C}_1$ and by Fubini's theorem,  we can find $I_1\subset \Omega\times[1,2]^b$ with (recalling $\rho\leq c_2\kappa_1\leq \frac{\kappa_1}{10^4(b+d)^2}$)
$${\rm meas}(I_1)\leq N^Ce^{-\frac14 (\log N)^{\frac{3\kappa_1}{\rho^4}}}\leq  e^{-(\log N)^{10}}\ll (\varepsilon +\delta)^{b+1}M^{-{\tilde r}},$$
 so that,  for all $(\bm\omega, \bm a)\notin I_1$ and all  $Q\in \mathcal{C}_1$,  $H_Q^{-1} (\vec u^{(r_1)})$    satisfies  \eqref{tqr11} and \eqref{tqr12}. \\
\item[ {\bf Case 2.}]  Assume $Q\in (\bm k, \bm n)+\mathcal{ER}_{\bm 0}(N_1)$ with $|\bm n|\leq 10N_1$ and $Q\subset W.$  Denote by $\mathcal{C}_2$ the set of all these $Q$.  In this case,  it must be that $|\bm k|\geq N/2.$ We will use the projection  lemma    and LDT to remove $\bm\omega.$ 
 For any $Q\in\mathcal{C}_2$, we can write $Q=(\bm k, \bm n)+Q_0 $ with $Q_0\in (\bm 0, \bm n)+\mathcal{ER}_{\bm 0}(N_1)$ for $|\bm n|\leq 10N_1,$ and $N/2\leq |\bm k|\leq N.$  This motivates us to  consider 
 $$H(\sigma)=H(\vec u^{(r_1)};\sigma)=D(\sigma)+\varepsilon (\Delta\oplus\Delta)+\delta T_{\vec u^{(r_1)}},$$
which has been investigated in Section \ref{ldtsec}. 
 Recall that $G_Q(\sigma)$ denotes the Green's function of $H(\sigma)$ restricted to $Q$. Then the  Toeplitz property in the $\bm k$-direction of $H(\sigma)$ implies 
 $$H^{-1}_Q(\vec u^{(r_1)})=G_Q(0)=G_{Q_0}(\bm k\cdot\bm\omega).$$

 Note that $\Omega_{N_1}$ given by Theorem \ref{ldtthm} is a semi-algebraic set, {\it independent of} $\bm a$, and of degree at most $N_1^{10(b+d)}$.  Without loss of generality, we may  assume $\Pi_{\bm\omega}(I\cap\Gamma_{r_1})\subset \Omega_{N_1} $ for each $I\in\mathcal{I}_{r_1}$.  For otherwise, we can  replace $\Pi_{\bm\omega}(I\cap\Gamma_{r_1})$ with $(\Pi_{\bm\omega}(I\cap\Gamma_{r_1}))\cap (\Omega_{N_1'}\setminus\Omega_{N_1})$, where  $N_1'=(\log M^{\tilde r})^{\frac{4}{\rho^4}}\ll N_1$ since by Remark \ref{LDTrem1}, 
 $${\rm meas}(\Omega_{N_1'}\setminus\Omega_{N_1})\leq e^{-\frac12 N_1^{\rho^4}}\leq e^{-\frac{1}{2}(\log N)^4}\ll (\varepsilon+\delta)^{b+1}N^{-10}.$$	
	Then we apply  the LDT at  scale $N_1$ on each $I\in\mathcal{I}_{r_1}$ (cf. Remark \ref{LDTrem2}) to obtain that  for all $Q\in (\bm 0, \bm n)+\mathcal{ER}_{\bm 0}(N_1)$ with $|\bm n|\leq 10N_1,$  and for $\sigma$ outside a set  of  measure at most $e^{-N_1^{\rho}}$, the following estimates hold true: 
		\begin{align}
\label{GQ1}\|G_{Q}(\sigma)\|&\leq e^{N_1^{\frac{3}{4}}},\\
\label{GQ2}|G_{Q}(\sigma)((\bm x,\xi); (\bm x',\xi'))|&\leq e^{-c|\bm x-\bm x'|}\  {\rm for}\ |\bm x-\bm x'|>N_1^{\frac{8}{9}}.	
\end{align}
For the usage of projection lemma  in the present setting (i.e., ${\rm meas}(\Omega)\sim \delta^b$), we need to make more precise descriptions of  the admitted  $\sigma$ first.  We say $\sigma $ is $Q$-bad if either \eqref{GQ1} or \eqref{GQ2} fails.  %for some $Q_0\in (0,n)+\mathcal{ER}_0(N_1)$ with $|n|\leq 2N_1.$
\begin{lem}\label{sgmlem}
Fix  $Q\in (\bm 0, \bm n)+\mathcal{ER}_{\bm 0}(N_1)$ with $|\bm n|\leq 10N_1.$  Denote by $\Sigma_{Q}$ the set of  $Q$-bad $\sigma\in\R$. Then 
\begin{align*}
\Sigma_Q\subset \bigcup_{1\leq \ell\leq N_1^C}J_\ell,
\end{align*}
where each $J_\ell$ is an interval of length $\sim (\varepsilon +\delta).$
\end{lem}
\begin{proof}
Note that $N=M^{\tilde r+1}\geq M^{2r_\star+1}>(\varepsilon +\delta)^{-c_1}$ since $r>r_\star$. So we have  
\begin{align*}
N_1= (\log N)^{\frac{4}{\rho^4}}>c_1\log^{\frac{4}{\rho^4}}\frac{1}{\varepsilon+\delta}>\log^{\frac {2}{\rho^4}}\frac{1}{\varepsilon+\delta},
\end{align*}
  which implies 
\begin{align}\label{edn1}
(\varepsilon+\delta)^{-1}<e^{N_1^{\rho^4/2}}.
\end{align}
%First, suppose that  for some  $C>1,$
%\begin{align}
%C(\varepsilon +\delta)^{-1}\leq e^{\frac12 N_1^{\frac{3}{4}}}.
%\end{align}
So we can define for each $\bm x=(\bm k, \bm n)\in Q$ and $\xi =\pm 1$ the set 
\begin{align*}
J_{\bm x,\xi}=\{\sigma\in\R:\ |\sigma+\bm k\cdot\bm \omega+\xi \mu_{\bm n}|\leq C (\varepsilon+\delta)\},
\end{align*}
where $C>1$ depends  only on $\|\Delta\|, \|T_{\vec u^{(r_1)}}\|.$
From \eqref{edn1} and the Neumann series argument (cf. Lemma \ref{lwlem}), we  have 
\begin{align*}
\Sigma_Q\subset \bigcup_{\bm x\in Q,\ \xi=\pm1} J_{\bm x,\xi}.
\end{align*}
This proves Lemma \ref{sgmlem}. % assuming \eqref{}.  

\end{proof}

Now fix $I_0\in\mathcal{I}_{r_1}.$    Solving    the $Q$-equations  at $r=r_1$ leads to the graph $\Gamma_{r_1}$. 
Then $\tilde I=\Pi_{\bm \omega}(\Gamma_{r_1}\cap I_0)$ is an interval of size at most $ \varepsilon+\delta.$
For $\bm\omega\in \tilde I$, let $\tilde \Sigma$ be the set of $\sigma\in\R$ so that  either \eqref{GQ1} or \eqref{GQ2} fails  for some $Q\in (\bm 0, \bm n)+\mathcal{ER}_{\bm 0}(N_1)$ with $|\bm n|\leq 10N_1$.  Then by Lemma \ref{sgmlem}, $\tilde \Sigma$ can be covered by $N_1^C$ intervals of size $\sim (\varepsilon+\delta).$
Pick $J$ to be one of such intervals  and consider $\mathcal{K}:=\tilde I\times (\tilde \Sigma\cap J)\subset \tilde I\times J\subset \R^b\times \R.$
We will show that $\mathcal{K}$ is a semi-algebraic set of degree $\deg \mathcal{K}\leq N_1^CM^{C{\tilde r}_1^3}$ and measure ${\rm meas}(\mathcal{K})\leq C(\delta+\varepsilon)^b e^{-N_1^{\rho}}$. This measure bound follows directly from the LDT at scale $N_1$ and the Fubini's theorem. For the semi-algebraic description of $\mathcal{K}$, we need the following lemma. 
\begin{lem}[Tarski-Seidenberg Principle, cf. e.g., \cite{Bou07}]\label{tsp}
Denote by $(\bm x, \bm y)\in\mathbb{R}^{d_1+d_2}$ the product variable. If $X\subset\mathbb{R}^{d_1+d_2}$ is semi-algebraic of degree $B$, then its projections $\Pi_{\bm x}X\subset\mathbb{R}^{d_1}$ and
 $\Pi_{\bm y}X\subset\mathbb{R}^{d_2}$ are semi-algebraic of degree at most $B^{C}$, where $C=C(d_1,d_2)>0$.
\end{lem}
We define $X\subset I_0\times J$ to be the set of all $(\bm\omega, \bm a,\sigma)$ so that either \eqref{GQ1} or \eqref{GQ2} fails  for some $Q\in (\bm 0, \bm n)+\mathcal{ER}_{\bm 0}(N_1)$ with $|\bm n|\leq 10N_1$. Similar to  the proof of Lemma \ref{lsclem1}, we can regard $X$ as a semi-algebraic set of degree $N_1^C$.   From ({\bf Hiv})  and solving the $Q$-equations  for $r=r_1$, we know that $\Gamma_{r_1}\cap I_0$ is given by an  algebraic equation  in $\bm \omega, \bm a$ (cf. \eqref{algeq}) of degree $M^{C{\tilde r}_1^3}$. 
Note also that 
\begin{align*}
\mathcal{K}=\Pi_{\bm\omega, \sigma}\left(X\cap ((\Gamma_{r_1}\cap I_0)\times \R)\right), 
\end{align*}
which together with Lemma \ref{tsp} implies  $\deg \mathcal{K}\leq N_1^CM^{C{\tilde r}_1^3}.$
%We 
%We first divide $\Pi_a I_0$ into disjoint subintervals of size $\sim\varepsilon+ \delta$, denoted by $\{ I_{a,l}\}_{1\leq l\leq \delta^{-b}}$.  Fixing  $(\omega,a)\in I_0$, we let $\mathcal{K}_{\omega,a}\subset\R$ be the set of $\sigma$ so that either \eqref{GQ1} or \eqref{GQ2} fails  for some $Q\in (0,n)+\mathcal{ER}_0(N_1)$ with $|n|\leq 2N_1$. Then by ({\bf Hiv.a}) (for $r=r_1$) and the standard argument, $\mathcal{K}_{\omega,a}$ can be regarded as a semi-algebraic set of degree $\deg \mathcal{K}_{\omega,a}\leq N_1^C.$  Using Lemma \ref{sgmlem}, we know $\mathcal{K}_{\omega,a}\subset \bigcup_{1\leq l\leq N_1^C} J_l$ with each $|J_l|\sim \varepsilon+ \delta.$ In the following we fix one   $I_{a,l}$ and one $J_{l'}$ from the above reductions. and consider now the set 
%\begin{align}
%\mathcal{K}_{l,l'}= (\Pi_\omega I_0\times I_{a,l})\times J_{l'}
%\end{align}
%Since $\Gamma_{r_1}\cap I_0$ is defined by an equation in $(\omega,a)$ of degree at most $M^{Cr_1^3},$ we have 
%$$\deg\mathcal{K}_{l,l'}\cap (\Gamma_{r_1}\times \R)\leq N_1^CM^{r_1^3}.$$
%By the LDT mentioned as above, for each $(\omega,a)$, the $(\omega,a)$-section set $\{\sigma\in\R:\ (\omega, a,\sigma)\in\mathcal{K}_{l,l'}\}$ has measure at most $e^{-N_1^{\rho}}.$ So applying the Fubini's theorem, we obtain 
%$${\rm meas}_{b+1}( \mathcal{K}_{l,l'}\cap (\Gamma_{r_1}\times \R))\leq e^{-N_1^{\rho}}.$$

 Our next aim is to estimate 
 \begin{align}\label{hypest}
{\rm meas}\left(\bigcup_{N/2\leq |\bm k|\leq N}\left\{\bm\omega:\  (\bm\omega, \bm k\cdot\bm \omega)\in \mathcal{K}\right\}\right).
 \end{align}
At this stage, we need an important  projection lemma.
\begin{lem}[cf. e.g., \cite{Bou07}]\label{proj}
 Let $Y\subset[0,1]^{d=d_1+d_2}$ be a semi-algebraic set of degree $\deg  Y=B$ and $\mathrm{meas}(Y)\leq\eta$, where
$\log B\ll \log\frac{1}{\eta}.$
Denote by $(\bm x, \bm y)\in[0,1]^{d_1}\times[0,1]^{d_2}$ the product variable. Suppose
$ \eta^{\frac{1}{d}}\leq\epsilon.$
Then there is a decomposition of $Y$ as
\begin{align*}
Y=Y_1\cup Y_2
\end{align*}
with the following properties: The projection of $Y_1$ on $[0,1]^{d_1}$ has small measure
$$\mathrm{meas}(\Pi_{\bm x}Y_1)\leq {B}^{C(d)}\epsilon,$$
and $Y_2$ has the transversality property
\begin{align*}
\mathrm{meas}(\mathbb{H}\cap Y_2)\leq B^{C(d)}\epsilon^{-1}\eta^{\frac{1}{d}},
\end{align*}
where $\mathbb H$ is any $d_2$-dimensional hyperplane in $\R^d,$  s.t.,
$\max\limits_{1\leq j\leq d_1}|\Pi_{\mathbb H}(\bm e_j)|<{\epsilon}$
(we denote by $\bm e_1,\cdots,{\bm e}_{d_1}$ the $\bm x$-coordinate  standard   basis).
\end{lem}
\begin{rem} This lemma permits scaling.
%If $\mathcal{S}_1\times \mathcal{S}_2\subset [0,\delta]^{d=d_1+d_2}$ and $\eta=\delta^d\eta$,  then we have the estimate $\mathrm{mes}_{d_1}(\Pi_{x_1}\mathcal{S}_1)\leq \delta^{d} {B}^{C(d)}\epsilon$, 
It is based on the Yomdin-Gromov triangulation theorem \cite{Gro87}; for a complete proof, see \cite{BN19}. 
\end{rem}
Taking $\eta=C(\varepsilon+\delta)^be^{-N_1^{\rho}}, \epsilon =10/N$,  %and identifying  $(\Gamma_{r_1}\cap I_0)\times \R$ as an interval on $\R^{b+1}$,  
we obtain  %by  $C'\rho>3$ that
$$C(\varepsilon+\delta)^be^{-\frac{1}{b+1} N_1^{\rho}}\leq C(\varepsilon+\delta)^be^{-\frac{1}{b+1}(\log N )^{4}}\ll \epsilon.$$
Using Lemma \ref{proj}, we have a decomposition 
	$$ \mathcal{K}=\mathcal{K}_1\cup\mathcal{K}_2,$$
	where ${\rm meas}\ (\Pi_{\bm \omega}\mathcal{K}_1)\leq C(\varepsilon+\delta)^{b+1}N_1^CM^{C{\tilde r}_1^3}M^{-{\tilde r}}\leq (\varepsilon+\delta)^{b+1}M^{-\frac {\tilde r}2}$ and the factor $(\varepsilon+\delta)^{b+1}$ comes from scaling when applying Lemma \ref{proj}. 
	Note however that $|\bm k|\geq N/2$ cannot ensure $\min_{1\leq \ell\leq b}|k_\ell|\geq N/2$.  So the hyperplane $\{(\bm\omega, \bm k\cdot\bm \omega)\}$  does not satisfy  the steepness condition. To address this issue,  we will  need to take into account of all possible directions  and apply Lemma  \ref{proj} for $b$ times, as done by Bourgain  (cf.  (3.26) of \cite{Bou07} and also (142) of \cite{LW22}). This then leads to an upper bound  $\
	C(\varepsilon+\delta)^{b+1}M^{-\frac{{\tilde r}+1}{C(b)}}$ on  \eqref{hypest}, where $C(b)>0$ only depends on $b$.    
	
	Taking into account of all $J$,  we have shown the existence of $\tilde I_1\subset \tilde I$ with ${\rm meas}_b(\tilde I_1)\leq C(\varepsilon+\delta)^{b+1}N_1^CM^{-\frac{{\tilde r}+1}{C(b)}}\leq C(\varepsilon+\delta)^{b+1}M^{-\frac{{\tilde r}+1}{C(b)}}$ so that for $\bm\omega\in \tilde I\setminus \tilde I_1$,  the estimates \eqref{tqr11} and \eqref{tqr12} hold  for all  $Q\in\mathcal{C}_2$.  Let $I_0$ range over $\mathcal{I}_{r_1}$. The total measure removed from $\Pi_{\bm\omega}\Gamma_{r_1}$ is at most 
	$ C(\varepsilon+\delta)^{b+1}M^{{\tilde r}_1^C}M^{-\frac{{\tilde r}+1}{C(b)}}\leq (\varepsilon+\delta)^{b+1}M^{-\frac{{\tilde r}+1}{C(b)}} $. Since \eqref{tqr11}--\eqref{tqr12} allow $O(e^{-N^2_1})$ perturbation of $(\bm\omega, \bm a)$ (again using Lemma \ref{lwlem}) and $\|\Gamma_r-\Gamma_{r_1}\|\lesssim \delta_{r_1}\ll e^{-N^2_1}$, we obtain a subset $\Gamma_r'\subset  \Gamma_r$ with $ {\rm meas}(\Pi_{\bm\omega}\Gamma_r')\leq (\varepsilon+\delta)^{b+1}M^{-\frac{{\tilde r}+1}{C(b)}}$ so that \eqref{tqr11} and \eqref{tqr12} hold on 
	\begin{align*}
	\bigcup_{I\in\mathcal{I}_{r_1}}(I\cap (\Gamma_r\setminus\Gamma_r')),
	\end{align*}
	and hence on 
	\begin{align*}
	\bigcup_{I\in\mathcal{I}_{r}}(I\cap (\Gamma_r\setminus\Gamma_r')).
	\end{align*}
	\end{itemize}
	{Next,  by the induction  hypothesis at step $r$, we have that \eqref{hivc1} and \eqref{hivc2} hold  for  $r$. Also, by (\textbf{Hii})  at step $r$, if $|\bm x-\bm x'|\leq (\tilde r+1)^{5C}\log M,$ then (we hide the dependence on $\xi, \xi', \bm\omega, \bm a$)
	\begin{align*}%\label{perturbation 1}
	&\ \ \ 	|{T}_{\vec u^{(r)}} (\bm x, \bm x')- {T}_{\vec u^{(r-1)}}(\bm x, \bm x')| \\
		\notag    &\leq C\Vert \Delta_{\rm cor}\vec u^{(r)}\Vert\\
		\notag	&\leq C\delta_r \leq CM^{-(\tilde r+1)^{10C}} \\
		\notag	&\leq M^{-10(\tilde r+1)^C}e^{-c|\bm x-\bm x'|}.
	\end{align*}
If $|\bm x-\bm x'|>(\tilde r+1)^{5C}\log M$, by  (\textbf{Hiii})  at step $r$ and Lemma  \ref{powerf}, we have 
	\begin{align*}%\label{perturbation 2}
		&\ \ \ |{T}_{\vec u^{(r)}}(\bm x, \bm x')-{T}_{\vec u^{(r-1)}}(\bm x, \bm x')| \\
		 \notag   &\leq C(1+|\bm x-\bm x'|)^{C}e^{-c|\bm x-\bm x'|}\\
		\notag	&\leq M^{-10(\tilde r+1)^C}e^{-(c-(\tilde r+1)^{-3C})|\bm x-\bm x'|}. 
	\end{align*}
	The above estimates  imply  that for all $\bm x, \bm x',$
	\begin{equation*}%\label{perturbation by approimation to operator 1}
		|{T}_{\vec u^{(r)}}(\bm x, \bm x')-{T}_{\vec u^{(r-1)}}(\bm x, \bm x')|\leq M^{-10(\tilde r+1)^C}e^{-(c-(\tilde r+1)^{-3C})|\bm x-\bm x'|}. 
	\end{equation*}
	Hence,  combining  \eqref{hivc1}--\eqref{hivc2} (at step $r$) and  Lemma \ref{lwlem},  it follows that for any $(\bm\omega,\bm a)\in \cup_{I\in \mathcal I_r}I$,  we  can replace $\vec u^{(r-1)}$ in \eqref{hivc1}--\eqref{hivc2}  with  $\vec u^{(r)}$. That is to say,
	\begin{align}\label{CB3161}
	\|H_{M^{\tilde r}}^{-1}(\vec u^{(r)})\|\leq 2M^{{\tilde r}^C},
	\end{align}
	and for $|\bm x-\bm x'|>{\tilde r}^C,$
	\begin{equation}\label{CB3162}
			 |H_{M^{\tilde r}}^{-1}(\vec u^{(r)})(\bm x, \bm x')| \leq e^{-(c-(\tilde r+1)^{-3C})|\bm x-\bm x'|}.
	\end{equation}
	Assume $|(\bm \omega, \bm a)-(\bm \omega', \bm a')|\leq M^{-({\tilde r}+1)^{10C}}$. Similar to the above proof, we have for all $\bm x, \bm x',$
	\begin{align*}%\label{perturbation by approimation to operator 1}
		&|{H}(\vec u^{(r)}; (\bm \omega, \bm a))(\bm x, \bm x')-{H}(\vec u^{(r)}; (\bm \omega', \bm a'))(\bm x, \bm x')|\\
		\leq& M^{-10(\tilde r+1)^C}e^{-(c-(\tilde r+1)^{-3C})|\bm x-\bm x'|}. 
	\end{align*}
Using  Lemma \ref{lwlem} again, we have that  the  estimates \eqref{CB3161}--\eqref{CB3162}  and \eqref{tqr11}--\eqref{tqr12} (replacing $\vec u^{(r_1)}$ with $\vec u^{(r)}$) allow $O(M^{-({\tilde r}+1)^{10C}})$ perturbation of $(\bm\omega, \bm a)$.  }
	Combining conclusions in the above {\bf Case 1}--{\bf Case 2} and the  resolvent identity of \cite{HSSY} (cf. Lemma 3.6) gives a collection $\mathcal{I}_{r+1}$ of intervals of size $M^{-({\tilde r}+1)^{10C}}$
so that for $I\in\mathcal{I}_{r+1}$, 	
		\begin{align*}
\|H_{N}^{-1}(\vec u^{(r)})\|&\leq M^{({\tilde r}+1)^C},\\
|H_{N}^{-1}(\vec u^{(r)})((\bm x,\xi); (\bm x',\xi'))|&\leq e^{-c_{r+1}|\bm x-\bm x'|}\  {\rm for}\ |\bm x-\bm x'|>({\tilde r}+1)^{C},
	\end{align*}
which concludes ({\bf Hiv, c}) at scale $r+1.$ We remark that in the above   off-diagonal exponential decay estimates, $c_{r+1}\geq c_r-(\log M)^C({\tilde r}+1)^{-C}$, which implies  $\inf_{r}c_r>0$. This explains why we need the sublinear distant off-diagonal decay  in LDT.

Finally, we have 
\begin{align*}
&\ \ \ {\rm meas}\left(\Pi_{\bm\omega}(\Gamma_r\cap(\bigcup_{I'\in\mathcal{I}_{r}}I'\setminus\bigcup_{I\in\mathcal{I}_{r+1}}I))\right)\\
&\leq (\varepsilon+\delta)^{b+1}M^{-\frac{{\tilde r}+1}{C(b)}}+e^{-\frac12 N_1^{\rho^4}}<(\varepsilon+\delta)^{b+1}M^{-\frac{{\tilde r}+1}{C(b)}}.
\end{align*}
Recalling that $\bm\omega\to \bm a$ is a $C^1$ diffeomorphism and $\det (\frac{\partial \bm\omega}{\partial \bm a})\sim \delta^{b}$, we obtain since $\varepsilon\leq \delta$  that 
\begin{align*}
{\rm meas}\left(\Pi_{\bm a}(\Gamma_r\cap(\bigcup_{I'\in\mathcal{I}_{r}}I'\setminus\bigcup_{I\in\mathcal{I}_{r+1}}I))\right)\leq M^{-\frac{{\tilde r}+1}{C(b)}},
\end{align*}
	which yields ({\bf Hiv, d}) at  step $r+1.$
	\hfill $\square$
	
	%Finally, let us  the exponential decay of $q^{(r)}.$
	
	%\end{proof}
	
	\begin{proof}[Proof of Theorem \ref{mthm}]
	The proof of Theorem \ref{mthm} is a direct corollary of the Inductive Theorem. We refer to \cite{BW08,LW22} and sects.~V and VI of  \cite {KLW23} for further details. 
	\end{proof}
	%\subsection{The convergence}
		\section*{Acknowledgements}

Y. Shi was  supported  by the National Key R\&D Program (2021YFA1001600) and  the  NSFC  (12522110). 
W.-M. Wang acknowledges support from the
CY Initiative of Excellence, ``Investissements d'Avenir" Grant No. ANR-16-IDEX-0008. The authors would like to thank the handling editor and the anonymous reviewers  for valuable suggestions.

\appendix

\section {Diophantine estimates }

Below we provide Diophantine estimates when $V(\bm \theta)$ are (arbitrary) trigonometric polynomials. These estimates hold on a large set in $(\bm\alpha, \bm\theta)$.
The main idea is to utilize an appropriate generalized Wronskian approach. To our knowledge, these estimates did not appear to be known in the literature before. 
Given the relation between analytic functions (hence trigonometric polynomials) and generalized Wronskians, this approach also seems natural in hindsight.

To our knowledge, using transversality properties to make Diophantine approximations on manifolds was first developed by Pyartli \cite{Pya69}. This method was later 
	introduced into the study of KAM theory  \cite{XYQ97,Eli02,Bam03}. It turns out that this idea also plays a key role in the Craig-Wayne-Bourgain type argument used in the present paper.
	In addition, we can handle  Diophantine estimates on general quasi-periodic polynomials on $[0,1]^d$. The main difference between \cite{XYQ97,Eli02, Bam03} and the present work is that we have no   {\it prior } transversality estimate,  which requires much additional effort.

Recall that   
\begin{align*}
V(\bm\theta)=\sum_{\bm \ell \in \Gamma_K}v_{\bm\ell}\cos2\pi(\bm \ell  \cdot\bm\theta)\ {\rm with}\  \bm \theta\in [0,1]^d, 
	\end{align*}
	where  $\Gamma_K\subset [-K, K]^d\setminus\{\bm 0\}$ ($K\in \N$)  is {\it maximal} so that   
	\begin{equation}\label{Gammk}
	\begin{aligned}
	&\text {if }\bm \ell\in\Gamma_K, {\rm then},   \ell_s\neq 0, \forall\  1\leq s\leq d; \\
	&\text {and if }\ \bm\ell,\bm\ell'\in\Gamma_K, \ {\rm then}\ \bm\ell+\bm\ell'\neq \bm 0, 
	\end{aligned}
	\end{equation}
	%Then  $$\#\Gamma_K=\frac{(2K+1)^d-1}{2}.$$
and 
	\begin{align*}
	\bm v=(v_{\bm \ell})_{\bm \ell\in\Gamma_K}\in \R^{\#\Gamma_K}\setminus\{\bm 0\}.
	\end{align*}
We also assume $V(\bm \theta)$ is {\it non-degenerate}.

To handle $\mu_{\bm n}=V(\bm n\bm \alpha+\bm\theta)$, we need  the following  \L ojasiewicz type Lemma.

\begin{lem}[cf. e.g., Lemma 4.2  in  \cite{JLS20}]\label{loj} There exist some  constants $C_V>0, 0<c_V<1$ depending only on $V$ so that, for any $\eta>0,$ one has 
$${\rm meas}\left(\left\{\bm \theta\in[0,1]^d:\ |V(\bm \theta)|\leq \eta\right\}\right)\leq C_V\eta^{c_V}.$$
\end{lem}

%\subsection{The Diophantine estimates}
%For $\bm x, \bm y\in \R^d$, denote 
%\begin{align*}
%\bm x\bm y&=(x_1y_1,\cdots, x_d y_d)\in \R^d,\\
%\bm x\cdot\bm y&=x_1y_1+\cdots+x_dy_d\in\R,
%\end{align*}
%	and consider
%	\begin{align*}
	%V(\bm\theta)=\sum_{\bm \ell \in \Gamma_K}p_{\bm\ell}\cos2\pi(\bm \ell  \cdot\bm\theta)\ {\rm with}\  \bm \theta\in [0,1]^d, 
	%\end{align*}
	%where  $\Gamma_K\subset [-K, K]^d\setminus\{0\},\ K\geq 1$  has the property that  
	%\begin{equation}\label{Gammk}
	%\begin{aligned}
	%&\text {if }\bm \ell\in\Gamma_K, {\rm then}, \bm \ell_k\neq 0, \forall k=1, 2, ..., d, \\
	%&\text {and if }\ \bm\ell,\bm\ell'\in\Gamma_K, \ {\rm then}\ \bm\ell+\bm\ell'\neq 0; 
	%\end{aligned}
	%\end{equation}
	%Then  $$\#\Gamma_K=\frac{(2K+1)^d-1}{2}.$$
%and 
	%\begin{align*}
	%\bm p=(p_{\bm \ell})_{\bm \ell\in\Gamma_K}\in \R^{\#\Gamma_K}\setminus\{0\}.
	%\bm r=(r_1,\cdots, r_K), \bm t=(t_1,\cdots, t_K)\in \Z^K
	%(p_\ell,q_{\ell})_{1\leq \ell\leq K}\in \R^K\times \R^K,\ (r_\ell,t_{\ell})_{1\leq \ell\leq K}\in \N^K\times \N^K
	%\end{align*}
	%\begin{rem} Note that  \eqref{Gammk} is a non-degeneracy condition on each term in the polynomial $P(\bm\theta)$,
	%cf. the paragraph below \eqref {(1.2)}.
	%\end{rem}

%For  $\bm n\in\Z^d$, let 
%\begin{align*}
%V_{\bm n}=V_{\bm n}(\bm\theta, \bm \alpha)=P(\bm\theta+\bm n \bm \alpha),
%\end{align*}
%where 
%\begin{align*}
%\bm \alpha \in(0, 1]^d.
%\end{align*}
%In the present we are interested in linear combinations  of $V_{\bm n}$. More precisely,  let 
Next, we fix 
$$\bm n_1',\cdots, \bm n_D'\in \Z^d$$
to be $D$ ($D\geq 2$)  distinct lattice points. 
%satisfying   
%\begin{align}\label{ncond}
%\min_{1\leq i\leq d}| n_\ell(i)-n_{\ell '}(i)|>0\ {\rm for}\ \ell\neq \ell',
%\end{align}
%where $\bm n_{\ell}=(n_\ell (1),\cdots, n_\ell(d))$ and $\bm n_{\ell'}=(n_{\ell'} (1),\cdots, n_{\ell'}(d))$.  
We are mainly concerned with Diophantine estimates for 
%\begin{itemize}
%\item {\bf NLS case}: Let
\begin{align*}
\bm \omega'=(\omega_1',\cdots,\omega_D'),\ \omega_s'=V({\bm\theta+\bm n_{s}'\bm\alpha})\ (1\leq s\leq D).
\end{align*}
%\item  {\bf NLW case}: Let $$m\in [\|P\|_\infty+1, \|P\|_\infty+2]$$ with $\|P\|_\infty=\sup_{\bm\theta\in[0,1]^d}|P(\bm\theta)|.$ Define
%\begin{align*}
%\bm \omega_{{\rm NLW}}=(\omega_1,\cdots,\omega_D),\ \omega_\ell=\sqrt{P({\bm\theta+\bm n_{\ell}\bm\alpha})+m}\ (1\leq \ell\leq D).
%\end{align*}
%\end{itemize}

%For   ${\bm x}=(x_1,\cdots,x_d)$,  let 
%$$ |\bm x|=\max_{1\leq i\leq d}|x_i|.$$ 
Throughout this Appendix,  all constants $C,c>0$ are assumed  to (depend only on $d, K, D, V$)  be  independent on  $\bm n_s'$ ($1\leq s\leq D$).  %For convenience, we omit  this dependence.

Our main result  on  the Diophantine estimates  is
\begin{thm}\label{nlsthm}
Let   $0<\eta<1$.  We have for  $R=R_D=D\cdot(\#\Gamma_K)=D\frac{(2K+1)^d-1}{2}$, 
\begin{align*}
{\rm meas}({\rm DC}_{\rm nls}(\eta))\geq 1-C\cdot (\max_{1\leq s \leq D}|\bm n_s'|)^{5dR^2}\eta^{\frac{1}{10dR^2}},
\end{align*}
where  $C=C(d,K,D, V)>0$ and 
\begin{align*} 
\nonumber&\ \ \ {\rm DC}_{\rm nls}(\eta)\\
&=\left\{(\bm \alpha, \bm\theta)\in[0,1]^{2d}:\ |\bm k'\cdot\bm\omega'|\geq \frac{\eta}{|\bm k' |^{4R^2}}\ {\rm for}\ \forall\ \bm k'\in \Z^D\setminus\{\bm 0\}\right\}. 
\end{align*}
We also have  for any $\bm k'\in\Z^D\setminus\{\bm0\},$
\begin{align*}
{\rm meas}(\{(\bm \alpha, \bm\theta)\in[0,1]^{2d}:\ |\bm k'\cdot\bm\omega'|\leq \eta \})\leq C\cdot (\max_{1\leq s \leq D}|\bm n_s'|)^{5dR^2}\eta^{\frac{1}{10dR^2}}. 
\end{align*}
%with  $\bm \omega=\bm\omega_{\rm NLS}$.
\end{thm}
\begin{rem}
In contrast, if $\bm \omega'$ were a free parameter varying in $[0,1]^D$,  we would have 
\begin{align} 
\nonumber&{\rm meas}\left(\left\{\bm \omega'\in[0,1]^{D}:\ |\bm k'\cdot\bm\omega'|\geq \frac{\eta}{|\bm k' |^{4R^2}}\ {\rm for}\ \forall\ \bm k'\in \Z^D\setminus\{\bm 0\}\right\}\right)= 1-O(\eta).
\end{align}
So the absence of independence of the coordinates of $\bm \omega'$ may lead to the corresponding measure bound  (from $1-O(\eta)$ to)  $1-O(\eta^{\frac{1}{10dR^2}}).$
\end{rem}

Recall that  we have fixed in Theorem \ref{mthm}, 
\begin{align*}
\bm n_1,\cdots, \bm n_b\in\Z^d,\  \bm\omega^{(0)}=(V(\bm n_1\bm \alpha+\bm\theta),\cdots,  V(\bm n_b\bm \alpha+\bm \theta)),  
\end{align*}
and for arbitrary $\bm n$, we have  $\mu_{\bm n}=V(\bm n\bm \alpha+\bm\theta).$

{As a consequence of Lemma \ref{loj} and Theorem \ref{nlsthm}, we obtain 
\begin{cor}\label{dccor}
There exist some $0<c_1=c_1(b,d,K,V)<\frac{1}{100b}, C_1=C_1(b,d, K, V)>1$ such that, if  $0<\varepsilon+\delta\leq \delta_0(b,d,K,V,\max\limits_{1\leq\ell\leq b}|\bm n_\ell|)\ll1$, then there is some $\mathcal M \subset [0, 1]^{2d}$
satisfying  $${\rm meas}([0,1]^{2d}\setminus\mathcal M)\leq \log^{-1}\frac{1}{\varepsilon+\delta}$$
so that the following properties hold true for $(\bm \alpha, \bm \theta)\in\mathcal M$ and  $L_{\varepsilon,\delta}=100(\varepsilon+\delta)^{-c_1}$.
\begin{itemize}
\item[(i).]  For any $\bm n\neq \bm n'$  satisfying  $|\bm n|, |\bm n'|\leq L_{\varepsilon,\delta}$, 
$$|\mu_{\bm n}-\mu_{\bm n'}|\geq (\varepsilon+\delta)^{\frac{1}{8b}}.$$
\item[(ii).] For any $\bm k\in\Z^b$ satisfying $0<|\bm k|\leq 2L_{\varepsilon,\delta},$
$$|\bm k\cdot \bm\omega^{(0)}|\geq (\varepsilon+\delta)^{\frac{1}{8b}}.$$
\item[(iii).] For any $\log \frac{1}{\varepsilon+\delta}\leq L\leq L_{\varepsilon,\delta}$ and all $(\bm k,\bm n)\in \Lambda_{L}\setminus\mathcal S,$
$$\min_{\xi=\pm1}|\xi \bm k\cdot \bm\omega^{(0)}+\mu_{\bm n}|\geq  L^{-C_1}.$$
%where $C_1=C_1(b,d, K, V)>0$; 
\item[(iv).] For any $(\bm k, \bm n, \bm n')\in\Z^b\times\Z^d\times \Z^d$ satisfying $|\bm k|\leq 2L_{\varepsilon,\delta}, |(\bm n, \bm n')|\leq L_{\varepsilon,\delta}$ and $\bm k\cdot\bm\omega^{(0)}+\mu_{\bm n}-\mu_{\bm n'}\not\equiv 0,$
$$|\bm k\cdot \bm\omega^{(0)}+\mu_{\bm n}-\mu_{\bm n'}|\geq  (\varepsilon+\delta)^{\frac{1}{8b}}.$$
\end{itemize}
\end{cor}

\begin{proof}
It suffices to apply Lemma \ref{loj}  (i.e., $D=1$) and Theorem \ref{nlsthm} with $\eta=(\varepsilon+\delta)^{\frac{1}{8b}}, L^{-C_1}$ and $D=2, b+1, b+2$. Without loss of generality, we only deal with (iii) and (iv) since   (i) and (iv)  can be handled similarly.

(iii).  Denote by $\mathcal{M}_{{\rm iii}}$ the set of $(\bm \alpha, \bm \theta)\in[0,1]^{2d}$ so that the conclusion of (iii) holds true.  \\
For the case of $\bm k=\bm 0$, applying Lemma \ref{loj} yields for all $\bm \alpha\in[0,1]^d,$
\begin{align*}
&\ \ \ {\rm meas} \left(\bigcup_{\log \frac{1}{\varepsilon+\delta}\leq L\leq L_{\varepsilon,\delta}}\bigcup_{|\bm n|\leq L}\left\{\bm \theta\in[0,1]^{d}:\ |\mu_{\bm n}=V(\bm n\bm\alpha+\bm\theta)|\leq L^{-C_1}\right\}\right)\\
&\leq \sum_{\log \frac{1}{\varepsilon+\delta}\leq L\leq L_{\varepsilon,\delta}}\sum_{|\bm n|\leq L} C_VL^{-C_1c_V} \\
%&\leq \sum_{\log \frac{1}{\varepsilon+\delta}\leq L\leq L_{\varepsilon,\delta}}\sum_{|\bm n|\leq L} C_VL^{-C_1c_V} \\
&\leq \sum_{\log \frac{1}{\varepsilon+\delta}\leq L\leq L_{\varepsilon,\delta}} C(V,d)L^{-C_1c_V+d}\\
&\leq \log^{-3}\frac{1}{\varepsilon+\delta} \ ({\rm choosing}\  C_1\  {\rm\  so\  that}\  C_1c_V-d>10),
\end{align*}
where in the last inequality we assume $0<\varepsilon+\delta\leq \delta_0(V,d)\ll1.$  
\ \\ 
For the case of $\bm k\neq\bm 0$ and $\bm n\notin\{\bm n_1,\cdots, \bm n_b\}$, we apply Theorem \ref{nlsthm} with $D=b+1, \eta=L^{-C_1}$ and $\bm k'=(\pm\bm k, 1)$ to get for 
$$ \mathcal M_{\rm iii}^{(1)}(L):=  \bigcup_{|(\bm k, \bm n)|\leq L,\bm n\notin\{\bm n_1,\cdots,\bm n_b\}}\left\{(\bm \alpha, \bm \theta)\in[0,1]^{2d}:\ \min_{\xi=\pm1}|\xi \bm k\cdot \bm\omega^{(0)}+\mu_{\bm n}|\leq  L^{-C_1}\right\}, $$
the following 
\begin{align*}
&\ \ \ {\rm meas} \left(\bigcup_{\log \frac{1}{\varepsilon+\delta}\leq L\leq L_{\varepsilon,\delta}}\mathcal M_{\rm iii}^{(1)}(L)\right)\\
&\leq \sum_{\log \frac{1}{\varepsilon+\delta}\leq L\leq L_{\varepsilon,\delta}} C(b,d,K,V)L^{2d}L^{5dR^2_{b+1}}L^{-\frac{C_1}{10dR^2_{b+1}}}\ ({\rm assume}\  L>\max\limits_{1\leq\ell\leq b}|\bm n_\ell|) \\
&\leq \sum_{\log \frac{1}{\varepsilon+\delta}\leq L\leq L_{\varepsilon,\delta}}C(b,d,K,V) L^{-10} \  ({\rm choosing}\  C_1\geq C({b,d, K})\gg1)\\
%&\leq \sum_{\log \frac{1}{\varepsilon+\delta}\leq L\leq L_{\varepsilon,\delta}} C(V,d)L^{-Cc_V+d}\\
&\leq \log^{-3}\frac{1}{\varepsilon+\delta}, 
\end{align*}
where in the last inequality we assume $0<\varepsilon+\delta\leq \delta_0(b,d, K, V,\max\limits_{1\leq\ell\leq b}|\bm n_\ell|)\ll1.$ 
The case of $\bm k\neq\bm 0$ and $\bm n\in\{\bm n_1,\cdots, \bm n_b\}$ can be handled similarly and we omit the details.  Since there are only 3 cases to deal with, $C_1=C_1(b,d,K,V)>0$ can be well-defined. 
\ \\ 
By combining all the above estimates, we obtain  the desired measure bound on $\mathcal{M}_{{\rm iii}}$, namely, 
$${\rm meas}([0,1]^{2d}\setminus \mathcal M_{\rm iii})\leq \log ^{-2}\frac{1}{\varepsilon+\delta}.$$

(iv).   If $x,y\in\R$, we denote $x\wedge y=\min\{x, y\}$ and $x\vee y=\max\{x, y\}$. Denote by $\mathcal{M}_{{\rm iv}}$ the set of $(\bm \alpha, \bm \theta)\in[0,1]^{2d}$ so that the conclusion of (iv) holds true.  First we consider  the case of  $(\bm k, \bm n, \bm n')\in \mathcal T_1$ with  $\mathcal T_1$  the set of  all $(\bm k, \bm n, \bm n')$ satisfying:  $0<|\bm k|\leq 2L_{\varepsilon,\delta}, $   $|(\bm n, \bm n')|\leq L_{\varepsilon,\delta}$, $ \bm n\neq \bm n'$,  ${\rm both}\ \bm n\ {\rm and}\ \bm n'\ \notin \{\bm n_1,\cdots, \bm n_b\}.$ We let 
\begin{align*}
 &\ \ \  \mathcal M_{\rm iv}^{(1)}:= \bigcup_{(\bm k, \bm n, \bm n')\in \mathcal T_1}\left\{(\bm \alpha, \bm \theta)\in[0,1]^{2d}:\ \min_{\xi=\pm1}|\bm k\cdot \bm\omega^{(0)}+\mu_{\bm n}-\mu_{\bm n'}|\leq  (\varepsilon+\delta)^{\frac{1}{8b}}\right\}. 
\end{align*}
So we can apply Theorem \ref{nlsthm} with $D=b+2,\eta=(\varepsilon+\delta)^{\frac{1}{8b}}$ and $\bm k'=(\bm k, 1, -1)$ to get 
\begin{align*}
{\rm meas}\left( \mathcal M_{\rm iv}^{(1)}\right)&\leq \sum_{(\bm k, \bm n, \bm n')\in \mathcal T_1}C(b,d,K,V) ((\max\limits_{1\leq \ell\leq b}|\bm n_\ell|)\vee |\bm n| \vee |\bm n'|)^{5dR^2_{b+2}}(\varepsilon+\delta)^{\frac{1}{80bdR^2_{b+2}}}\\
&\leq C(b,d,K,V)L_{\varepsilon,\delta}^{b+2d+5dR^2_{b+1}}(\varepsilon+\delta)^{\frac{1}{80bdR^2_{b+2}}}\ ({\rm assume}\  L_{\varepsilon,\delta}>\max\limits_{1\leq\ell\leq b}|\bm n_\ell|) \\
&\leq C(b,d,K,V){(\varepsilon+\delta)}^{-c_1(b+2d+5dR^2_{b+1})+\frac{1}{80bdR^2_{b+2}}}\\
%&\leq \sum_{\log \frac{1}{\varepsilon+\delta}\leq L\leq L_{\varepsilon,\delta}} C(V,d)L^{-Cc_V+d}\\
&\leq {(\varepsilon+\delta)}^{\frac{1}{100bdR^2_{b+2}}}\\
 &({\rm choosing}\ 0<c_1\leq \frac{1}{(b+2d+5dR^2_{b+1})720bdR^2_{b+2}}), 
%&\leq \log^{-2}\frac{1}{\varepsilon+\delta}, 
\end{align*}
where in the last inequality we assume $0<\varepsilon+\delta\leq \delta_0(b,d, K, V,\max\limits_{1\leq\ell\leq b}|\bm n_\ell|)\ll1.$  \\
Next, we consider the case  of $(\bm k, \bm n,\bm n')$ with $\bm k\cdot\bm\omega^{(0)}+\mu_{\bm n}\equiv 0$. In this case, applying Lemma \ref{loj} yields for all $\bm \alpha\in[0,1]^d,$
\begin{align*}
&\ \ \ {\rm meas} \left(\bigcup_{|\bm n'|\leq L_{\varepsilon,\delta}}\left\{\bm \theta\in[0,1]^{d}:\ |\mu_{\bm n'}=V(\bm n'\bm\alpha+\bm\theta)|\leq (\varepsilon+\delta)^{\frac{1}{8b}}\right\}\right)\\
&\leq\sum_{|\bm n'|\leq L_{\varepsilon,\delta}} C_V (\varepsilon+\delta)^{\frac{c_V}{8b}} \\
%&\leq \sum_{\log \frac{1}{\varepsilon+\delta}\leq L\leq L_{\varepsilon,\delta}}\sum_{|\bm n|\leq L} C_VL^{-C_1c_V} \\
&\leq  C(V,d) (\varepsilon+\delta)^{-c_1d+\frac{c_V}{8b}}\\
&\leq (\varepsilon+\delta)^{\frac{c_V}{10b}} \ ({\rm choosing}\  0<c_1\leq \frac{c_V}{72db}),
\end{align*}
where in the last inequality we assume $0<\varepsilon+\delta\leq \delta_0(b,d,V)\ll1.$\\ 
Other cases are easier  to handle, and we omit the details. By taking account of all the above estimates, we have 
$${\rm meas}([0,1]^{2d}\setminus \mathcal M_{\rm iv})\leq (\varepsilon+\delta)^{\frac{c_V}{1000bdR^2_{b+2}}}.$$
Note that since there are only finitely many  (say, at most 10)  cases to deal with, $c_1=c_1(b,d, K, V)>0$ can be well-defined. 
\end{proof}

}

%Next, we consider the NLW case. 
%Denote by $\{\bm e_\ell\}_{1\leq \ell\leq D}$ the canonical basis of $\Z^D.$   Define for $\gamma>0,$
%\begin{align}\label{Gga}
%{G}_\gamma=\left\{ (\bm \alpha, \bm\theta)\in[0,1]^{2d}:\  \min_{1\leq \ell<\ell'\leq D}|(\bm e_{\ell}-\bm e_{\ell '})\cdot\bm\omega_{\rm NLS}|\leq {\gamma}  \right\}.
%\end{align}
%%For the NLW frequencies, we have
%Recall Lemma \ref{lajlem}, we have
%  \begin{thm}[]\label{nlwthm}
%  Let  $0<\delta\ll1, 0<\gamma\leq \delta^{10d R^2}$ and $(\bm \alpha, \bm \theta)\in [0,1]^{2d}\setminus G_\gamma$ with $G_\gamma$ given by \eqref{Gga}.  We have for all $0<\gamma_1\leq \gamma^{2D(D^2-1)} $ and $\tau> 2D^2$, 
%\begin{align*}
%{\rm meas}({\rm DC}_{\rm nlw}(\gamma_1;\tau))\geq 1-C\gamma_1^{\frac{1}{2D}}
%\end{align*}
%where
%\begin{align*} 
%\nonumber&\ \ \ {\rm DC}_{\rm nlw}(\gamma_1;\tau)\\
%&=\left\{m\in[\|P\|_\infty+1, \|P\|_\infty+2]:\ |\bm k\cdot\bm\omega|\geq \frac{\gamma_1}{|\bm k |^{\tau}}\ {\rm for}\ \forall\ \bm k\in \Z^D\setminus\{0\}\right\}
%\end{align*}
  %with  $\bm \omega=\bm\omega_{\rm NLW}$.
%  \end{thm}
  
  %\begin{rem} The above two theorems and their proofs can be easily generalized to finite Fourier series with both cosine and sine terms. 
 % \end{rem}

%\subsection{Structure of the paper} In sect.~2, we introduce some useful lemmas and their detailed proofs. We prove Theorem~\ref{nlsthm} and Theorem~\ref{nlwthm} in sect.~3 and sect.~4, respectively. 

	\subsection{Some useful lemmas}
	In this section,  we will introduce some useful lemmas. % concerning transversality estimates. 
	
		\begin{lem}[\cite{KM98}]\label{km98}
	Let $I\subset \R$ be an interval of finite length (i.e., $0<|I|<\infty$) and $k\geq 1.$  If  $f\in C^{k}(I;\R)$ satisfies 
	\begin{align}\label{tv}
	\inf_{x\in I}|\frac{d^k}{dx^k} f(x)|\geq A>0,
	\end{align}
	then for all $\varepsilon>0,$
	\begin{align}\label{mea}
	{\rm meas}(\{x\in I:\ |f(x)|\leq \varepsilon\})\leq  \zeta_k( \frac\varepsilon A)^{\frac1k},
	\end{align}
where 
$\zeta_k=k(k+1)((k+1)!)^{\frac1k}.$
\end{lem}

	\begin{proof}
	The following  proof is taken from Kleinbock-Margulis \cite{KM98}.
	
	From \eqref{tv} and the Rolle's theorem,  the equation $\frac{d}{dx}f(x)=0$ has at most $k$ zeros on $I$. This then implies that  the set 
	$\{x\in I:\ |f(x)|\leq \varepsilon\}$ consists of at most $k+1$ intervals.  Denote by $I_1$ one of those subintervals  having maximal length. So we get 
	\begin{align*}
	{\rm meas}(\{x\in I:\ |f(x)|\leq \varepsilon\})\leq  (k+1)|I_1|.
		\end{align*} 
	%where $|\cdot|$ is the length of an interval. 
	
	To prove \eqref{mea}, it remains to estimate $|I_1|>0$.  We divide $I_1$ into $k$ equal parts by points $x_1,\cdots, x_{k+1}.$  Let $P(x)$ be the Lagrange polynomial of $f$ on points $x_1,\cdots, x_{k+1},$ namely, 
	\begin{align*}
	P(x)=\sum_{i=1}^{k+1}f(x_i)\frac{\prod_{1\leq j\neq i\leq k+1} (x-x_i)}{\prod_{1\leq j\neq i\leq k+1} (x_j-x_i)}.
	\end{align*}
	By applying the Rolle's theorem $k$  times and since $f(x_i)=P(x_i)$ ($1\leq i \leq k+1$),  we can find $x_0\in I_1$ so that 
	\begin{align*}
	\frac{d^k}{dx^k} f(x_0)=\frac{d^k}{dx^k} P(x_0)=k!\sum_{i=1}^{k+1}\frac{f(x_i)}{\prod_{1\leq j\neq i\leq k+1}(x_j-x_i)},
	\end{align*}
	which together with \eqref{tv} yields
	\begin{align*}
	A\leq k! \sum_{i=1}^{k+1} \frac{\varepsilon }{|I_1|^k k^{-k}}.
	\end{align*}
This implies 
\begin{align*}
|I_1|\leq k(k+1)!)^{\frac 1k} (\frac\varepsilon A)^{\frac 1k}.
\end{align*}
We have finished the proof. 
	\end{proof}

In the following,  we turn to the analysis of multi-variable functions.  We first introduce the notations.
\begin{itemize} 
\item For any  ${\bm x}=(x_1,\cdots,x_d)$ and any $d\geq 1$, let 
$$|{\bm x}|_2=\sqrt{\sum_{i=1}^d|x_i|^2},\ |\bm x|_1=\sum_{i=1}^d|x_i|. $$%\ |\bm x|=\max_{1\leq i\leq d}|x_i|.$$

\item Let $\bm a=(a_1,\cdots,a_d), \bm b=(b_1,\cdots,b_d)$.  Define the $d$-dimensional interval to be
\begin{align}\label{hint}
I=I_{\bm a, \bm b}=\prod_{i=1}^d[a_i,b_i]\subset\R^d.
\end{align}
  Denote 
$$|\bm a\vee \bm b|=\sum_{i=1}^d \max(|a_i|, |b_i|).$$ 

\item 
For  ${\bm \beta}=(\beta_1,\cdots,\beta_d)\in\R^d\setminus\{\bm 0\}$,  define 
\begin{align*}
d_{\bm \beta}:=\sum_{i=1}^d\beta_i{\partial}_i,\ \partial_i:=\frac{\partial}{\partial x_i}.
\end{align*}
We  also write for $\ell \geq 1, $
$$d_{\bm \beta}^\ell=d_{\bm\beta}\cdots d_{\bm\beta}\ (\ell-{\rm times}).$$
Denote 
$$\nabla=(\partial_1,\cdots,\partial_d).$$
\item 
For a multi-index ${\bm \gamma}=(\gamma_1,\cdots,\gamma_d)\in \N^d$, denote 
\begin{align*}
|\bm \gamma|=\sum_{i=1}^d\gamma_i,\  \partial^{\bm \gamma}=\partial^{\gamma_1}_1\cdots\partial^{\gamma_d}_d.
\end{align*}
\end{itemize}
We have %a higher dimensional generation of Lemma \ref{lajlem}.
\begin{lem}\label{lajlem}
Fix $k\in \N$. Let $I=I_{\bm a, \bm b}\subset\R^d$ be defined by \eqref{hint}. Assume that  the function $f\in C^{k+1}(I;\R)$ satisfies for some $A>0,$
\begin{align}\label{ass1}
\inf_{\bm x\in I}\sup_{1\leq \ell \leq k}\left|{d_{\bm \beta}^\ell}f(\bm x)\right|\geq A. 
\end{align}
Let 
\begin{align*}
\|f\|_{k+1}=\sup_{\bm x\in I}\sup_{1\leq |\bm \gamma|\leq k+1}\left| {\partial^{\bm \gamma}}f(\bm x)\right|<\infty. 
\end{align*}
Then for $0<\varepsilon<1,$
\begin{align}\label{keye}
\nonumber&\ \ \ {\rm meas}(\{\bm x\in I:\ |f(\bm x)|\leq \varepsilon\})\\
&\leq C(\bm \beta,k,d) (\|f\|_{k+1}|\bm b-\bm a|_2+1)^d {|\bm a\vee\bm b|}^{d-1}A^{-d}(A\wedge 1)^{-1}{\varepsilon}^{\frac{1}{k}},
\end{align}
where $C=C(\bm \beta, k,d)>0$ depends only on $\bm\beta, k,d $ (but not on $f$).
\end{lem}
{\color{blue}\begin{rem}
If $f$ is a polynomial on $\R^d$,    \cite{CW01} proved  some   strong  distributional inequalities, which may  directly imply  \eqref{keye}.   However, in the present case, we have a weak transversality condition \eqref{ass1}. So we would like  to give an elementary proof applying for more general functions (cf. \cite{LSZ25} for an application),  which  also    produces   an  upper bound  in the estimate  \eqref{keye}  with   the  explicit  dependence  on $\|f\|_{k+1}$. 
\end{rem}}
\begin{proof}
The proof can be divided into two steps: \\

{\bf Step 1}. Assume that for some $1\leq k_1\leq k$ and interval $I_1=\prod_{i=1}^d[h_i,m_i],$
\begin{align}\label{low}
\inf_{\bm x\in I_1}\left|{d_{\bm\beta}^{k_1}}f(\bm x)\right|\geq A. 
\end{align}
We will  combine Lemma \ref{km98}  and Fubini's theorem to deal with the measure estimate.  Let $\bm e_1=\frac{\bm\beta}{|\bm\beta|_2}$  and choose a normalized orthogonal basis $\bm e_2,\cdots, \bm e_d\in\R^d$ of 
$$\bm\beta^{\perp}=\{\bm x\in\R^d:\  \bm x\cdot \bm\beta =0\}.$$
We can then define the coordinate transform:  
$$\bm x=\bm\varphi(s,\bm y), (s,\bm y)\in\R\times\R^{d-1},$$
via
\begin{align*}
\bm\varphi(s, \bm y)=s\bm e_1+\sum_{i=2}^d y_i\bm e_i.
\end{align*}
So the Jacobian  satisfies 
\begin{align*}
|\det\frac{\partial \bm x}{\partial(s,\bm y)}|=1.
\end{align*}
Define 
\begin{align*}
S_{\varepsilon}=\{\bm x\in I_1:\ |f(\bm x)|\leq \varepsilon\},
\end{align*}
and denote by $\chi_{(\cdot)}$  the indicator  function.  We have 
\begin{align}
\nonumber {\rm meas}(S_{\varepsilon})&=\int_{I_1}\chi_{S_{\varepsilon}}(\bm x) d\bm x\\
\label{jac}&=\int_{\varphi^{-1}(I_1)}\chi_{S_{\varepsilon}}(\bm\varphi(s,\bm y)) dsd\bm y.
\end{align}
We proceed to control $\bm\varphi^{-1}(I_1)$. From $\bm x=s\bm e_1+\sum_{i=2}^d y_i\bm e_i\in I_1$, we obtain
\begin{align*}
s= \bm x\cdot\bm e_1 \in [-|\bm h\vee \bm m|, |\bm h\vee \bm m|].%\ |\bm h\vee \bm m|=\sum_{i=1}^d\max{(|h_i|, |m_i|)}>0.
\end{align*}
Similarly, we also have $y_i\in [-|\bm h\vee \bm m|, |\bm h\vee \bm m|]$ for $2\leq i\leq d,$  which yields $\bm\varphi^{-1}(I_1)\subset[-|\bm h\vee \bm m|, |\bm h\vee \bm m|]^d$.  We define 
$$\Pi_1\bm\varphi^{-1}(I_1)=\{ \bm y\in\R^{d-1}:\ \exists\  s\in\R \ s.t.,\  (s,\bm y)\in \bm \varphi^{-1}(I_1)\} $$
and,  define for $\bm y\in\Pi_1\bm\varphi^{-1}(I_1), $ the interval $I_{\bm y}=\{s\in\R:\ (s,\bm y)\in \bm \varphi^{-1}(I_1)\}$ of finite length. 
Applying the Fubini's theorem together with \eqref{jac} implies that 
\begin{align*}
{\rm meas}(S_{\varepsilon})&= \int_{\Pi_1\bm\varphi^{-1}(I_1)}d\bm y\int_{I_{\bm y}}\chi_{S_{\varepsilon}}(\bm\varphi(s,\bm y)) ds\\
&\leq {\rm meas}_{d-1}(\Pi_1\bm\varphi^{-1}(I_1))\sup_{\bm y\in \Pi_1\bm\varphi^{-1}(I_1)}\int_{I_{\bm y}}\chi_{S_{\varepsilon}}(\bm\varphi(s,\bm y)) ds\\
&\leq 2^{d-1} |\bm h\vee \bm m|^{d-1}\sup_{\bm y\in \Pi_1\bm\varphi^{-1}(I_1)}\int_{I_{\bm y}}\chi_{S_{\varepsilon}}(\bm\varphi(s,\bm y)) ds\\
&\ \ \  ({\rm since}\   \Pi_1\bm\varphi^{-1}(I_1)\subset [-|\bm h\vee \bm m|, |\bm h\vee \bm m|]^{d-1}).
\end{align*}
For fixed $\bm y\in \Pi_1\bm\varphi^{-1}(I_1)$, we get 
\begin{align*}
&\ \ \  \int_{I_{\bm y}}\chi_{S_{\varepsilon}}(\bm\varphi(s,\bm y)) ds\\
 &={\rm meas}(\{s\in I_{\bm y}:\  |f(s\bm e_1+\sum_{i=2}^dy_i\bm e_i)|\leq \varepsilon\}).
\end{align*}
Denote $g(s)=f(s\bm e_1+\sum_{i=2}^dy_i\bm e_i)$. From \eqref{low}, we have for all $s\in I_{\bm y},$
{\begin{align*}
|\frac{d^{k_1}}{ds^{k_1}}g(s)|=|\frac {1}{|\bm \beta|_2^{k_1}}d^{k_1}_{\bm \beta} f(s\bm e_1+\sum_{i=2}^dy_i\bm e_i)|\geq \frac {A}{|\bm \beta|_2^{k_1}}.
\end{align*}}
So applying  Lemma \ref{km98} yields 
\begin{align}\label{step1}
{\rm meas}(S_{\varepsilon})\leq \zeta_{k_1}2^{d-1} |\bm h\vee \bm m|^{d-1}|\bm \beta|_2 (\frac{\varepsilon}{A})^{\frac{1}{k_1}}.
\end{align}

{\bf Step 2}.  We deal with the general case in this step.  So we first divide each  one dimensional interval $[a_i,b_i]$ into $N$ (will be specified below) equal subintervals $I_{i,j_i}$ ($1\leq j_i\leq N$).  This then induces a decomposition of $I$, namely, 
$$I=\bigcup_{1\leq  j_1,\cdots,j_d\leq N}\prod_{i=1}^d I_{i,j_{i}}=\bigcup_{\bm J\in ([1,N]\cap\N)^d} I_{\bm J}.$$
We will apply the argument proved in {\bf Step 1} on each $I_{\bm J}$. For this purpose,   we fix $I_{\bm J}$. For any $\bm x,\bm y\in I_{\bm J}$, we have 
\begin{align*}
|\bm x-\bm y|_2&\leq \sqrt{\sum_{i=1}^d \frac{|b_i-a_i|^2}{N^2}}\\
&\leq \frac{|\bm b-\bm a|_2}{N}.%\ (\bm b=(\rho,\cdots,b_d), \ \bm a=(a_1,\cdots,a_d)).
\end{align*}
Fixing any $\bm x_0\in I_{\bm J}$,  we get from the assumption  \eqref{ass1} that there exists $1\leq k_1\leq k$ with 
\begin{align*}
|d_{\bm \beta}^{k_1}f(\bm x_0)|\geq A>0.
\end{align*}
As a result, for any $\bm x\in I_{\bm J}$, we obtain  
\begin{align}
\nonumber |d_{\bm \beta}^{k_1}f(\bm x)|&\geq |d_{\bm \beta}^{k_1}f(\bm x_0)|-|d_{\bm \beta}^{k_1}f(\bm x)-d_{\bm \beta}^{k_1}f(\bm x_0)|\\
\nonumber&\geq A-\sup_{\bm x\in I_{\bm J}}|\nabla d_{\bm \beta}^{k_1}f(\bm x)|_2\cdot|\bm x-\bm x_0|_2\\
\nonumber&\geq A-\frac{\|f\|_{k_1+1}L_{\bm\beta,d}|\bm b-\bm a|_2}{N}\\
\label{lower1}&\geq A-\frac{\|f\|_{k+1}L_{\bm\beta,d}|\bm b-\bm a|_2}{N},
\end{align}
where 
\begin{align*}
L_{\bm\beta,d}=\sqrt{d}\sum_{1\leq l_1,\cdots,l_d\leq d}|\beta_{l_1}\cdots\beta_{l_d}|.
\end{align*}
In fact,   $L_{\bm\beta, d}$ is derived from 
\begin{align*}
|\nabla d_{\bm \beta}^{k_1}f(\bm x)|_2&=\sqrt{\sum_{j=1}^d |\sum_{1\leq l_1,\cdots,l_d\leq d}\beta_{l_1}\cdots\beta_{l_d} \partial_{jl_1\cdots l_d}^{k_1+1}f(\bm x)|^2}\\
&\leq \|f\|_{k_1+1}\sqrt{\sum_{j=1}^d (\sum_{1\leq l_1,\cdots,l_d\leq d}|\beta_{l_1}\cdots\beta_{l_d} |)^2}\\
&:= \|f\|_{k_1+1} L_{\bm \beta,d}.
\end{align*}
From \eqref{lower1}, we can set 
\begin{align*}
N=[2\|f\|_{k+1}L_{\bm \beta,d}|\bm b-\bm a|_2 A^{-1}]+1
\end{align*}
so that 
\begin{align*}
\inf_{\bm x\in I_{\bm J}}|d_{\bm \beta}^{k_1}f(\bm x)|\geq A/2>0,
\end{align*}
where $[x]$ denotes the integer part of $x\in\R.$
Thus applying \eqref{step1} yields (since $I_J\subset I, 0<\varepsilon<1$) 
\begin{align*}
{\rm meas}(\{\bm x\in I_{\bm J}:\ |f(\bm x)|\leq \varepsilon \})&\leq \zeta_{k_1}2^{d-1}|\bm a\vee \bm b|^{d-1} 2^{\frac{1}{k_1}}|\bm \beta|_2 (\frac{\varepsilon}{A})^{\frac{1}{k_1}}\\
&\leq  C(k,d)|\bm a\vee \bm b|^{d-1}|\bm \beta|_2 (A\wedge1)^{-1} \varepsilon^{\frac{1}{k}}.
\end{align*}
Collecting all $I_{\bm J}$ leads to,
\begin{align*}
&\ \ \ {\rm meas}(\{\bm x\in I:\ |f(\bm x)|\leq \varepsilon \})\\
&\leq C(k,d)N^d|\bm a\vee \bm b|^{d-1} |\bm\beta|_2(A\wedge 1)^{-1} \varepsilon^{\frac{1}{k}}\\% (\frac{\varepsilon}{A})^{\frac{1}{k}}\\
&\leq C(\bm\beta,k,d) (\|f\|_{k+1}|\bm b-\bm a|_2 A^{-1}+1)^d|\bm a\vee \bm b|^{d-1}(A\wedge 1)^{-1} \varepsilon^{\frac{1}{k}}\\ %(\frac{\varepsilon}{A})^{\frac{1}{k}} \\
&\leq C(\bm\beta,k,d) (\|f\|_{k+1}|\bm b-\bm a|_2+1)^d|\bm a\vee \bm b|^{d-1}A^{-d}(A\wedge 1)^{-1}{\varepsilon}^{\frac{1}{k}}.	
\end{align*}
We have proven \eqref{keye}. 
\end{proof}

Finally,  we  introduce a key lemma  (cf. Proposition in Appendix B, \cite{BGG85} for a more precise form) on determinant estimates.

\begin{lem}[]\label{BGGlem}
Let $\bm v^{(1)},\cdots, \bm v^{(r)}\in\R^r$ be $r$ linearly independent vectors with $|\bm v^{(\ell)}|_1\leq M$ for $1\leq \ell\leq r.$ Then for any  $\bm w\in\R^r$, we have
%\begin{align*}
%\sup_{1\leq l\leq r}|\bm w\cdot \bm v^{(l)}|<\xi\ {\rm for}\ \xi>0.
%\end{align*}
\begin{align*}
\max_{1\leq \ell\leq r}|\bm w\cdot\bm v^{(\ell)}|\geq r^{-3/2}M^{1-r}|\bm w|_2\cdot|\det \left[ \bm v^{(\ell)}\right]_{1\leq \ell\leq r}|.%<rK^{r-1}\xi.
\end{align*}
%where $|w|_2^2=\sum\limits_{1\leq j\leq r}w_j^2.$
\end{lem}
\begin{proof}
For completeness, we give a proof based on Cramer's rule and the Hadamard's inequality. We assume $\bm v^{(\ell)}$ is the row vector in $\R^r$ ($1\leq \ell\leq r$).  Denote by $X=[\bm v^{(1)},\cdots, \bm v^{(r)}]$ the $r\times r $ matrix and consider the equation %or  any $\bm w\in\R^r$
\begin{align*}
X\bm w=\bm \xi,\ \bm\xi=(\xi_1,\cdots,\xi_r)\ {\rm with}\ \xi_\ell=\bm w\cdot \bm v^{(\ell)}.
\end{align*}
By the Cramer's rule, we obtain 
\begin{align*}
\bm w=X^{-1} \bm \xi=\frac{X^*}{\det X} \bm \xi,
\end{align*}
where $X^*=[\tilde v_{ij}]_{1\leq i,j\leq r}$ denotes the adjugate matrix of $X$. As a result, we get 
\begin{align*}
|\bm w|_2\cdot|\det X|&=|X^*\bm\xi|_2\\
&=\sqrt{\sum_{i=1}^r|\sum_{j=1}^r\tilde v_{ij}\xi_j|^2}\\
&\leq \max_{1\leq j\leq r}|\xi_j|\sqrt{\sum_{i=1}^r(\sum_{j=1}^r|\tilde v_{ij}|)^2}\\
%&\leq \max_{1\leq j\leq r}|\xi_j|\\
&\leq r^{3/2} \max_{1\leq j\leq r}|\xi_j|\max_{1\leq i,j\leq r}|\tilde v_{ij}|.
\end{align*}
It remains to bound $\max_{1\leq i,j\leq r}|\tilde v_{ij}|,$ which will be completed by using the Hadamard's inequality. In fact, we have 
\begin{align*}
|\tilde v_{ij}|&\leq \prod_{1\leq\ell\neq i\leq r}\sqrt{\sum _{1\leq m\neq j\leq r} |v^{(\ell)}_m|^2}\\
&\leq \prod_{1\leq \ell\neq i\leq r}{\sum _{1\leq m\neq j\leq r} |v^{(\ell)}_m|}\\
&\leq M^{r-1}\ ({\rm since}\ \max_{1\leq \ell\leq r}|\bm v^{(\ell)}|_1\leq M).
\end{align*}
This finishes the proof. 
\end{proof}

\subsection{Proof of Theorem~\ref{nlsthm}}	

In this section,  we prove Theorem~\ref{nlsthm}.
\begin{proof}[Proof of Theorem~\ref{nlsthm}]

%\subsection{The NLS case}
Recall that %for $\bm \omega=\bm\omega_{\rm NLS}, $ 
\begin{align*}
\bm k'\cdot\bm \omega'&=\sum_{s=1}^Dk_s' V(\bm \theta+\bm n_s'\bm\alpha)\\
&=\sum_{s=1}^D\sum_{\bm\ell\in\Gamma_ K}k_s'v_{\bm\ell} \cos2\pi\bm \ell\cdot(\bm\theta+\bm n_s'\bm \alpha)\\
&= (\bm k'\otimes \bm v)\cdot \bm V,
%&=(\bm k\otimes \bm p, \bm k\otimes \bm q)\cdot (\bm V, \bm W),
\end{align*}
where 
\begin{align*}
&\bm k'\otimes \bm v=(k_s' v_{\bm\ell})_{1\leq s\leq D,   \bm\ell\in \Gamma_ K},\\
%\bm k\otimes \bm q=(k_iq_\ell)_{1\leq i\leq D, 1\leq \ell \leq K}\in\R^{DK},\\
&\bm V=(V_{(s, \bm\ell)})_{1\leq s\leq D,   \bm\ell \in \Gamma_ K}:=(\cos2\pi\bm \ell\cdot(\bm\theta+\bm n_s'\bm\alpha))_{1\leq s\leq D,   \bm\ell \in \Gamma_ K}.
%\bm W= (\sin  t_\ell(\theta+\bm n_i\cdot\bm \alpha))_{1\leq i\leq D, 1\leq \ell \leq K}\in\R^{DK}.
\end{align*}
Denote $\bm \xi=(\bm \beta, \bm q)$ for  $\bm \beta, \bm q\in [0,1]^d$.   
%satisfying 
%\begin{align}
%\|\bm n\cdot \bm \beta\|_{\T}=\inf_{k\in\Z}|\bm n\cdot \bm \beta-k|\geq |\bm n|^{-2d}\ {\rm for} \  {\rm all}\ \bm n\in \Z^d\setminus\{0\}
%\end{align}
Write also 
\begin{align*}
d _{\bm \xi} =\sum_{i=1}^d(\beta_i\frac{\partial}{\partial \alpha_i}+q_i\frac{\partial}{\partial \theta_i})
\end{align*}
and accordingly $d_{\bm \xi}^\ell$ for $\ell \geq 1.$

A key observation is that for all $j\geq 1$ and $\bm \ell \in \Gamma_K$, 
\begin{align*}
&\ \ \ d _{\bm \xi}^{2j} \cos2\pi\bm  \ell\cdot (\bm\theta+\bm n_s'\bm \alpha)\\
&=(-(2\pi(\bm\ell\bm n_s')\cdot \bm\beta+2\pi\bm q\cdot \bm\ell)^2)^{j} \cos2\pi\bm \ell\cdot (\bm\theta+\bm n_s' \bm \alpha).
%\nabla _{\bm \xi}^{2j} \sin  t_\ell (\theta+\bm n_i\cdot\bm \alpha)=(- t_\ell^2(\bm n_i\cdot \bm\beta+1)^2)^{j} \sin  t_\ell (\theta+\bm n_i\cdot\bm \alpha).
\end{align*}
This  motivates us to consider the Wronskian 
\begin{align*}
W=[ d _{\bm \xi}^{2j}  {V}_{(s,\bm \ell)} ]_{s, \bm \ell}^{1\leq j\leq R}\ {\rm with}\  R= D\cdot (\#\Gamma_K),
\end{align*}
which is  a $R\times R$ real matrix.  Direct computations show
\begin{align}
\nonumber&\ \ \ |\det W|\\
\label{wrons}&=\prod_{s,\bm\ell}| \cos2\pi {\bm\ell}\cdot(\bm\theta+\bm n_s' \bm \alpha)| \times \prod_{s, \bm\ell}(2\pi)^2((\bm\ell \bm n_s')\cdot\bm \beta+\bm q\cdot\bm\ell)^2\times |\det W_1|,
\end{align}
where $W_1$ is a Vandermonde matrix  with 
\begin{align}\label{w1}
  |\det W_1|&= \sqrt{\prod_{ (s,\bm\ell)\neq (s',\bm\ell')}(2\pi)^2|((\bm\ell \bm n_s')\cdot\bm \beta+\bm q\cdot\bm\ell)^2-((\bm\ell' \bm n_{s'}')\cdot\bm \beta+\bm q\cdot\bm\ell')^2|}.
 %&\ \ \times \prod_{ (i,\ell)\neq (i',\ell')}|t_\ell^2(\bm n_i\cdot \bm\beta+1)^2-t_{\ell'}^2(\bm n_{i'}\cdot \bm\beta+1)^2|\\
% &\ \ \times \prod_{ (i,\ell), (i',\ell')}|r_\ell^2(\bm n_i\cdot \bm\beta+1)^2-t_{\ell'}^2(\bm n_{i'}\cdot \bm\beta+1)^2|.
\end{align}

Then the proof can be divided into three steps: \\

{\bf Step 1.}  {\bf Prior  conditions on $(\bm \alpha , \bm\theta)$}\\

We aim to get a lower bound for   \eqref{wrons}, which requires the prior restrictions on $\bm \alpha, \bm\theta.$  %Note that  
Let $\|x\|_{\T/2}={\rm dist} (x, \Z/2)$. From
\begin{align*}
%&\frac{2}{\pi}\|x\|_{\T/2}\leq |\sin x|\leq {\rm dist}(x,\pi\Z):=\|x\|_{\T/2},
&4\|x-\frac14\|_{\T/2}\leq |\cos 2\pi x|\leq 2\pi \|x-\frac14\|_{\T/2},
\end{align*}
we can define for $\delta>0$ the resonant set (of $\bm \alpha, \bm\theta$) 
\begin{align}
\nonumber&\ \ \ S_\delta\\
\label{sdelta}&=\bigcup_{1\leq s \leq D,  \bm\ell\in\Gamma_K}\left\{(\bm \alpha, \bm\theta)\in[0,1]^{2d}:\  \| \bm \ell \cdot(\bm \theta+\bm n_s' \bm \alpha)-\frac14\|_{\T/2}\leq \delta \right\}.
%S_{\delta,2}&= \bigcup_{1\leq\ell\leq  K; 1\leq i\leq D}\{(\bm \alpha,\theta)\in[0,1]^{d+1}:\  \|t_\ell (\theta+\bm n_i\cdot\bm \alpha)-\frac\pi 2\|_{\T/2}\leq \delta \}.
\end{align}
Denote $L=\max\limits_{1\leq s\leq D} K |\bm n_s'|$.  We have 
\begin{lem}\label{coverlem}
Let  $0<\delta<1$.  Then $S_\delta$ can be covered by $C L^{2d-1}\delta^{-2d+1}$  intervals of side length  $\leq  \delta,$ where $C=C(d,R)>0.$
\end{lem}

\begin{proof}
 It suffices to cover for each $\bm n_s', \bm \ell$ the set 
$$ S_{\bm n_s', \bm \ell}=\{(\bm \alpha, \bm\theta)\in[0,1]^{2d}:\  \|\bm\ell \cdot (\bm \theta+\bm n_s'\bm \alpha)-\frac14\|_{\T/2}\leq \delta \}$$
with $C(d) L^{2d-1} \delta^{-2d+1}$  intervals of side length $\leq  \delta$ since the total number of  all $\bm n_s', \bm\ell$ is at most $R.$ Since $|\bm \ell| >0,$
 we may assume $\ell_1\neq 0$ without loss of generality.  For any $\bm x=(x_1,\cdots, x_d),$  write $\bm x=(x_1, \bm  x^c_1)$. 
So we can divide $[0,1]^d\times [0,1]^{d-1}\ni (\bm \alpha, \bm\theta_1^c)$ into $C_1=[100L\delta^{-1}]^{2d-1}$  cubes  $\bm C_i$ ($1\leq i\leq C_1$)  with disjoint exteriors  of  the same side length $\leq  \frac{\delta }{ L}$. Denote by  $ (\bm\alpha_i, \bm\theta^c(i))\in [0,1]^{d}\times [0,1]^{d-1}$ ($1\leq i\leq C_1$)  the centers of those cubes.    We further define 
\begin{align*}
&\ \ \  I_{\bm n_s', \bm \ell}\\
 &=\bigcup_{1\leq i\leq C_1}\{(\bm \alpha, \bm \theta_1^c,  \theta_1)\in \bm C_i\times [0,1]:\  \| \ell_1\theta_1+ \bm\ell_1^c\cdot\bm  \theta^c(i)+(\bm n_s'\bm\ell)\cdot \bm\alpha_i-\frac14\|_{\T/2}\leq 2\delta \}\\
 &:= \bigcup_{1\leq i\leq C_1} J_i.
\end{align*}
We first claim that  $S_{\bm n_s', \bm\ell} \subset I_{\bm n_s', \bm \ell}.$  Let $(\bm \alpha, \bm\theta)\in S_{\bm n_s', \bm\ell}$.  Then by the definition (of $S_{\bm n_s', \bm\ell}$), we have 
\begin{align*}
 \|\ell_1\theta_1+\bm\ell_1^c \cdot\bm\theta_1^c+(\bm\ell\bm n_s')\cdot\bm  \alpha-\frac14\|_{\T/2}\leq  \delta .
\end{align*}
Since $\bm C_i$ ($1\leq i\leq C_1$) covers $[0,1]^{2d-1}$,  there is $1 \leq i_0\leq C_1$ so that $(\bm \alpha, \bm\theta_1^c)\in \bm C_{i_0}$, namely, $|\bm \alpha-\bm\alpha_{i_0}|+|\bm\theta_1^c-\bm \theta^c(i_0)|\leq \frac{\delta }{ 2L}.$ This implies that 
\begin{align*}
&\ \ \ \| \ell _1\theta_1+ \bm\ell_1^c\cdot \bm \theta^c(i_0)+ (\bm n_s' \bm\ell) \cdot\bm  \alpha_{i_0}-\frac14\|_{\T/2}\\
&\leq  \|\bm \ell \cdot (\bm \theta+\bm n_s' \bm  \alpha)-\frac14\|_{\T/2}+|\bm \ell_1^c\cdot(\bm\theta^c(i_0)-\bm\theta_1^c)|\\
&\ \ +|(\bm \ell \bm n_s')\cdot(\bm \alpha-\bm \alpha_{i_0})|\\
&\leq 2\delta. 
\end{align*}
The claim is proved.  So  it remains to cover each $J_i$ with  intervals of side length $\leq \delta$.   This can be accomplished by the fact that  the set $\{\theta_1\in [0,1]:\ \|\ell_1\theta_1+x\|_{\T/2}\leq 2\delta\}$ ($x\in\R$)
consists of $r\sim |\ell_1|$  one dimensional intervals of length $2\delta/|\ell_1|$.    %In fact, since $\theta \in[0,1]$,  one has $\ell\theta\in [0,\ell]$. This implies that $[0, \ell]$ can be divide into $\sim \frac{ 2 \ell}{\pi}$  equal intervals of length $\frac{\pi}{2}$. On each of such intervals,  $\{\theta:\ \|\ell\theta\|_{\T/2}\leq 2\delta\}$ is given by an interval of length $\frac{\delta}{\ell}$.   So the desired intervals decomposition is just a translation of these intervals 

We have finished the proof. 
\end{proof}

Given $0<\delta<1/10,$ we divide  $[0,1]$ into $N=[\delta^{-2}]+1$ intervals of equal length $\frac1N\sim \delta^2$. This then induces a partition of $[0,1]^{2d}$   into cubes of equal side length $1/N$. We denote by $\Lambda_{1/N}(\bm x_i) $ ($1\leq i\leq N^{2d}$) all these cubes.  We have
\begin{lem}
Let $0<\delta<1/10$ and let $S_\delta$ be defined by \eqref{sdelta}.    Then we have  for $N=[\delta^{-2}]+1$, 
\begin{align}\label{countb}
\#\{1\leq i\leq N^{2d}:\  \Lambda_{1/N}(\bm x_i)\cap S_\delta\neq \emptyset\}\leq C(d,R)L^{2d-1} \delta^{-4d+1}.
\end{align}
%where $C>0$  depends only on $D,K,d$ and $\max_{1\leq i\leq D}|\bm n_i|.$
In particular, we have 
\begin{align}\label{measb}
{\rm meas}\left(\bigcup_{1\leq i\leq N^{2d}:\ \Lambda_{1/N}(\bm x_i)\cap S_\delta\neq \emptyset}\Lambda_{1/N}(\bm x_i)\right)\leq  C(d,R)L^{2d-1}\delta.
\end{align}
%where $C>0$   also depends only on $D,K,d$ and $\max_{1\leq i\leq D}|\bm n_i|.$
\end{lem}

\begin{proof}
Let $ \bm C_i$ be given as in the proof of  Lemma \ref{coverlem}. Then $S_\delta\subset \cup_{1\leq i\leq C(d,R)L^{2d-1}\delta^{-2d+1}}\bm C_i$. So it suffices to  bound 
\begin{align*}
L_1=\#\{1\leq  i\leq N^{2d}:\  (\cup_{1\leq j\leq CL^{2d-1}\delta^{-2d+1}}\bm C_j)\cap \Lambda_{1/N}(\bm x_i)\neq \emptyset\}.
\end{align*}
Note that if $\Lambda_{1/N}(\bm x_i)\cap \bm C_j\neq \emptyset$, then 
$$ \Lambda_{1/N}(\bm x_i) \subset  \bm C_j',$$
where $ \bm C_j'$ is also an interval  containing  $\bm C_j$ and satisfying  $${\rm dist}(\partial \bm C_j, \partial \bm C_j' )\leq 10/N,$$
with $\partial$ the boundary.  Regarding the  disjointness (of the exterior) of $\Lambda_{1/N}(\bm x_i)$ for different $i$, we must have 
\begin{align*}
L_1\leq CL^{2d-1}\delta^{-2d+1} (\delta+10/N)^{2d}\delta ^{-4d}\leq CL^{2d-1}\delta ^{-4d+1}.
\end{align*}
This implies the bound \eqref{countb}.  Then the measure bound follows directly from \eqref{measb}. 
\end{proof}

From this lemma, we get immediately that 

\begin{lem}\label{coslem}
Let  $0<\delta<1/10$ and define $N=[\delta^{-2}]+1$. Then there is a collection of  cubes  $\Lambda_{1/N}(\bm x_i) \subset [0,1]^{d+d}$ ($1\leq i\leq N_1$) with disjoint exteriors  so that 
\begin{align*}
{\rm meas}\left([0,1]^{2d}\setminus \cup_{1\leq i\leq N_1}\Lambda_{1/N}(\bm x_i)\right)\leq CL^{2d-1}\delta,
\end{align*}
and  if $(\bm\alpha,\bm\theta)\in \cup_{1\leq i\leq N_1}\Lambda_{1/N}(\bm x_i),$ then
\begin{align*}
\prod_{1\leq s \leq D, \bm\ell\in\Gamma_K}|\cos2\pi\bm \ell \cdot (\bm\theta+\bm n_s'\bm\alpha)|\geq c\delta^{R}, \ c=c(d,R)>0.
\end{align*}
\end{lem}

{\bf Step 2.} {\bf Conditions on $(\bm \beta, \bm q)$}\\

Next we impose conditions on $(\bm \beta, \bm q)$.  \\

\begin{lem}\label{betalem}
%Assume $\bm \ell \bm n_\ell\neq \bm \ell \bm n_{\ell'} \ {\rm for}\ \forall\ \ell\neq \ell',\ \bm\ell\in\Gamma_K $ (this holds true since \eqref{ncond}). 
Let $\epsilon\in (0,1)$ and $\bm B\subset[0,1]^d\times  [0,1]^d $ be the set of $\bm \xi=(\bm \beta, \bm q)$  satisfying
\begin{align*}
&\min_{s, \bm \ell} ((\bm\ell\bm n_s') \cdot\bm \beta+\bm q\cdot \bm \ell)^2\geq  \epsilon>0,\\
%&\ \ \ \min_{(i,\ell)\neq (i',\ell_1)}|\ell^2(\bm n_i\cdot \bm\beta+1)^2-{\ell_1}^2(\bm n_{i'}\cdot \bm\beta+1)^2|\\
&\min_{(s, \bm\ell)\neq (s', \bm\ell')}|(\bm \ell \bm n_s'-\bm \ell'\bm n_{s'}')\cdot\bm \beta+\bm q\cdot(\bm\ell -\bm\ell')|\geq \epsilon>0,\\
&\min_{(s,  \bm\ell)\neq (s', \bm\ell')}|(\bm \ell \bm n_s'+\bm \ell'\bm n_{s'}')\cdot\bm \beta+\bm q\cdot(\bm\ell +\bm\ell')|\geq \epsilon>0.
\end{align*}
Then 
\begin{align*}
{\rm meas}\left( [0,1]^{2d}\setminus\bm B\right)\leq 3R^2\sqrt{\epsilon}.
\end{align*}
Moreover, for $\bm\xi\in\bm B,$ we have 
\begin{align*}
&\ \ \  \prod_{s, \bm\ell}((\bm \ell\bm n_s')\cdot\bm \beta+\bm q\cdot\bm \ell)^2\times |\det W_1| \\
 &=\prod_{s, \bm\ell}((\bm \ell\bm n_s')\cdot\bm \beta+\bm q\cdot\bm \ell)^2\\
 &\ \ \times \prod_{ (s, \bm\ell)\neq (s', \bm \ell')}|((\bm\ell \bm n_s')\cdot\bm \beta+\bm q\cdot\bm\ell)^2-((\bm\ell' \bm n_{s'}')\cdot\bm \beta+\bm q\cdot\bm\ell')^2|\\
& \geq \epsilon^{R^3}.
\end{align*}

\end{lem}

\begin{proof}
 The proof  is based on the Fubini's theorem,   \eqref{Gammk} and the assumption  on $\bm n_s'$ ($1\leq s\leq D$). %\eqref{ncond}.  %that  for all $\bm\ell \in \Gamma_K$, 
% \begin{align*}
%\bm \ell \bm n_\ell\neq \bm \ell \bm n_{\ell'} \ {\rm for}\ \forall\ \ell\neq \ell'. 
%\end{align*}
Actually, in all cases we have 
\begin{align*}
(\bm \ell \bm n_s', \bm\ell)\neq \bm 0,\\
(\bm \ell \bm n_s'-\bm \ell'\bm n_{s'}',  \bm\ell -\bm\ell')\neq  \bm 0,\\
(\bm \ell \bm n_s+\bm \ell'\bm n_{s'}',  \bm\ell +\bm\ell')\neq  \bm 0,
\end{align*}
which  combined with $s\leq D$ and $\bm \ell\in \Gamma_K$ leads to the excision of $(\bm \beta, \bm q)$ in a set of  measure $O(\sqrt{\epsilon})$.
\end{proof}

{\bf Step 3.} {\bf Usage of the Lemma \ref{BGGlem}}\\

Now combining Lemma \ref{coslem} and Lemma \ref{betalem} (we set $\epsilon=10^{-2}R^{-4}$) leads to 
\begin{lem}\label{wlem}
Let $0<\delta<1/10$ and define $N=[\delta^{-2}]+1$. Then there are  a collection of  cubes  $\Lambda_{1/N}(\bm x_i) \subset [0,1]^{2d}$ ($1\leq i\leq N_1$ for some $N_1>1$) with disjoint exteriors, and a set $\bm B\subset [0,1]^{2d}$ so that 
\begin{align*}
{\rm meas}\left([0,1]^{2d}\setminus \cup_{1\leq i\leq N_1}\Lambda_{1/N}(\bm x_i)\right)&\leq CL^{2d-1}\delta,\\
{\rm meas}( [0,1]^{2d}\setminus\bm B)&<1/10.
\end{align*}
Moreover, if $(\bm\alpha,\bm\theta)\in \cup_{1\leq i\leq N_1}\Lambda_{1/N}(\bm x_i) $ and $\bm\xi\in\bm B,$ then
\begin{align*}
|\det W|\geq c\delta^{R},\ c=c(d,R)>0.
%10^{-2} (2/\pi)^{DK} (DK)^{-4}\delta^{DK}.
\end{align*}
\end{lem}

We are ready to conclude the proof of Theorem~\ref{nlsthm}.  Fix $\bm \xi\in\bm B$ and  a $1/N\sim\delta^2$-cube $\Lambda_i=\Lambda_{1/N}(\bm x_i)\subset[0,1]^{2d}$ in Lemma \ref{wlem}. We will deal with  $(\bm\alpha, \bm\theta)\in \Lambda_i.$  Denote for $\bm k'\in\Z^D\setminus\{\bm 0\}$, 
\begin{align*}
f(\bm \alpha, \bm\theta)=\bm k' \cdot\bm \omega'=(\bm k'\otimes \bm v)\cdot \bm V.
\end{align*}
Direct computations show that for $\bm \xi=(\bm\beta, \bm q),$
\begin{align*}
&\ \ \ \sup_{(\bm\alpha, \bm\theta)\in[0,1]^{2d}}\sup_{1\leq j\leq R} |\partial^{2j}_{\bm \xi}\bm V|_1\\
&=\sup_{(\bm\alpha, \bm\theta)\in[0,1]^{2d}}\sup_{1\leq j\leq R}\sum_{1\leq s\leq D,  \bm \ell \in \Gamma_K}|(2\pi(\bm\ell\bm n_s')\cdot\bm \beta+2\pi\bm q\cdot\bm \ell)^{2j}\cos2\pi\bm\ell\cdot(\bm \theta+\bm n_s' \bm\alpha)|\\
&\leq CR (\sup_{s}|\bm n_s'|+dK)^{2R}:=M.
\end{align*}
We have by Lemma \ref{wlem} and Lemma \ref{BGGlem} that 
\begin{align}\label{tsvcon}
\inf_{(\bm\alpha,\bm\theta)\in\Lambda_i}\max_{1\leq j\leq R} |d^{2j}_{\bm \xi} f(\bm \alpha, \bm \theta)|\geq c\delta^{R}M^{1-R}|\bm k'|\cdot |\bm v|,
\end{align}
where $c=c(d,R, V)>0.$
%\begin{align}
% c=10^{-2}  (DK)^{-5-DK} (\sup_{i}|\bm n_i|+1)^{2DK(1-DK)}(2/\pi)^{DK} >0.
%\end{align}

%\begin{rem}
%In contrast, if we consider $\bm \omega$ as free parameter varying in $[0,1]^D$,  we have 
%\begin{align} 
%\nonumber&{\rm meas}\left(\left\{\bm \omega\in[0,1]^{D}:\ |\bm k\cdot\bm\omega|\geq \frac{\gamma}{|\bm k |^{\tau}}\ {\rm for}\ \forall\ \bm k\in \Z^D\setminus\{0\}\right\}\right)\geq 1-O(\gamma).
%\end{align}
%So the absence of independence of coordinates of  $\bm \omega_{\rm NLS}$ leads to the corresponding measure bound  (from $1-O(\gamma)$ to)  $1-O(\gamma^{\frac{1}{10dR^2}}).$ 
%\end{rem}

%\begin{proof}
It suffices to apply Lemma \ref{lajlem} by letting $k=2R$. The transversality condition   has been verified by \eqref{tsvcon}.  For $f(\bm \alpha, \bm\theta)=\bm k'\cdot\bm \omega'$, we have 
\begin{align*}
\|f\|_{2R+1}=\sup_{|\bm\gamma|\leq 2R+1, (\bm \alpha, \bm \theta)\in \Lambda_i}| \partial^{\bm\gamma}_{\bm \xi} f(\bm \alpha,\bm\theta)|\leq C(d,R,V)\cdot (\sup_{1\leq s\leq D}|\bm n_s'|)^{2R+1}  |\bm k'|.
\end{align*}
So using Lemma \ref{lajlem} by setting $\varepsilon =\frac{\eta}{|\bm k'|^{4R^2}}$ leads to 
\begin{align*}
&\ \ \  {\rm meas}\left(\left\{ (\bm \alpha, \bm\theta)\in\Lambda_i:\ |\bm k' \cdot \bm\omega'|\leq \frac{\eta}{|\bm k'|^{4R^2}} \right\}\right)\\
&\leq C(d, R, V)(\max_{1\leq s\leq D}|\bm n_s'|)^{3Rd}M^{2Rd}\delta^{-R (d+ 1)}\eta^{\frac{1}{2R}} |\bm k'|^{-{2R}}\\
&\leq C(d,R,V)(\max_{1\leq s\leq D}|\bm n_s'|)^{5dR^2}\frac{\delta^{3dR}}{|\bm k'|^{2D}}\  ({\rm choose}\  \delta= \frac{1}{10}\eta^{\frac{1}{10dR^2}}).
\end{align*}
Together with Lemma \ref{wlem} this yields
\begin{align*}
&\ \ \ {\rm meas}([0,1]^{2d}\setminus {\rm DC}_{\rm nls}(\eta))\\
%&\leq {\rm meas}()
&\leq C(d,R,V)L^{2d-1}\delta+C(d,R,V) (\max_{1\leq s\leq D}|\bm n_s'|)^{5dR^2}\delta^{-4d}\sum_{\bm k'\in\Z^D\setminus\{\bm 0\}} \frac{\delta^{3dR}}{|\bm k'|^{2D}}\\
&\leq C(d,K,D, V) (\max_{1\leq s\leq D}|\bm n_s'|)^{5dR^2} \delta\ ({\rm since}\ R\geq 2\ {\rm and}\ L=K\max_{1\leq s\leq D}|\bm n_s'|)\\ 
&\leq C(d,K,D, V) (\max_{1\leq s\leq D}|\bm n_s'|)^{5dR^2} \eta^{\frac1{10dR^2}}.
\end{align*}

This proves Theorem~\ref{nlsthm}.
%\end{proof}
\end{proof}

	\section{Some important lemmas}
	Let $M_2(\C)$ denote the  Banach algebra of all $(2\times 2)$-complex matrices equipped with the standard operator norm.   Since  the index $\xi\in\{+,-\}$, we can identify  the matrix  from  $\Z^{b+d}_{\rm pm}$ to $\C$ with  that  from $\Z^{b+d}$ to $M_2(\C)$.  So, in the following analysis, we only focus on  matrices with their entries  indexes    $\bm x=(\bm k, \bm n)\in\Z^{b+d}$.

	We first introduce an important  perturbation lemma.
	{\begin{lem}[cf. Lemma 4.2 of \cite{LW22} and Lemma A.1 of \cite{Shi22}]\label{lwlem}
	Let $X\subset\Z^{b+d}$ with $|X|=\# X$. Assume that $A$ and $B$ are two matrices with entries $A(\bm x,\bm x')$ and $B(\bm x,\bm x')$, where $\bm x,\bm x'\in X$. Assume that for some positive constants $\epsilon_1,\epsilon_2\in(0,1)$, $C_2\ge0$, $c>0$ and $0\le M<{\rm diam} \ X$: (1). for all $\bm x$ and $\bm x'$, $|B(\bm x,\bm x')|\le \epsilon_2 (1+|\bm x-\bm x'|)^{C_2}e^{-c|\bm x-\bm x'|}$; (2). for $|\bm x-\bm x'|>M$, $|A^{-1}(\bm x,\bm x')|\le e^{-c|\bm x-\bm x'|}$. Suppose that $\|A^{-1}\|\leq \epsilon_1^{-1}$ and 
	\begin{align*}
	%&\\
	&  |X|^2e^{2cM}(1+{\rm diam}\ X)^{C_2}\cdot \epsilon_2\epsilon_1^{-1}\sum_{\bm x\in\Z^{b+d}}e^{-c|\bm x|}\leq \frac12.
	\end{align*}
	Then
	\begin{align*}
		\|(A+B)^{-1}\|&\le 2\epsilon_1^{-1},\\
		|(A+B)^{-1}(\bm x,\bm x')-A^{-1}(\bm x,\bm x')|&\le \epsilon_1^{-1}e^{-c|\bm x-\bm x'|}.
	\end{align*}
\end{lem}}
\begin{proof}
		The proof is based on the Neumann series argument  similar to  \cite{Shi22, LW22}. For completeness, we give a proof. 
	
	First, we note  that $|A^{-1}(\bm x,\bm x')|\le \epsilon_1^{-1}e^{cM}e^{-c|\bm x-\bm x'|}$ for all $\bm x,\bm x'\in X$. %If $|\bm x-\bm x'|\le M$, we have
	%\begin{align*}
	%|A^{-1}(\bm x,\bm x')|\le \|A^{-1}\|= \|A^{-1}\|e^{c|\bm x-\bm x'|}e^{-c|\bm x-\bm x'|}\le \epsilon_1^{-1}e^{cM}e^{-c|\bm x-\bm x'|}.
	%	\end{align*}
	%If $|\bm x-\bm x'|> M$, combining with $\epsilon_1\in(0,1)$, we get
	%	\begin{align*}
	%|A^{-1}(\bm x,\bm x')|\le e^{-c|\bm x-\bm x'|}\le \epsilon_1^{-1}e^{cM}e^{-c|\bm x-\bm x'|}.
	%\end{align*}
	%We finish proof of the claim. 
	For $\bm x^0=\bm x, \bm x^s=\bm x'\in X$ and $s\ge1$, we can obtain by direct computations  that 
	\begin{align*}
		(BA^{-1})^{s}(\bm x,\bm x')=\sum_{\bm x^1,\cdots,\bm x^{s-1},\bm y^1,\cdots,\bm y^{s}\in X}\prod_{t=1}^{s}B(\bm x^{t-1},\bm y^{t})A^{-1}(\bm y^{t},\bm x^{t}).
	\end{align*}
	Thus for $s\ge1$, one has
	\begin{align*}
		&\ \ \  |(BA^{-1})^{s}(\bm x,\bm x')|\\
		%&\le\sum_{\bm x^1,\cdots,\bm x^{s-1},\bm y^1,\cdots,\bm y^{s}\in X}e^{scM}(1+{\rm diam}\ X)^{sC_2}\cdot\epsilon_2^s\epsilon_1^{-s}\prod_{t=1}^{s}(e^{-c|\bm x^{t-1}-\bm y^{t}|}e^{-c|\bm y^{t}-\bm x^{t}|})\\
		&\le \sum_{\bm x^1,\cdots,\bm x^{s-1},\bm y^1,\cdots,\bm y^{s}\in X}e^{scM}(1+{\rm diam}\ X)^{sC_2}\cdot\epsilon_2^s\epsilon_1^{-s}e^{-c|\bm x-\bm x'|}\\
		&\le |X|^{2s-1}e^{scM}(1+{\rm diam}\ X)^{sC_2}\cdot\epsilon_2^{s}\epsilon_1^{-s}e^{-c|\bm x-\bm x'|},
	\end{align*}
	and  for $\bm x\neq \bm x',$
	\begin{align*}
		&\ \ \  \left|\left(\sum_{s\ge1}(-BA^{-1})^s\right)(\bm x,\bm x')\right|\\
		%&\le \sum_{s\ge1}|(BA^{-1})^{s}(\bm x,\bm x')|\\
		&\le  \sum_{s\ge1}|X|^{2s-1}e^{scM}(1+{\rm diam}\ X)^{sC_2}\cdot\epsilon_2^{s}\epsilon_1^{-s}e^{-c|\bm x-\bm x'|}\\
		&=|X|e^{cM}(1+{\rm diam}\ X)^{C_2}\cdot\epsilon_2\epsilon_1^{-1}\cdot\frac{1}{1-|X|^2e^{cM}(1+{\rm diam}\ X)^{C_2}\cdot\epsilon_2\epsilon_1^{-1}}e^{-c|\bm x-\bm x'|}\\
		&\le 2|X|e^{cM}(1+{\rm diam}\ X)^{C_2}\cdot\epsilon_2\epsilon_1^{-1}e^{-c|\bm x-\bm x'|}.
	\end{align*}
	Hence, from the Schur's test, we have 
	\begin{align*}
		\|BA^{-1}\|\le |X|e^{cM}(1+{\rm diam}\ X)^{C_2}\cdot\epsilon_2\epsilon_1^{-1}\cdot\sum_{\bm x\in\Z^{b+d}}e^{-c|\bm x|}\le \frac{1}{2}.
	\end{align*}
Using the Neumann series argument, we get
	\begin{align*}
		(A+B)^{-1}=A^{-1}\sum_{s\ge0}(-BA^{-1})^{s}.
	\end{align*}
	Thus %one has
	\begin{align*}
		\|(A+B)^{-1}\|\le \|A^{-1}\|\frac{1}{1-\|BA^{-1}\|}\le 2\epsilon_1^{-1},
	\end{align*}
	and for $\forall\ \bm x,\bm x'\in X$,  
	\begin{align*}
		&\ \ \ |(A+B)^{-1}(\bm x,\bm x')-A^{-1}(\bm x,\bm x')|\\
		&\le \sum_{\bm y\in X}|A^{-1}(\bm x,\bm y)|\cdot\left|\left(\sum_{s\ge1}(-BA^{-1})^s\right)(\bm y,\bm x')\right|\\
		&\le \sum_{\bm y\in X}\epsilon_1^{-1}e^{cM}e^{-c|\bm x-\bm y|}\cdot2|X|e^{cM}(1+{\rm diam}\ X)^{C_2}\cdot\epsilon_2\epsilon_1^{-1}e^{-c|\bm y-\bm x'|}\\
		&\le 2|X|^2e^{2cM}(1+{\rm diam}\ X)^{C_2}\cdot\epsilon_2\epsilon_1^{-2}e^{-c|\bm x-\bm x'|}\\
		&\le \epsilon_1^{-1}e^{-c|\bm x-\bm x'|}.
	\end{align*}
	We complete the  proof.
\end{proof}

	In the following, let $A$ be a matrix on $\Z^{b+d}$ satisfying
	\[|A(\bm x, \bm x')|\leq C_2 (1+|\bm x-\bm x'|)^{C_2}e^{-b_1|\bm x-\bm x'|},\ C_2,b_1>0.\]
	
We introduce two types of coupling lemmas by  iterating resolvent identities, which are repeatedly used to derive off-diagonal exponential  decay estimates for Green's functions. 

\begin{rem}
	In (linear) spectral problems (cf. \cite{Bou07, JLS20}),  one  typically  deals  with operators such as
	\[|A(\bm x,\bm x')|\leq C_2 e^{-b_1|\bm x-\bm x'|}\]
	without the polynomial factor $(|\bm x-\bm x'|+1)^{C_2}$. However, the extra polynomial factor does not significantly affect the proofs in \cite{JLS20} (cf. Lemma 3.2 and Theorem 3.3) and \cite{Liu22} (cf. Theorem 2.1), as one can replace, at each step, the estimate
	\[\sum_{\bm x_1\in W \atop \bm x_2\in \Lambda\setminus W}e^{-b_1|\bm x-\bm x'|} |G_{\Lambda}(\bm x_2,\bm x')|\leq (2N+1)^{2d}\sup_{\bm x_2\in \Lambda\setminus W}|G_{\Lambda}(\bm x_2,\bm x')|\]
	with 
	\[\sum_{\bm x_1\in W \atop \bm x_2\in \Lambda\setminus W}(|\bm x-\bm x'|+1)^{C_2}e^{-b_1|\bm x-\bm x'|} |G_{\Lambda}(\bm x_2,\bm x')|\leq (2N+1)^{2d+C_2}\sup_{\bm x_2\in \Lambda\setminus W}|G_{\Lambda}(\bm x_2,\bm x')|.\]
	\end{rem}

We say that an elementary region $\Lambda\in \mathcal{ER}(L)$ is in the class $\mathbb G$ (good) if  ($G_\Lambda=(R_{\Lambda}AR_{\Lambda})^{-1}$)
\begin{equation}\label{range1}
	|G_\Lambda(\bm x,\bm x')|\leq e^{-b_2|\bm x-\bm x'|}\  {\rm for} \ |\bm x-\bm x'|\geq L^{\frac89},
\end{equation}
where $\frac{1}{2}b_1< b_2\leq b_1$.

We have a revised version of the coupling lemma of   \cite{Liu22} (cf. also \cite{BGS02}). 
\begin{lem}[cf. Theorem 2.1, \cite{Liu22}]\label{Liulem1}
	 Let $\tilde{\Lambda}_0\in\mathcal {ER}(N)$ be an elementary region with the property that  for all $\Lambda\subset\tilde{\Lambda}_0,\Lambda\in \mathcal R_{L}^{\sqrt{N}}$ with $\sqrt N \leq L\leq N$, the 
	Green's function $(R_{\Lambda}A R_{\Lambda})^{-1}$ satisfies 
	\[\| (R_{\Lambda}AR_{\Lambda})^{-1}\| \leq e^{L^{\frac{3}{4}}}.\]
	Assume that for any family $\mathcal F$ of pairwise disjoint elementary regions of size  $M=[ N^{\xi}]$ contained in $\tilde{\Lambda}_0,$  
	\[\#\{\Lambda\in \mathcal F:\ \Lambda \ {\rm is \ not \ in \ class} \ \mathbb G\}\leq  {N^{\frac14}}.\]
	Then for large $N$ (depending on $C_2, b_1$),  we have 
	\begin{equation}\label{range2}
			|(R_{\tilde{\Lambda}_0}AR_{\tilde{\Lambda}_0})^{-1}(\bm x,\bm x')|\leq e^{-(b_2-N^{-\vartheta})|\bm x-\bm x'|}\ {\rm for} \ |\bm x-\bm x'|\geq N^{\frac89},
	\end{equation}
    where $\vartheta=\vartheta\in(0,1)$ is an absolute constant.
\end{lem}
\begin{rem}\label{rmkLiu}
	There is an important difference between the original form of  Lemma \ref{Liulem1} in \cite{Liu22} and the present one. In \cite{Liu22},  \eqref{range1}--\eqref{range2} hold under the condition  of  $|\bm x-\bm x'|\geq\frac{N}{10}$,  $0<b_2\leq \frac{4}{5}b_1$. However, the relation $0<b_2\leq \frac{4}{5}b_1$ leads to the deterioration of decay rates in the Newton iteration when solving the $P$-equations. To address this issue, we impose a stronger restriction  $|\bm x-\bm x'|>N^{\frac89}$ (this idea was first introduced in \cite{HSSY}), which ensures that $0<b_2\leq b_1$. So Lemma \ref{Liulem1} is proved via the argument of \cite{Liu22} combined with the argument   of  \cite{HSSY} (we can take $b=\frac34, \tau=\frac12, \theta=\frac89$ in  Lemma 4.2 of \cite{HSSY}). We omit the details here. \end{rem}

We also need the following  lemmas  of \cite{JLS20} (with again minor modifications). 

\begin{lem}[cf. Lemma 3.2, \cite{JLS20}]\label{JLSlem1}
	Let $M_0\geq (\log N)^{\frac{4}{3}},\ \frac{b_1}{2}<b_2\leq b_1$ and $M_1\leq N$. Suppose that  $\Lambda\subset\Z^{b+d}$ is connected and ${\rm diam}(\Lambda)\leq 2N+1$. Suppose that for any $\bm y\in\Lambda,$ there exists some 
	$W=W(\bm y)\in \mathcal E_M$ with $M_0\leq M\leq M_1$ such that $\bm y\in W\subset \Lambda,{\rm dist}(\bm y,\Lambda\setminus W)\geq \frac{M}{2}$ and 
	\[\| (R_{W} A R_{W})^{-1} \|\leq e^{M^{\frac{3}{4}}},\]
	\[	|(R_{W}AR_{W})^{-1}(\bm x,\bm x')|\leq e^{-b_2|\bm x-\bm x'|}\ {\rm for} \ |\bm x-\bm x'|\geq {M}^{\frac 89}.\]
    Assume further that $N$ is large enough such that 
	\[\sup_{M_0\leq M\leq M_1} e^{M^{\frac{3}{4}}}(2M+1)^{b+d+C_2}e^{b_2\sqrt{N}}\sum_{j=0}^{\infty}(M+2j+1)^{b+d} e^{-b_2(j+M/2)}\leq \frac{1}{2}.\]
	Then we  have 
	\[\| (R_{\Lambda} A R_{\Lambda})^{-1} \|\leq 4(2M_1+1)^{b+d}  e^{M_1^{\frac{3}{4}}}.\]
\end{lem}

\begin{lem}[cf. Theorem 3.3, \cite{JLS20}]\label{JLSlem2}
	Assume $\Lambda\subset\Z^{b+d}$ is connected and ${\rm diam}\ \Lambda\leq 2N+1$. Assume  ${\rm diam}\ \Lambda_1\leq N^{\frac{1}{2(b+d)}}$.  Let $M_0\geq (\log N)^{\frac{4}{3}},\ \frac{b_1}{2}<b_2\leq b_1$.   Suppose that for any $\bm y\in\Lambda,$ there exists some 
	$W=W(\bm y)\in \mathcal {ER}(M)$ with $M_0\leq M\leq N^{\frac34}$ such that $\bm y \in W\subset \Lambda, {\rm dist}(\bm y,\Lambda\setminus\Lambda_1\setminus W)\geq \frac M2$ and 
	\begin{align*}
	\| (R_{W} A R_{W})^{-1} \|&\leq e^{M^{\frac{3}{4}}},\\
|(R_{W}AR_{W})^{-1}(\bm x,\bm x')|&\leq e^{-b_2|\bm x-\bm x'|}\ {\rm for} \ |\bm x-\bm x'|\geq {M}^{\frac 89}. 
   \end{align*}
 Suppose  that 
    $$\|R_\Lambda A R_\Lambda\|\leq e^{N^{\frac 34}}.$$
    Then 
   \begin{align*}
    |(R_{\Lambda}AR_{\Lambda})^{-1}(\bm x,\bm x')|&\leq e^{-(b_2-\frac{C(b,d)}{M_0^{\frac19}})|\bm x-\bm x'|}\ {\rm for} \ |\bm x-\bm x'|\geq {M}^{\frac 89}. 
    \end{align*}
    
\end{lem}

\begin{rem}\label{rmkJLS}
The similar issue  in Lemma \ref{Liulem1} also appears Lemma  \ref{JLSlem1} and Lemma \ref{JLSlem2}. Again, we can improve the original estimates of \cite{JLS20} from  $|\bm x-\bm x'|\geq \frac {N}{10}$  to $|\bm x-\bm x'|\geq N^{\frac89}.$
\end{rem}

The following lemma deals with projection property   of  the elementary regions. 
  {\begin{lem}\label{projlemz}
		Let $\Lambda \in \mathcal{ER}_{\Z^{b+d}}(N)$ and denote by $\Pi_b: \ \Z^{b+d} \to \Z^b$ the canonical projection map. Then for any $\bm k \in \Pi_b\Lambda$, the section
		\[
		\Lambda(\bm k) = \{ \bm n \in \Z^d : (\bm k,\bm n) \in \Lambda \}
		\]
		either belongs to $\mathcal{ER}_{\Z^d}(N)$ or is a rectangle of width at least  $N$.
	\end{lem}}
	
%	\textbf{Remark:} The geometric intuition of $\Lambda(k)$ is illustrated in Figure \ref{fig:Pl1}.
%	
%	\begin{figure}[h]
%		\centering
%		\includegraphics[width=0.6\linewidth]{Pl1}
%		\caption{Schematic illustration of the Projection Lemma.}
%		\label{fig:Pl1}
%	\end{figure}
	
	\begin{proof}
		By translation invariance, it suffices to consider the regions centered at the origin, i.e., $\Lambda \in \mathcal{ER}_{\bm 0,\Z^{b+d}}(N)$. We proceed by analyzing the two possible forms of $\Lambda$.
		
		\textbf{Case 1:} $\Lambda = \Lambda_{N,\Z^{b+d}} := [-N, N]^{b+d} \cap \Z^{b+d}$.
		
		In this case, the section at any $\bm k \in \Pi_b\Lambda$ is simply $\Lambda(\bm k) = [-N, N]^d \cap \Z^d$, which belongs to $\mathcal{ER}_{\Z^d}(N)$.
		
		\textbf{Case 2:} $\Lambda = \Lambda_{N,\Z^{b+d}}^{\bm\iota}$ for some constraint vector $\bm \iota = (\iota_i)_{1 \le i \le b+d} \in \{>, <, \emptyset\}^{b+d}$ with at least two non-empty components.
		
		By definition, $\Lambda$ can be expressed as:
		\begin{align}
			\nonumber
			\Lambda 
			&= \Lambda_{N,\Z^{b+d}} \setminus 
			\left\{ \bm x=(x_i)_{1 \le i \le b+d}:\ \ x_i \  \iota_i \ 0, \ \forall\  1\le i \le b+d \right\}
			\\ 
			& = \Lambda_{N,\Z^{b+d}} \setminus
			\left( \bigcap_{1 \le i \le b+d} \ \left\{ \bm x:\  \ x_i \  \iota_i \ 0 \right\} \right).
			\label{eq1}
		\end{align}
		Using De Morgan's laws, we can rewrite \eqref{eq1} as:
		\begin{align}\label{eq:Lam_decomp}
			\Lambda = \Lambda_{N,\Z^{b+d}} \bigcap
			\left( \bigcup_{1 \le i \le b+d} \ \left\{ \bm x :\  \ x_i \  \iota_i \ 0 \right\}^{c} \right),
		\end{align}
		where the complement sets are given by
		\begin{align*}
			\left\{ \bm x \ \big| \ x_i \  \iota_i \ 0 \right\}^{c}
			=
			\begin{cases}
				\emptyset , & \text{if } \iota_i = \emptyset, \\
				\left\{ \bm x:\  \ x_i \  \overline{\iota_i} \ 0 \right\}, & \text{otherwise},
			\end{cases}
		\end{align*}
		with the convention $\overline{>} := <$ and $\overline{<} := >$.
		
		Note that the sectioning operation $A \mapsto A(\bm k) = \{\bm n:\  (\bm k,\bm n) \in A\}$ has the following properties.  %distributes over both unions and intersections:
		Given $A,B \in \mathbb{Z}^{b+d}$:
		\begin{itemize}
			\item For $\bm k \in \prod_b (A \cup B)$, one has
			$(A \cup B)(\bm k) = A(\bm k) \cup B(\bm k)$ since
			\begin{align*}
				\{ \bm n : \ (\bm k,\bm n) \in (A \cup B) \}
				=
				\{ \bm n : \ (\bm k,\bm n) \in A \}
				\cup
				\{\bm n : \ (\bm k,\bm n) \in B \}.
			\end{align*}
			\item For $\bm k \in \prod_b (A \cap B)$, one has
			$(A \cap B)(\bm k) = A(\bm k) \cap B(\bm k)$ since
			\begin{align*}
				\{ \bm n : \ (\bm k,\bm n) \in (A \cap B) \}
				=
				\{ \bm n : \ (\bm k,\bm n) \in A \}
				\cap
				\{\bm n :\  (\bm k,\bm n) \in B \}.
			\end{align*}
		\end{itemize}
		
		Therefore, for a fixed $\bm k \in \Pi_b\Lambda$, taking the section of \eqref{eq:Lam_decomp} yields:
		\begin{align}\label{eq:L1}
			\Lambda(\bm k) 
			&= \Lambda_{N,\Z^{d}} \bigcap
			\left( \bigcup_{1 \le i \le b+d} \ 
			\left\{ \bm x:\  \ x_i \  \iota_i \ 0 \right\}^c(\bm k) \right).
		\end{align}
		We now analyze the sections of the complements for $1 \le i \le b+d$. 
		If $\iota_i = \emptyset$, the section is trivially empty. Thus it is sufficient to consider the case $\iota_i \neq \emptyset$. 
		We decouple the Cartesian product $\bm x \in \Z^{b+d}$ as $(\bm k,\bm n) \in \Z^b \times \Z^d$:
		\begin{align}\label{eq:L3}
			\left\{ \bm x \in \Z^{b+d}:\  \ x_i \  \overline{\iota_i} \ 0 \right\}(\bm k)
			=
			\begin{cases}
				\Z^d, & 1 \le i \le b \ \text{ and } \ k_i \  \overline{\iota_i} \ 0, \\
				\emptyset, & 1 \le i \le b \ \text{ and } \ k_i \  \iota_i \ 0, \\
				\left\{ \bm n \in \Z^{d}:\ \ n_{i-b} \  \overline{\iota_{i}} \ 0 \right\}, & b+1 \le i \le b+d.
			\end{cases}
		\end{align}
		
		Combining \eqref{eq:L1} and \eqref{eq:L3}, we observe that if there exists any $1 \le i \le b$ such that  $k_i \ \overline{\iota_i} \ 0$ ($\iota_i \neq \emptyset$), the union in \eqref{eq:L1} evaluates to $\Z^d$. 
		In this scenario, $\Lambda(\bm k) = \Lambda_{N,\Z^d}$, which belongs to $\mathcal{ER}_{\Z^d}(N)$.
		
		Otherwise, the first $b$ constraints do not contribute to the union, and the geometric shape is entirely determined by the remaining $d$ spatial constraints. Thus, there exists a projected constraint vector $\bm\iota' \in \{ >, <, \emptyset\}^{d}$ such that
		\begin{align*}
			\bigcup_{1 \le i \le b+d} \left\{\bm x:\ \ x_i \  \iota_i \ 0 \right\}^{c}(\bm k)
			&= \bigcup_{1 \le j \le d} \ \left\{ \bm n \in \Z^d:\ \ n_j \  \iota_j' \ 0 \right\}^{c}.
		\end{align*}
		Substituting this back into \eqref{eq:L1} and applying De Morgan's laws once more, we obtain
		\begin{align*} 
			\Lambda(\bm k)
			&= \Lambda_{N,\Z^{d}} \setminus
			\left( \bigcap_{1 \le j \le d} \ \left\{ \bm n \in \Z^d:\  \ n_j \  \iota_j' \ 0 \right\} \right).
		\end{align*}
		Let $|\bm \iota'|$ denote the number of non-empty components (i.e., $\iota_j' \neq \emptyset$) in $\bm\iota'$. We classify $\Lambda(\bm k)$ as follows:
		\begin{itemize}
			\item If $ |\bm \iota'| = 0 $, then $\Lambda(\bm k) = \Lambda_{N,\Z^d}$, which belongs to $\mathcal{ER}_{\Z^d}(N)$.
			\item If $ |\bm \iota'| = 1 $, $\Lambda(\bm k)$ is the full box $\Lambda_{N,\Z^d}$ minus a single half-space. This leaves a rectangle where $(d-1)$ sides have length $2N$ and one side has length $N$. Hence, it is a rectangle of width  at least $N$.
			\item If $ |\bm \iota'| \ge 2 $, by the definition, $\Lambda(\bm k) = \Lambda_{N,\Z^d}^{\bm\iota'}$, which implies $\Lambda(\bm k) \in \mathcal{ER}_{\Z^d}(N)$. 
		\end{itemize}
		This completes the proof for the origin centered case. %By the translation invariance of elementary regions and rectangles, the conclusion holds for any $\Lambda \in \mathcal{ER}_{\Z^{b+d}}(N)$.
	\end{proof}

\normalem
%\bibliographystyle{alpha}
%\bibliography{NLS}

\end{document}